\documentclass[journal,onecolumn,10pt,draftclsnofoot]{IEEEtran}
\usepackage[T1]{fontenc}

\ifCLASSINFOpdf
\else
\fi
\usepackage{cite}
\usepackage{amsmath,amssymb}
\usepackage{algorithm}
\usepackage{graphicx}
\usepackage{textcomp}
\usepackage{xcolor}
\usepackage{algorithm} 
\usepackage{algpseudocode} 
\usepackage{dsfont}
\usepackage{changepage}
\usepackage{setspace}
\usepackage{pdfpages}
\usepackage{mathrsfs}
\usepackage{diagbox}
\usepackage{multirow}
\usepackage{lscape}
\usepackage{rotating}
\usepackage{lipsum}
\usepackage[normalem]{ulem}
\newcommand{\tabincell}[2]{\begin{tabular}{@{}#1@{}}#2\end{tabular}}

\newtheorem{theorem}{Theorem}
\newtheorem{remark}{Remark}
\newtheorem{Corollary}{Corollary}
\newtheorem{Definition}{Definition}

\newtheorem{lemma}{Lemma}

\begin{document}
\title{Coded Distributed Computing with Pre-set  Data Placement and Output Functions Assignment}

\author{Yuhan~Wang,~\IEEEmembership{Student Member,~IEEE,}
        and~Youlong~Wu,~\IEEEmembership{Member,~IEEE}% <-this % stops a space
\thanks{This work was done when Yuhan Wang was with the School of Information Science and Technology, ShanghaiTech University. The work of Youlong Wu was in part supported by the National Nature Science Foundation of China (NSFC) under Grant 61901267. This paper was in part presented at the IEEE Global Communications Conference (Globecom), Rio de Janeiro, Brazil, Dec 2022\cite{GlobalCom22_Coded_MapReduce_Pre_set}. \emph{(Corresponding author: Youlong Wu.)}}
\thanks{Yuhan Wang is with the Department of Information Engineering, The Chinese University of Hong Kong, Hong Kong (e-mail: wy023@ie.cuhk.edu.hk).}% <-this % stops a space
\thanks{Youlong Wu is with the School of Information Science and Technology, ShanghaiTech University, 201210 Shanghai, China. (e-mail: wuyl1@shanghaitech.edu.cn).}% <-this % stops a space
}
\maketitle

\begin{abstract}
  Coded distributed computing can reduce the communication load for distributed computing systems by introducing redundant computation and creating multicasting opportunities.  However, the existing schemes require delicate data placement and output function assignment, which is not feasible when distributed nodes fetch data without the orchestration of a master node.  In this paper, we consider the general systems where the data placement and output function assignment are arbitrary but pre-set. We propose two coded computing schemes,  One-shot Coded Transmission (OSCT) and Few-shot Coded Transmission (FSCT), to reduce the communication load. Both schemes first group the nodes into clusters and divide the transmission of each cluster into multiple rounds, and then design coded transmission in each round to maximize the multicast gain.  The key difference between OSCT and FSCT is that the former uses a one-shot transmission where each encoded message can be decoded independently by the intended nodes, while the latter allows each node to jointly decode multiple received symbols to achieve potentially larger multicast gains. Furthermore, based on the lower bound proposed by Yu \emph{et al.}, we derive sufficient conditions for the optimality of OSCT and FSCT, respectively. This not only recovers the existing optimality results but also includes some cases where our schemes are optimal while others are not.
\end{abstract}


% \fi
\IEEEpeerreviewmaketitle

\section{Introduction}
\IEEEPARstart{D}{istributed} computing such as MapReduce \cite{Mapreduce} and Spark \cite{Spark}  has been widely applied in practice, as it can exploit distributed storage and computing resources to improve reliability and reduce computation latency. In the MapReduce-type framework, the computation task is decomposed into three phases: Map phase, Shuffle phase, and Reduce phase. In the Map phase, a master server splits the input files into multiple subfiles and assigns the subfiles to distributed computing nodes to generate intermediate values (IVs). Then, the IVs are exchanged among the distributed computing nodes during the Shuffle phase. In the Reduce phase, the computing nodes use the obtained IVs to compute the output functions. For large-scale computation tasks, the Map phase may generate a massive number of IVs, which in turn causes the communication bottleneck in the Shuffle phase, i.e., the communication latency could occupy a large proportion of the total execution time. For instance,  it is observed that 70\% of the overall job execution time is spent during the Shuffle phase when running a self-Join\cite{Selfjoin} application on an Amazon EC2 cluster\cite{EC2}.

Recently, many coded computing techniques have been proposed to mitigate the communication bottleneck for distributed computing systems. In particular,  Li \emph{et al.} proposed a coded distributed computing (CDC) scheme for MapReduce-type systems to reduce the communication load  \cite{CDC}. The key idea of CDC  is to introduce redundant computation in the Map phase and create multicast opportunities in the Shuffle phase.  A theoretical tradeoff between computation load and communication load was also established in \cite{CDC} which shows that increasing the computation load in the Map phase by a factor $r$ (representing the total number of computed files by all nodes normalized by the total number of files) can reduce the communication load in the Shuffle phase by the same factor $r$.  Due to its promising performance and theoretical optimality, the CDC technique has attracted wide interest and has been applied in many applications such as Artificial Intelligence (AI), the Internet of Things, and edge computing (see a survey paper in \cite{CDCSurvey}). 
For instance, the performance of CDC was evaluated on the TeraSort\cite{Tera} on Amazon EC2 clusters\cite{EC2}, exhibiting a speedup factor ranging from approximately 1.97 to 3.39 times faster than the conventional uncoded TeraSort algorithm\cite{CDC}.
In \cite{CodedMTL} the CDC-based approach has been applied to distributed multi-task learning (MTL) systems, and experiments on the MNIST dataset show that the total execution time is reduced by about 45\% compared to the traditional uncoded method when computation load $r=3$ and the number of tasks $m=4$.  %For instance,   in  \cite{CodedMTL} the authors applied the idea of CDC to distributed multi-task learning (MTL) systems, and experiments on the MNIST dataset show that the total execution time is reduced by about 45\% compared to the traditional uncoded method when computation load $r=3$ and the number of tasks $m=4$. }Many
In \cite{FLCD} it proposed a flexible and low-complexity design for the CDC scheme.

The CDC technique for various homogenous systems with equal computing or storage capabilities has been extensively investigated. For instance, the authors in \cite{AltCurve} established an alternative tradeoff between computation load and communication load under a predefined storage constraint. Also, different approaches have been proposed to further reduce the communication cost of CDC. Works in \cite{FileAll1,FileAll2,Locality,Resolve,PDA1,PDA2,PDA3} focused on the data placement of the system or the construction of the coded messages, and \cite{Compress1,SGD,Random} applied compression and randomization techniques to design the coded shuffling strategy. In particular, the authors in \cite{FileAll2} presented an appropriate resolvable design based on linear error correcting codes to avoid splitting the input file into too tiny pieces, and this resolvable design based scheme was used in \cite{Resolve} to solve the limitation of the compressed CDC scheme\cite{Compress1}. By applying the concept of the placement delivery array (PDA) to the distributed computation framework, \cite{PDA1,PDA2,PDA3} further reduced the number of subfiles and the number of output functions. By combining the compression with coding techniques, several IVs of a single computation task can be compressed into a single pre-combined value\cite{Compress1}, and the same idea can also be extended to distributed machine learning, such as Quantized SGD\cite{SGD}, which is a compression technique used to reduce communication bandwidth during the gradient updates between the computing nodes. 

Note that in practical distributing systems the computing nodes could have different storage and computational capabilities. To address the heterogeneous problem, several works have considered the heterogeneous distributed computing (HetDC) system and proposed coded computing schemes to reduce the communication load \cite{ThreeWorker,JiNonCa,Ji,Tao}. In \cite{ThreeWorker}, the authors characterized the information-theoretically minimum communication load for the 3-node system with arbitrary storage size at each node. In \cite{JiNonCa}, the authors proposed a combinatorial design that operates on non-cascaded heterogeneous networks where nodes have varying storage and computing capabilities but each Reduce function is computed exactly once. A cascaded HetDC system was considered in \cite{Ji}, where each Reduce function is allowed to be computed at multiple nodes. The communication load obtained in \cite{Ji} is optimal within a constant factor given their proposed mapping strategies. The authors in \cite{Tao} characterized the upper bound of the optimal communication load as a linear programming problem, and jointly designed the data placement, Reduce function allocation, and data shuffling strategy to reduce the communication load.  However, all existing schemes \cite{ThreeWorker,JiNonCa,Ji,Tao} require delicate data placement and output function assignment (i.e., the data stored by each node and the set of output functions computed by each node both are design parameters), which is not feasible in some practical scenarios when the nodes fetch data without the orchestration of a central server, or when each node autonomously determines its desired output functions. For instance, in \cite{Decentralized'15, Scalable} each user randomly and independently fetches a subset of all data from a central server, and in distributed multi-task learning \cite{wang2016distributed,CodedMTL} the users wish to learn their own model depending on its task goal.  Until now, there exists no work considering the  HetDC systems with pre-set data placement and output function assignment, i.e., the data placement and output function assignment are system parameters that cannot be designed by schemes. 

\subsection{Main Contributions}

In this paper, we consider the {general} HetDC system with \emph{arbitrary} but \emph{pre-set} data placement and Reduce function assignment (general HetDC system in short). We aim to find coded distributed computing schemes that are applicable to arbitrary HetDC systems and characterize the theoretical computation-communication tradeoff. The main contributions of this paper are summarized as follows:

 \begin{itemize}
 \item  To the best of our knowledge, our work is the first attempt to study the general HetDC system with arbitrary and pre-set data placement and Reduce function assignment, with no precedent research in either coded caching or code computing scope. This is a challenging problem because when the data placement at nodes is asymmetric and different Reduce functions are computed by different numbers of nodes, the IVs to be shuffled can be distributed in an extremely irregular manner, prohibiting designing the data shuffling strategy with the help of input files and Reduce functions allocations. This makes all previous schemes in \cite{ThreeWorker,JiNonCa,Ji,Tao} infeasible for the general HetDC system (see detailed discussions in Section \ref{SecComparison}).    To address this challenge, we first characterize the heterogeneities of the HetDC system from the perspective of computing nodes, input files, and IVs, respectively. We introduced the concept of deficit ratio to study the heterogeneous feature of a distributed system on the IV level, while previous definitions only focus on the file or node levels.
 
 % Since the averaged metrics such as computation load in \cite{CDC}, Mapping load and Reduce load in \cite{Tao} are not sufficient to characterize the status of the pre-set file and Reduce function allocations, we introduce new definitions of \emph{Mapping times} for each file as the number of nodes mapping this file, and \emph{Reducing times} for each Reduce function as the number of nodes computing this Reduce function. 
 \item  We propose two CDC approaches for the general HetDC system, namely \emph{One-shot Coded Transmission (OSCT)} and \emph{Few-shot Coded Transmission (FSCT)}. By grouping nodes into clusters with different sizes and categorizing IVs, OSCT delicately designs the size of the encoded messages and partition IVs in a \emph{non-uniform} manner by solving an optimization problem that maximizes the multicast opportunities, and sends each message block in a one-shot manner where each message block can be decoded independently by the intended nodes. To further exploit the gain of multicasting, FSCT splits IVs into smaller pieces and carefully designs the number of linear combinations to be sent by each node based on our newly proposed \emph{deficit condition} and \emph{feasible condition}, and then sends each message block in a few-shot manner where each node needs to jointly decode multiple received message blocks to obtain its desired IVs. 
%  {\color{blue}\sout{Numerical results shows the performance of OSCT and FSCT schemes under pre-set date and Reduce function assignment, demonstrating FSCT can potentially improve OSCT at the cost of higher coding complexity.}}

 \item With the lower bound proposed in \cite{Howto}, we establish optimality results on the fundamental computation-communication tradeoff for several types of HetDC systems, and derive the sufficient conditions for the optimality of OSCT and FSCT, respectively. In particular, we associate the optimality of OSCT with an optimization problem and prove that OSCT is optimal if the optimal value of the objective function is zero. For the FSCT scheme, we associate its optimality with the concept of \emph{deficit ratio} for each computing node, which is the ratio between the number of IVs that the node desires and the number of IVs that the node locally knows. Then we prove that if the deficit condition and feasible condition are satisfied for all nodes in each cluster, the system is well-balanced and FSCT is optimal. %This includes the optimality results of \cite[Theorem 2]{CDC} and \cite[Theorem 1]{ThreeWorker} for the homogeneous system and 3-node system, respectively. 
This not only recovers the existing optimality results known for homogeneous systems \cite{CDC} and 3-node systems \cite{ThreeWorker},  but also includes some cases where our schemes are optimal while other existing methods are not.

 % in Note that the communication loads achieved by OSCT and FSCT can be reduced to \cite[Theorem 2]{CDC} and \cite[Theorem 1]{ThreeWorker} in the homogeneous system and 3-node system respectively.
 %Comparing our upper bounds with the lower bound proposed in \cite[Lemma 1]{Howto},
 \end{itemize}

\subsection{Comparison with Coded Caching Schemes} \label{CDCandCaching}

With the pre-set file allocation and Reduce function assignment, the distributed computing problem, in essence, can be viewed as a special communication problem with (heterogeneous) message cooperation and side information.
The CDC-based approaches are closely related to the coded caching scheme that was first proposed by Maddah-Ali and Niesen \cite{CodedCaching}, as they both involve repetitive data placement and coded multicasting.  Our coded computing schemes can also be viewed as the device-to-device coded caching problem where each user has an arbitrary but pre-set data placement and requires an arbitrary set of files. To the best of our knowledge, there exists no coded caching scheme proposed for this setup.  In the following, we compare the related works in coded caching under heterogeneous setups with our schemes from the perspective of problem formulation and algorithm design. 

The coded caching problem with non-uniform file popularity was considered in \cite{CachingNonuniform,CachingMutilevelPop,CachingArbitraryPop,CachingUncodedOptimi}. In this setup,  one server has a library of $N$ files and each user independently requests file $n\in [N]$ with probability $p_n$. Authors in \cite{CachingNonuniform} proposed a \emph{file gouping} method where the files were partitioned into groups with uniform popularities to preserve the symmetry within each group. Following the idea of file grouping, a specific multi-level file popularity model was considered in \cite{CachingMutilevelPop} and an order-optimal scheme was derived. The work in \cite{CachingArbitraryPop} considered a more general system with arbitrary popularity distribution. The scheme divided the files into 3 groups using a simple threshold and attained a constant-factor performance gap under all popularity distributions. 
An optimization approach was proposed in \cite{CachingUncodedOptimi}, where each file was partitioned into $2^K$ non-overlapping subfiles and the placement strategy was optimized to minimize the average load. This scheme focused on the design of file partitions, while ours focus on the design of coded messages. Note that among all literature mentioned above, each node requires only one file at a time and their cache sizes were uniform, while in our work the data placement and the number of functions computed by each node both are arbitrary and could be non-uniform.

The coded caching problem with the heterogeneity of user demand has been studied in \cite{CachingOverlappingDemand,CachingHeterUserProfile}. In their setups,  users have distinct distributions of demands, and the set of files that different users may desire can possibly overlap. However,   they assumed that each user requests only one file at a time. There are also some works considering the case when users demand multiple files at each time\cite{CachingMultiDemand}\cite{CachingMultiRandomRequ}. Specifically, authors in \cite{CachingMultiDemand} developed a new technique for characterizing information-theoretic lower bounds in both centralized content delivery and device-to-device distributed delivery. In \cite{CachingMultiRandomRequ}, an order-optimal scheme based on the random vector (packetized) caching placement and multiple group cast index coding was proposed. However, the works above considered the homogeneous case where each user has a cache of size $M$ {files} and requests $L\geq 1$ files independently according to common demand distribution, while we allow arbitrary storage and computation capabilities for each node.

For the coded caching setup with distinct file size, \cite{CachingDistinctFileSize} categorized all files into $T$ types by their size, and each user $k$ will cache $q_\ell F_\ell$ bits of every file in $\ell$-th type. By carefully deciding the caching parameter $q_l$ based on the number of files with the same type, the proposed scheme achieved optimal load within a constant gap to the proposed lower bound. For the setup with distinct cache sizes, \cite{CachingHeterCachSizeOptimi} considered a centralized system in which cache sizes of different users were not necessarily equal. The authors proposed an achievable scheme optimized by solving a linear program, where each file was split into $2^K$ subfiles and the sizes of partitions were determined by an optimization problem. In \cite{CachingCachFileSize}, the authors considered a more general setting with arbitrary file and cache size. By first proposing a parameter-based decentralized coded caching scheme and then optimizing the caching parameter, this work provided an optimization-based solution and an iterative algorithm. All aforementioned works considered heterogeneous file sizes or storage capabilities, but the data placement is still free to design, differing from our setting that has arbitrary and present data placement. Meanwhile, the demand for each node is homogeneous, where the request for each node is a single file, while in our setting each node could have multiple requests (Reduce functions), and the demand for every node is arbitrary.

Some works have considered multiple types of non-uniform parameters, such as\cite{CachingFullHeter2User,CachingOptimi}. In \cite{CachingFullHeter2User}, the authors studied fully heterogeneous cases under distinct file sizes, unequal cache size, and user-dependent heterogeneous file popularity. They characterized the exact capacity for a two-file two-user system, while we focus on the coding strategy design in the general case with an arbitrary number of users. In \cite{CachingOptimi}, an optimization framework to design caching schemes with distinct file sizes, unequal cache size, and non-uniform file popularity was proposed. This framework is based on exquisitely designing data placement, which differs from our assumption of pre-set data and output function assignment.

The rest of this paper is organized as follows. Section \ref{Sec:Model} introduces our problem formulation, including the system model and useful definitions. Section \ref{SecToyExample} gives a toy example to elaborate the definitions and provides some intuitions on our two schemes: OSCT and FSCT. Section \ref{MainResults} presents the main results of this paper. Sections \ref{SchOneshot} and \ref{SchFewShot} describe the general OSCT and FSCT schemes, respectively, as well as more examples. Section \ref{Subsec_Numerical} provides some numerical results of OSCT and FSCT,  as well as some discussions. Section \ref{Conclusion} concludes this paper.
 
 \textbf{Notations:} Let $[n]$ denote the set $\{1,2,\cdots,n\}$ for some $n\in \mathbb{Z}^+$, where $\mathbb{Z}^+$ is the set of positive integers. $\mathbb{Q}^+$ denotes the set of positive rational numbers. For a set $\mathcal{S}$, let $|\mathcal S|$ denote its cardinality. For a vector $\boldsymbol{v}$, let $\boldsymbol{v}[i]$ denote its $i_{th}$ element. $(x)^+$ denotes the positive part of $x$, i.e., $(x)^+=x$ for $x\geq 0$ and $(x)^+=0$ for $x<0$. Given an event $\mathscr{E}$, we use $\mathds{1}(\mathscr{E})$ as an indicator, i.e., equals to 1 if the event happens, and 0 otherwise.

 \section{Problem Formulation}\label{Sec:Model}
 \subsection{System Model and Definitions}
 Consider a MapReduce-type system that aims to compute $Q$ output functions from $N$ input files with $K$ distributed nodes, for some positive integers $K,Q,N$. The $N$ input files are denoted by ${w_1,\dots,w_N}\in\mathbb{F}_{2^{F}}$, and the $Q$ output functions are denoted by $\phi_{1},\dots,\phi_{Q}$, where $\phi_{q}$: $(\mathbb{F}_{2^{F}})^{N} \rightarrow \mathbb{F}_{2^{B}}$, $q \in [Q]$, maps all input files into a ${B}$-bit output value $u_q = \phi_{q}(w_{1},\dots,w_{N}) \in \mathbb{F}_{2^{B}}$, for some ${B} \in \mathbb{Z}^+$.
 
  Node $k$, for $k\in [K]$, is required to map a subset of input files and compute a subset of output functions. Denote the indices of input files mapped by Node $k$ as $\mathcal{M}_k\subseteq [N]$ and the indices of output functions computed by $k$ as $\mathcal{W}_k\subseteq [Q]$. To ensure  that all input files are stored and all Reduce functions are computed, we have $\bigcup_{k\in [K]}\mathcal{M}_k =[N]$ and $\bigcup_{k\in [K]}\mathcal{W}_k = [Q]$. 
 
 Unlike the existing CDC-based works (e.g., \cite{CDC,Ji,Tao}) in which the data placement $\mathcal{M}=\{\mathcal{M}_k\}_{k=1}^K$ and Reduce functions assignment $\mathcal{W}=\{\mathcal{W}_k\}_{k=1}^K$ are to be designed by schemes, here we consider the general HetDC system where
  \begin{itemize}
  \item $\{\mathcal{M}_k\}_{k=1}^K$ and $\{\mathcal{W}_k\}_{k=1}^K$ are \emph{pre-set}.
  \item $\{|\mathcal{M}_k|\}_{k=1}^K$ and $\{|\mathcal{W}_k|\}_{k=1}^K$ are \emph{arbitrary} and \emph{unnecessarily equal}, i.e., it is possible to have $|\mathcal{M}_i|\neq|\mathcal{M}_j|, |\mathcal{W}_i|\neq|\mathcal{W}_j|$, for some $i,j\in[K],i\neq j$.
  \end{itemize}
  In the MapReduce framework, there are $N$ Map functions $\Vec{g}_{n} = (g_{1,n},\dots,g_{Q,n}):\,\mathbb{F}_{2^{F}} \rightarrow (\mathbb{F}_{2^{V}})^{Q},$ for each file $n \in [N]$, and each $g_{q,n}$ maps the corresponding file into $Q$ \emph{intermediate values (IVs)} $v_{q,n} = g_{q,n}(w_{n}) \in \mathbb{F}_{2^{V}}$.  Here $V$ denotes the bit-length of a single IV. Meanwhile, for each output function $\phi_q$, there is a corresponding Reduce function $h_q(v_{q,1},\dots,v_{q,N}):(\mathbb{F}_{2^{V}})^{N} \rightarrow \mathbb{F}_{2^{B}}$, which maps the corresponding IVs from all input files into the output value $u_{q}$, for $q\in[Q]$. Combining the Map and Reduce functions, we have $u_q = \phi_{q} (w_{1},\!\ldots\!,w_{N}\!) \!= \! h_{q}\!\left(g_{q,1}(w_{1}),\!\ldots\!,g_{q,{N}}(w_{N})\right).$
 
 The system consists of three phases: \emph{Map}, \emph{Shuffle} and \emph{Reduce}.
 
 \subsubsection{Map Phase} Node $k$, for $k\in [K]$, uses the Map functions to map its stored input files, and obtains local IVs $\{v_{q,n}:q\in [Q], n\in \mathcal{M}_k\}$.

 \subsubsection{Shuffle Phase} 
 Based on the local mapped IVs, Node $k$, for $k\in [K]$, generates an $\ell_k$-bit message $X_k$ using some encoding functions $e_k:\left(\mathbb{F}_{2^V}\right)^{Q|\mathcal{M}_k|}\rightarrow \mathbb{F}_{2^{\ell_k}}$ such that
 \begin{align}
  X_k = e_k(\{v_{q,n}:q\in [Q],n\in\mathcal{M}_k\}).
 \end{align} 
 
 Then Node $k$ broadcasts $X_k$ to other nodes. After the Shuffle phase, each node receives messages $X_1,X_2,\cdots,X_K$.
 
 \subsubsection{Reduce Phase} In this phase, each node first recovers all required IVs and then computes the assigned Reduce functions. More specifically, Node $k$ recover its desired IVs based on the received message $X_1,X_2,\cdots,X_K$ and its local IVs using the decoding function $d_k:\mathbb{F}_{2^{\ell_1}}\times\mathbb{F}_{2^{\ell_2}}\times\cdots\times\mathbb{F}_{2^{\ell_K}}\times\left(\mathbb{F}_{2^V}\right)^{Q|\mathcal{M}_k|}\rightarrow\left(\mathbb{F}_{2^V}\right)^{N|\mathcal{W}_k|}$: 
 \begin{align*}
  &\{v_{q,n}:n\in[N],q\in\mathcal{W}_k\}=d_k(X_1,X_2,\cdots,X_K,\{v_{q,n}:q\in [Q],n\in\mathcal{M}_k\}).
 \end{align*}
 
 Then Node $k$ computes the assigned Reduce functions and get the output values $u_{q}=h_{q}(v_{q,1},\dots,v_{q,N})$, for $q\in \mathcal{W}_k$.
 
 \begin{figure}
   \centering
   \includegraphics[width=8.6cm]{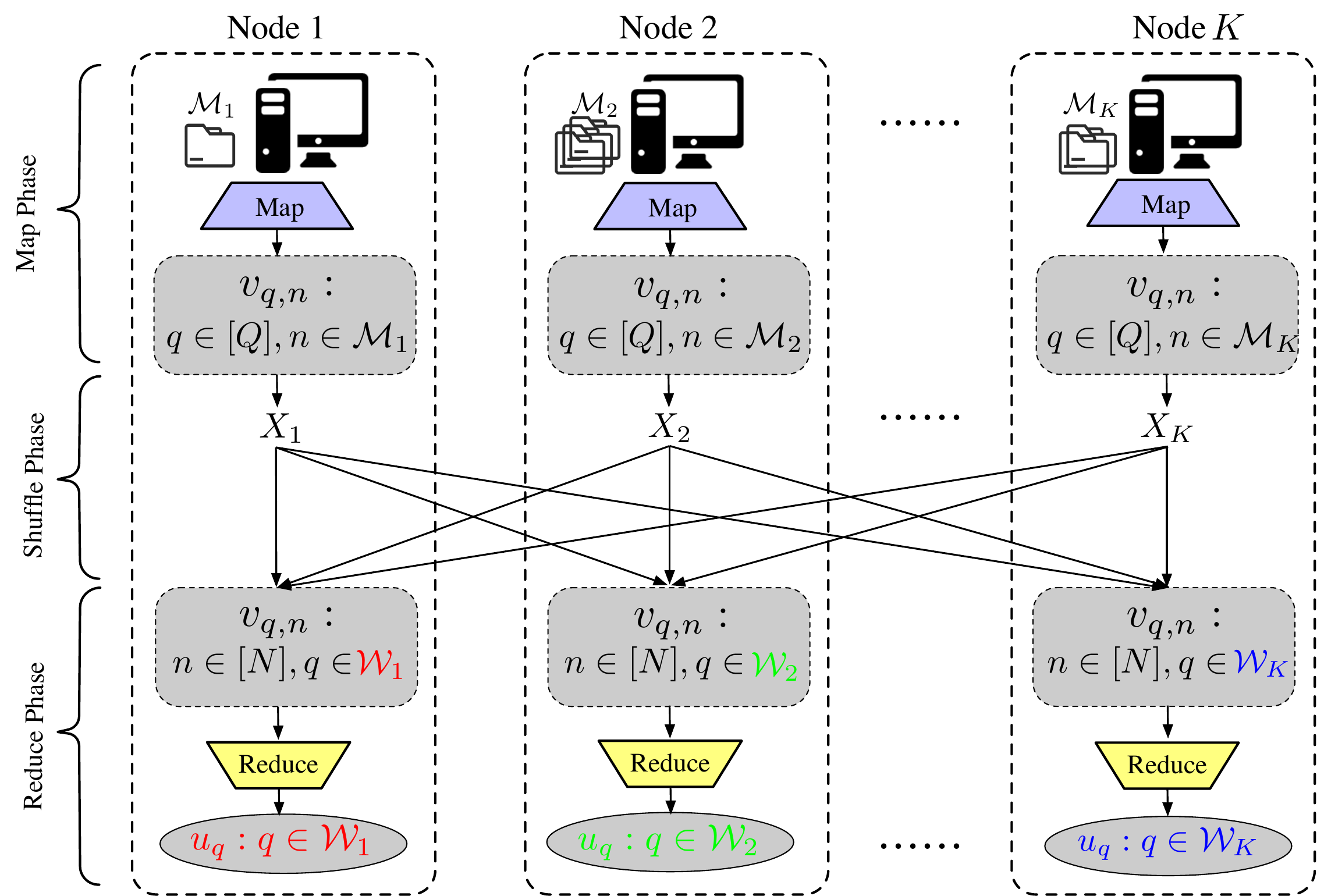}
   \caption{Illustration of MapReduce-type system, where $\mathcal{M}_k\subseteq [N]$ denotes the indices of input files mapped by Node $k$ and $\mathcal{W}_k\subseteq [Q]$ denotes the indices of output functions computed by $k$. $X_k$ is the encoded message generated by Node $k$ during the Shuffle Phase. The IVs are shown in the gray rectangles while the gray ellipses represent the output values.}\label{system_model}
 \end{figure}

 \begin{Definition}[Computation load and Communication load\cite{CDC}]
   The computation load is defined as the number of files mapped across the $K$ nodes normalized by the total number of files, i.e. $r \triangleq {\sum_{k\in[K]}|\mathcal{M}_{k}|}/{N}$.
   The communication load is defined as the total number of bits {(normalized by $NQV$)} communicated by the $K$ nodes in the Shuffle phase, i.e.,
   $L\triangleq\frac{\sum\limits_{k \in[K] }\ell_{k}}{NQV}$. \label{Comload}
 \end{Definition}
 
  We say that a communication load $L(\mathcal{M},\mathcal{W})$ is \emph{feasible} under pre-set data placement $\mathcal{M}$ and Reduce function assignment $\mathcal{W}$, if for any $\epsilon>0$, we can find a set of encoding functions and decoding functions that achieve the communication load $\tilde{L}(\mathcal{M},\mathcal{W})$ s.t. $|L(\mathcal{M},\mathcal{W})-\tilde{L}(\mathcal{M},\mathcal{W})|\leq \epsilon$ and Node $k\in[K]$ can successfully compute all the desired output functions $\{u_q:q\in\mathcal{W}_k\}$.
 
 \begin{Definition}
  The optimal communication load $L^*$ is defined as 
  \begin{align*}
   L^*(\mathcal{M},\mathcal{W}) = \inf \{L(\mathcal{M},\mathcal{W}):L(\mathcal{M},\mathcal{W})\ \textnormal{is feasible}\}.
  \end{align*}
 \end{Definition}

 Given the pre-set data placement and Reduce function assignment, we present two homogeneous setups.
 
 \begin{Definition}[Homogeneous System\cite{CDC}]\label{homo}
  A MapReduce system with pre-set data placement $\mathcal{M}$ and Reduce function assignment $\mathcal{W}$ is called homogeneous if it satisfies the following conditions: 
  (1) the data placement is symmetric where $N$ input files are evenly partitioned into $\tbinom{K}{r}$ disjoint batches $\{\mathcal{B}_{\mathcal{T}}: \mathcal{T}\subset [K],|\mathcal{T}|=r\}$ for some $r\in[K]$, where each batch of files $\mathcal{B}_{\mathcal{T}}$ is mapped by Node $k$ if $k\in\mathcal{T}$, so every $r$ nodes will jointly map ${N}/{\tbinom{K}{r}}$ files; 
  (2) $Q$ Reduce functions are evenly partitioned into $\tbinom{K}{s}$ disjoint batches $\{\mathcal{D}_{\mathcal{P}}:\mathcal{P}\subset [K],|\mathcal{P}|=s\}$ for  some $s\in[K]$, where each batch of Reduce functions $\mathcal{D}_{\mathcal{P}}$ is computed by Node $k$ if $k\in\mathcal{P}$. Thus, every subset of $s$ node will jointly reduce ${Q}/{\tbinom{K}{s}}$ functions.
 \end{Definition}

 \begin{Definition}[Semi-Homogeneous System] \label{semi-sys} A MapReduce system with data placement $\mathcal{M}$ and Reduce function assignment $\mathcal{W}$ is called semi-homogeneous system if it satisfies the following conditions: 
 (1) the data placement is the same as the homogeneous system in Definition \ref{homo}; (2) the Reduce function assignment is multi-level symmetric, i.e., the $Q$ Reduce functions can be partitioned into $K$ disjoint sets: $\mathcal{Q}_1,\cdots,\mathcal{Q}_{K}$ and each set of Reduce functions $\mathcal{Q}_s$ is evenly partitioned into $\tbinom{K}{s}$ disjoint batches $\{\mathcal{D}_{\mathcal{P}}:\mathcal{P}\subset [K],|\mathcal{P}|=s\}$ for  some $s\in[K]$, where each batch of Reduce functions $\mathcal{D}_{\mathcal{P}}$ is computed by Node $k$ if $k\in\mathcal{P}$. 
 \end{Definition}
 
 For the homogeneous system, the pre-set data placement and Reduce function assignment follow the symmetric design of cascaded distributed computing framework in \cite{CDC}, and each Reduce function will be computed by exactly $s$ nodes in a cyclic symmetric manner. In the semi-homogeneous system, we relax the constraint on Reduce function assignment such that the Reduce functions are divided into different sub-groups, and each Reduce function in the same group will be computed by exactly $s$ node in a cyclic symmetric manner. In particular, when only one sub-group exists, the semi-homogeneous system reduces to the homogeneous system.
 
 Next, we introduce definitions to characterize the heterogeneities of the system from the perspective of computing nodes, input files, and IVs, respectively.  From the perspective of computing nodes, we borrow the definitions of Mapping load and Reducing load from \cite{Tao}.
 \begin{Definition}[Mapping load and Reducing load\cite{Tao}]
  The Mapping load for Node $k\in\mathcal{S}$ is defined as $m_k\triangleq\frac{|\mathcal{M}_k|}{N}\in [0,1]$, which denotes the fraction of files that Node $k$ mapped. Similarly, the Reducing load for Node $k$ is defined as $\omega_k\triangleq\frac{|\mathcal{W}_k|}{Q}\in [0,1]$. 
 \end{Definition}
 
 From the perspective of input files, we introduce new definitions of Mapping times and Reducing times as follows:
 \begin{Definition}[Mapping times and Reducing times]
  The Mapping times of file $w_i$, denoted as $r_i$, for $i\in[N]$ is defined as the number of nodes who map file $w_i$. Then the minimum Mapping times of a system is defined as $r_{\textnormal{min}} \triangleq \min\limits_{1\leq i\leq N}r_i$. Similarly, the Reducing times of function $h_i$, denoted as $q_i$, for $i\in[Q]$ is defined as the number of nodes who compute Reduce function $h_i$. Then the minimum Reducing times of a system is defined as $q_{\textnormal{min}} \triangleq \min\limits_{1\leq i\leq Q}q_i$.
 \end{Definition}
 
 From the perspective of IVs, we first define $(z,\mathcal{S})$-mapped IVs and then define the \emph{deficit} of Node $i\in\mathcal{S}$ as the ratio between the number of locally known IVs and the number of desired IVs at Node $i$. Before giving the formal definitions, we first define $\mathcal V^{\mathcal S\backslash \mathcal S_1}_{\mathcal S_1}$ as the set of IVs needed by \emph{all} nodes in $\mathcal S\backslash \mathcal S_1$, no node outside $\mathcal S$, and known \emph{exclusively} by nodes in $\mathcal S_1$, for some $\mathcal{S}\subseteq [K]$ and $\mathcal S_1\subset \mathcal S:|\mathcal S_1|\geq r_{\min}$, i.e.,
 \begin{align}
 \mathcal {V}_{\mathcal {S}_{1}}^{\mathcal {S} \backslash \mathcal {S}_{1}}\triangleq\{v_{q,n}: q \in &\underset {k \in \mathcal {S} \backslash \mathcal {S}_{1}}{\cap } {\mathcal{ W}}_{k}, q\notin \underset {k \notin \mathcal {S}}{\cup } {\mathcal{ W}}_{k}, n  \in \underset {k \in {\mathcal{ S}}_{1}}{\cap } \mathcal {M}_{k}, n\notin \underset {k \notin \mathcal {S}_{1}}{\cup } {\mathcal{ M}}_{k}\}. \label{szmap}
 \end{align}
 
 \begin{Definition}[$(z,\mathcal{S})$-mapped IVs]
 Given a subset $\mathcal{S}\subseteq [K]$ and an integer $r_{\min}\leq z\leq |\mathcal{S}|$, we define the \emph{$(z,\mathcal{S})$-mapped IVs}, denoted by $\mathcal{V}_{z}^{\mathcal{S}}$, as the set of IVs required by $|\mathcal S|-z$ nodes and known exclusively by $z$ nodes in $\mathcal{S}$, i.e.,
 \begin{align}\label{z-mapped}
  \mathcal{V}_{z}^{\mathcal S} = \bigcup_{\mathcal S_1\subseteq \mathcal S:|\mathcal S_1|=z}\mathcal {V}_{\mathcal S_1}^{\mathcal {S} \backslash \mathcal {S}_1}.
 \end{align}
 \end{Definition}

 \begin{Definition}[Deficit Ratio]\label{DefDeficit}
 Given a system with data placement $\mathcal{M}$ and Reduce function assignment $\mathcal{W}$, we define the deficit ratio of Node $i$ on the $(z,\mathcal{S})$-mapped IVs $\mathcal{V}_{z}^{\mathcal{S}}$ as
  \begin{align}
   \textit{D}(i, \mathcal S,z)\triangleq \frac{\sum_{\substack{S_1\subseteq \mathcal S:\\|S_1|=z,i\not\in S_1}}| \mathcal {V}_{\mathcal {S}_1}^{\mathcal {S} \backslash \mathcal {S}_1}|}{\sum_{\substack{S_1\subseteq \mathcal S:\\|S_1|=z,i\in S_1}} |\mathcal {V}_{\mathcal {S}_1}^{\mathcal {S} \backslash \mathcal {S}_1}|}. \label{deficit}
  \end{align}
 For convenience, we let $\textit{D}(i, \mathcal S,z)=0$ when $$\sum\limits_{\substack{S_1\subseteq \mathcal S: |S_1|=z,i\not\in S_1}}| \mathcal {V}_{\mathcal {S}_1}^{\mathcal {S} \backslash \mathcal {S}_1}|=\sum\limits_{\substack{S_1\subseteq \mathcal S: |S_1|=z,i\in S_1}} |\mathcal {V}_{\mathcal {S}_1}^{\mathcal {S} \backslash \mathcal {S}_1}|=0.$$
 \end{Definition}
 The deficit ratio $ \textit{D}(i, \mathcal S,z)$ characterizes the heterogeneity of a system at the IV level, which represents the ratio between the number of IVs required by  Node $i$  and the number of files mapped by Node $i$ with respect to the $(z,\mathcal{S})$-mapped IVs. Since we consider the pre-set scenario where the data placement cannot be designed, the IVs could be distributed in an extremely irregular manner but previous definitions could not reveal this feature, which prompts us to introduce the above definition of deficit ratio.

\section{A Toy Example}\label{SecToyExample}
We first introduce a toy example to elaborate on the definitions introduced in Section \ref{Sec:Model} and illustrate the key idea of transmission strategies of OSCT and FSCT. We will focus on the heterogeneous case and explain why the existing CDC transmission fails to handle the heterogeneities from the perspective of IVs.

\begin{figure}
  \centering
  \includegraphics[width=9.8cm]{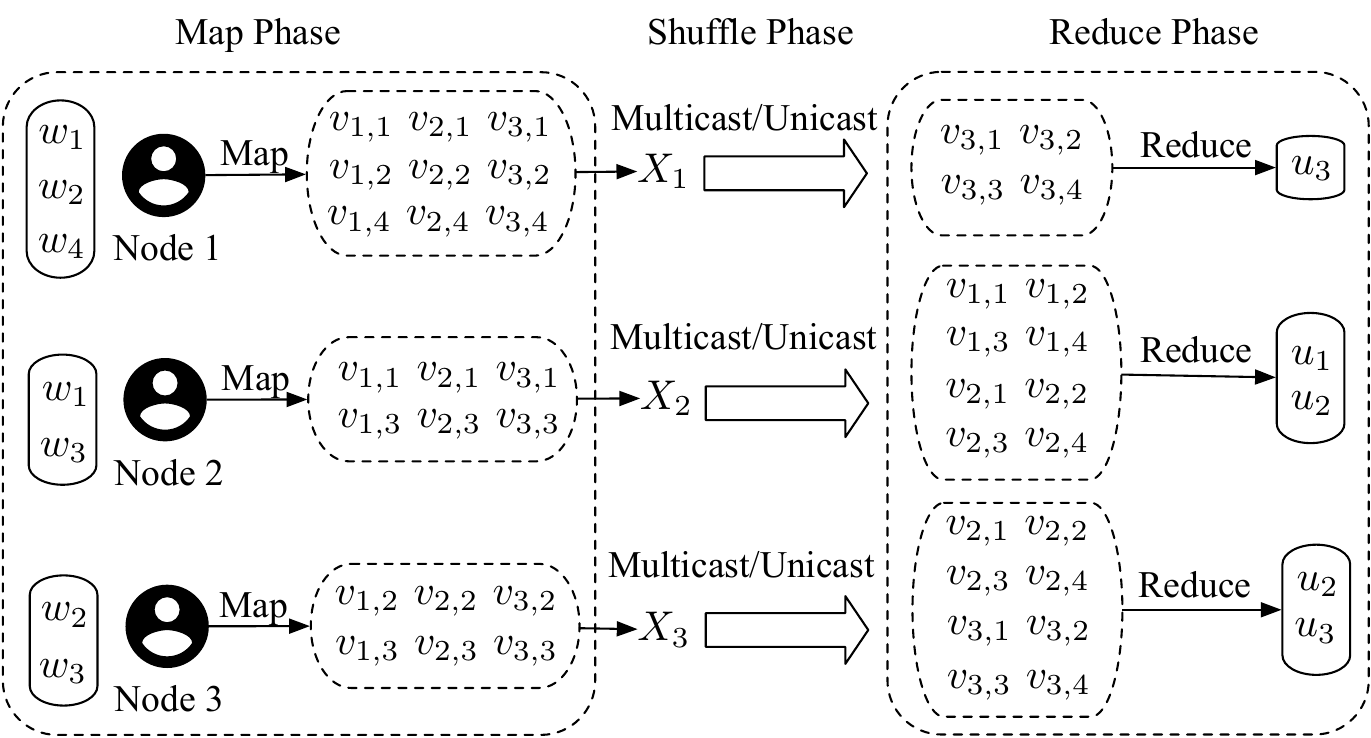}
   \caption{A toy example of a heterogeneous system with $K=3$, $N=4$, $Q=3$.}\label{toy_example2}
\end{figure}
Consider a heterogeneous system with $K=3$ nodes, $N=4$ files, and $Q=3$ Reduce functions, as shown in Fig. \ref{toy_example2}. Note that the data placement and Reduce function assignment are imbalanced, where Node 1 stores more files but computes fewer functions. In the Map phase, each node will map its stored input files to generate corresponding IVs. For example, Node 1 will obtain IVs $v_{i,1}$, $v_{i,2}$, and $v_{i,4}$ for $i=1,2,3$ from files $w_1$, $w_2$, and $w_4$. 

To better understand the system from the perspective of IVs, we first group IVs into different sets and consider the $(z,\mathcal{S})$-mapped IVs for $\mathcal{S}\subset [K]$ and $r_{\min}\leq z\leq |\mathcal S|-1$. In this toy example, we only need to consider three cases: $|\mathcal S|=2, z=1$, $|\mathcal S|=3, z=1$, and $|\mathcal S|=3, z=2$. When $z=1$, each $(z,\mathcal{S})$-mapped IV is known by only one node, indicating that no multicast gain can be obtained and we thus directly unicast the corresponding IVs, i.e., $\mathcal{V}_{\{1\}}^{\{2\}}=\{v_{1,4}\},~\mathcal{V}_{\{1\}}^{\{3\}}=\{v_{1,4}\},~\mathcal{V}_{\{1\}}^{\{2,3\}}=\{v_{2,4}\}$.

In the following, we focus on the case of $\mathcal{S}=\{1,2,3\}, z=2$ and consider the $(z,\mathcal{S})=(2,\{1,2,3\})$-IVs:
\begin{align*}
  \mathcal{V}_{\{1,2\}}^{\{3\}}=\{v_{2,1},v_{3,1}\},~ \mathcal{V}_{\{1,3\}}^{\{2\}}=\{v_{1,2},v_{2,2}\},~ \mathcal{V}_{\{2,3\}}^{\{1\}}=\{v_{3,3}\}.
\end{align*}
Note that the heterogeneity over IVs can \emph{not} be eliminated by simply grouping. Within each cluster $S$, the IVs are distributed unevenly among different nodes (e.g., $|\mathcal{V}_{\{1,3\}}^{\{2\}}|\neq|\mathcal{V}_{\{2,3\}}^{\{1\}}|$), making the existing CDC scheme infeasible\cite{CDC,Ji,Tao}. We can characterize the heterogeneity with the concept of \emph{deficit ratio}. By Definition \ref{DefDeficit}, the deficit ratio of Node 1 in the node cluster $\{1,2,3\}$ with $z=2$ is $\frac{1}{2+2}=\frac{1}{4}$, and for Node 2 and Node 3 the deficit ratio is $\frac{2}{3}$. The intuition is that, in a specific round, the node that maps fewer IVs but requires more IVs has a larger deficit ratio. The values of deficit ratios are closely related to the transmission design of our OSCT and FSCT.

\begin{figure}
  \centering
  \includegraphics[width=15cm]{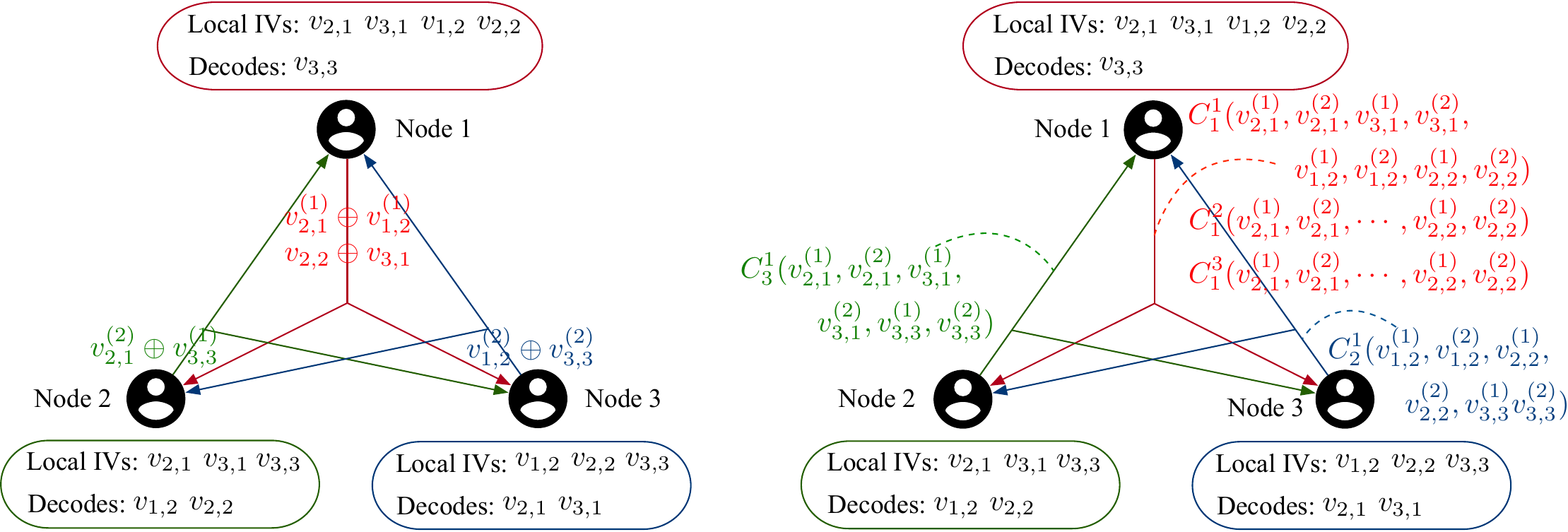}
  \caption{Transmission strategies for $(z,\mathcal{S})=(2,\{1,2,3\})$-mapped IVs of OSCT (left) and FSCT (right) in the toy example.}\label{toy_example_trans}
\end{figure}

By carefully designing the size of messages or the number of linear combinations (LCs) sent from each node, OSCT and FSCT can overcome the heterogeneity of the IV distribution. Specifically, in the OSCT, we first compute a parameter $\alpha_{j,z}^{\mathcal{S}}$ (to be defined in \eqref{alpha_opt}) for Node $j$, $j\in\mathcal{S}$ according to Theorem \ref{ThmOneshot} and Remark \ref{explicitAlpha}, obtaining that $\alpha_{1,2}^{\{1,2,3\}}=\frac{3}{2}$, $\alpha_{2,2}^{\{1,2,3\}}=\frac{1}{2}$, $\alpha_{3,2}^{\{1,2,3\}}=\frac{1}{2}$. Then, the IVs in each set $\mathcal{V}_{\mathcal{S}_1}^{\mathcal{S}\backslash\mathcal{S}_1}$ for $\mathcal{S}_1\subset\mathcal{S}$ and $|\mathcal{S}_1|=z$ are concatenated and split into $z$ segments according to the value of $\alpha_{j,z}^{\mathcal{S}}$. For example, since $\alpha_{1,2}^{\{1,2,3\}}=\frac{3}{2}$ and $\alpha_{2,2}^{\{1,2,3\}}=\frac{1}{2}$, the concatenated $(v_{2,1},v_{3,1})$ form $\mathcal {V}_{\{ 1,2\}}^{\{3\}}$ should be splitted into two segments with sizes $\frac{3}{2}V$ and $\frac{1}{2}V$. Operationally, we can evenly split $v_{2,1}$ into $v_{2,1}^{(1)}$ and $v_{2,1}^{(2)}$, where $(v_{2,1}^{(1)}, v_{3,1})$ with size of $\alpha_{1,2}^{\{1,2,3\}}V=\frac{3}{2}V$ bits will be encoded by Node $1$ and $v_{2,1}^{(2)}$ with size of $\alpha_{2,2}^{\{1,2,3\}}V=\frac{1}{2}V$ bits will be encoded by Node $2$. We list all the split IV segments in round $2$ of cluster $\{1,2,3\}$ in Table \ref{ex1_split_table}.
\begin{table}[!htbp]
 \caption{Split IV segments based on the solution of $\{\alpha_{j,z}^{\mathcal{S}}\}_{j\in\mathcal{S}}$ in round $2$ of cluster $\{1,2,3\}$.}
 \centering\label{ex1_split_table}
 \begin{tabular}{|c|c|c|c|} \hline
 \diagbox[width = 3cm]{IV set}{Size (bits)} & $\alpha_{1,2}^{\{1,2,3\}}V$ & $\alpha_{2,2}^{\{1,2,3\}}V$ & $\alpha_{3,2}^{\{1,2,3\}}V$ \\\hline
 $\mathcal {V}_{\{1,2\}}^{\{3\}}\!\!=\!\!\{v_{2,1},v_{3,1}\}$ & $v_{2,1}^{(1)}, v_{3,1}$ & $v_{1,1}^{(2)}$ & \\\hline
 $\mathcal {V}_{\{1,3\}}^{\{2\}}\!\!=\!\!\{v_{1,2},v_{2,2}\}$ & $v_{1,2}^{(1)},v_{2,2}$ & & $v_{1,2}^{(2)}$ \\\hline
 $\mathcal {V}_{\{2,3\}}^{\{1\}}\!\!=\!\!\{v_{3,3}\}$ & & $v_{3,3}^{(1)}$ & $v_{3,3}^{(2)}$ \\\hline
 \end{tabular}
\end{table}

Then, each node encodes the corresponding segments into LCs, and the receivers can decode the message block instantaneously.  For example, Node 1 multicasts $(v_{2,1}^{(1)}\oplus v_{1,2}^{(1)}, v_{2,2}\oplus v_{3,1})$ to Node 2 and Node 3. Since Node 2 knows $v_{2,1}^{(1)}, v_{3,1}$ and Node 3 knows $v_{1,2}^{(1)}, v_{2,2}$, both of them can decode the messages sent by Node 1. Fig. \ref{toy_example_trans}(left) shows the complete transmission process. The detail design of OSCT will be discussed in Section \ref{SchOneshot}. 

For the FSCT, after computing and updating a parameter $n_{j,z}^{\mathcal{S}}$ (to be defined in \eqref{nAndBeta}) for each node $j$,  we can determine the number of LCs to be sent in each message block. Here $n_{1,2}^{\{1,2,3\}}=3, n_{2,2}^{\{1,2,3\}}=1, n_{3,2}^{\{1,2,3\}}=1$ according to \eqref{nAndBeta}, and the deficit and feasible condition (to be defined in Section \ref{subsec_FSCT}) are satisfied. Then every IV in $(z,\mathcal{S})=(2,\{1,2,3\})$ will be split into $|\mathcal S|-1=2$ pieces with equal size, and each node will transmit a specific number of independent LCs of its locally mapped $(z,\mathcal{S})$-IVs. For example, Node 1 multicast three LCs of IVs $v_{2,1}^{(1)}, v_{2,1}^{(2)}, v_{3,1}^{(1)}, v_{3,1}^{(2)}, v_{1,2}^{(1)}, v_{1,2}^{(2)}, v_{2,2}^{(1)}, v_{2,2}^{(2)}$, denoted as
\begin{align*}
  C_1^1(v_{2,1}^{(1)}, v_{2,1}^{(2)}, v_{3,1}^{(1)}, v_{3,1}^{(2)}, v_{1,2}^{(1)}, v_{1,2}^{(2)}, v_{2,2}^{(1)}, v_{2,2}^{(2)}),\\
  C_1^2(v_{2,1}^{(1)}, v_{2,1}^{(2)}, v_{3,1}^{(1)}, v_{3,1}^{(2)}, v_{1,2}^{(1)}, v_{1,2}^{(2)}, v_{2,2}^{(1)}, v_{2,2}^{(2)}),\\
  C_1^3(v_{2,1}^{(1)}, v_{2,1}^{(2)}, v_{3,1}^{(1)}, v_{3,1}^{(2)}, v_{1,2}^{(1)}, v_{1,2}^{(2)}, v_{2,2}^{(1)}, v_{2,2}^{(2)}),
\end{align*}
with randomly generated coefficient over $\mathbb{F}_{2^{\frac{V}{|\mathcal{S}|-1}}}$. Similarly, Node 2 and Node 3 multicast one LC of their locally computed IV segments, respectively, denoted as
\begin{align*}
  C_2^1(v_{2,1}^{(1)}, v_{2,1}^{(2)}, v_{3,1}^{(1)}, v_{3,1}^{(2)}, v_{3,3}^{(1)}, v_{3,3}^{(2)}),\\
  C_3^1(v_{1,2}^{(1)}, v_{1,2}^{(2)}, v_{2,2}^{(1)}, v_{2,2}^{(2)}, v_{3,3}^{(1)}, v_{3,3}^{(2)}),
\end{align*}
with randomly generated coefficient over $\mathbb{F}_{2^{\frac{V}{|\mathcal{S}|-1}}}$. See Fig. \ref{toy_example_trans}(right) for the illustration. Note that the receiver could obtain the desired IVs by \emph{jointly} decoding the LCs. For example, Node 2 could decode the unknown segments $v_{1,2}^{(1)}, v_{1,2}^{(2)}, v_{2,2}^{(1)}, v_{2,2}^{(2)}$ from the received four independent LCs sent by Node 1 and Node 3. The detail design of FSCT will be discussed in Section \ref{SchFewShot}.

In this example, both OSCT and FSCT achieve the sample communication load of $\frac{11}{24}$, which is optimal according to Lemma \ref{LemmaLower} to be presented in Section \ref{MainResults}. Here, OSCT shows a lower coding complexity. Specifically, when transmitting $(z,\mathcal{S})$-mapped IVs, in FSCT, each IV is split into two pieces, and Node 1 multicasts three LCs of eight variables, while for the same round in OSCT, $v_{3,1}$ and $v_{2,2}$ need not to be split and the coding only involves XORs. However, we will see by example in Section \ref{ExFSCT} that FSCT can potentially outperform OSCT through further exploiting the multicast gain in some circumstances.

\section{Main Results} \label{MainResults}
In this section, we first present a lower bound of the optimal communication load proved in \cite{Howto}, and then we present two achievable upper bounds for the optimal communication load under pre-set data placement and Reduce function assignment. By comparing both upper and lower bounds, we show that our schemes are optimal in some non-trivial cases.
 
\begin{lemma}(\cite[Lemma 1]{Howto})\label{LemmaLower}
 Consider a distributed computing system with parameters $N$ and $Q$, given a data placement $\mathcal{M}$ and Reduce function assignment $\mathcal{W}$ that uses $K$ nodes. For any integers $t, d$, let $a_{t,d}$ denotes the number of intermediate values that are available at $t$ nodes, and required by (but not available at) $d$ nodes. The communication load of any valid shuffling design is lower bounded by:
 \begin{equation} 
 L^{*}(\mathcal{M}, \mathcal{W}) \geq\frac{1}{QN}\sum_{t=1}^{K}\sum_{d=1}^{K-t}a_{t,d}\frac{d}{t+d-1}.\label{lowerbound}
 \end{equation}
 \end{lemma}

\begin{remark}
  If the lower bound \eqref{lowerbound} is tight, it indicates that a scheme will be optimal if it can always transmit the IVs available at $t$ nodes and required by (but not available at) $d$ nodes with a multicast gain (ratio between the achievable communication load of the coded and unicast schemes) of $\frac{d}{t+d-1}$. We will show in the following sections that by delicately designing the size of the encoded messages and the number of linear combinations to be sent, our OSCT and FSCT could achieve the maximum multicast gain in many scenarios.
\end{remark}

\subsection{One-shot Coded Transmission (OSCT)}
 \subsubsection{ Upper bound of OSCT}
Given a cluster $S\subseteq [K]$ and $z\in \mathbb{Z}$ with $r_{\min}\leq z\leq |\mathcal S|$, denote the $\binom{|\mathcal S|}{z}$ subsets of $S$ of size $z$ as $\mathcal S_z[1], \mathcal S_z[2], \cdots, \mathcal S_z[\binom{|\mathcal S|}{z}]$ respectively. 
%Here we propose the achievable upper bound given by OSCT in the following theorem. 
The upper bound achieved by OSCT is given in the following theorem.
\begin{theorem}\label{ThmOneshot}
 Consider the distributed computing system with $K$ nodes, $N$ files and $Q$ Reduce functions under pre-set data placement $\mathcal{M}$ and Reduce function assignment $\mathcal{W}$ for $k\in [K]$. The optimal communication load is upper bounded by $L^*\leq L_{\textnormal{OSCT}}$ where $L_{\textnormal{OSCT}}$ is defined as:
 \begin{align}
 L_{\textnormal{OSCT}}&
 = {\sum\limits_{\mathcal{S}\subseteq[K]}\sum\limits_{z=r_{\min}}^{|\mathcal{S}|-1} \frac{1}{QN}\Big(\sum\limits_{j\in\mathcal{S}}\tbinom{|\mathcal S|-2}{z-1}\alpha_{j,z}^{\mathcal{S}}\! +\!\sum\limits_{i=1}^{\binom{|\mathcal{S}|}{z}}(\tau_{i,z}^{\mathcal{S}})^+\Big)} \label{load1},
 \end{align}
 where $$\tau_{i,z}^{\mathcal{S}}=|\mathcal{V}_{\mathcal {S}_z[i]}^{\mathcal {S} \backslash \mathcal {S}_z[i]}|- \sum_{j\in\mathcal S_z[i]}\alpha_{j,z}^{\mathcal{S}},$$ and $\{\alpha_{j,z}^{\mathcal{S}}\}_{j\in\mathcal{S}}$ are given by the following optimization problem:
 \begin{align}
  & \mathcal{P}_{\textnormal{OSCT}}(\mathcal{S},z):\label{alpha_opt} \\
  & \min_{\{\alpha_{j,z}^{\mathcal{S}}\}_{j\in\mathcal{S}}} \sum_{i=1}^{\binom{|\mathcal S|}{z}}\Big(\left|\mathcal{V}_{\mathcal {S}_z[i]}^{\mathcal {S} \backslash \mathcal {S}_z[i]}\right|-\sum_{j\in\mathcal S_z[i]}\alpha_{j,z}^{\mathcal{S}}\Big)^2\label{conditionOptThm2} \tag{7a}\\
  &\quad {s.t.}~\alpha_{j,z}^{\mathcal{S}}\in \mathbb{Q}^+, \forall j\in\mathcal S \tag{7b}\\
  &\ \ \ \ \ \ \sum_{j\in \mathcal \mathcal S_z[i]}\alpha_{j,z}^{\mathcal{S}}=0,~\textnormal{if}~ |\mathcal{V}_{\mathcal {S}_z[i]}^{\mathcal {S} \backslash \mathcal {S}_z[i]}|=0, \forall i \in \left[\tbinom{|\mathcal{S}|}{z}\right] \tag{7c}.
  \label{min}
 \end{align}\label{t1}
\end{theorem}

We prove Theorem \ref{ThmOneshot} in Section \ref{SchOneshot} by proposing the achievable scheme OSCT. Meanwhile, one may refer to the toy example in Section \ref{SecToyExample} for an intuitive illustration.

\begin{remark}\label{explicitAlpha}
 Here $\alpha^{\mathcal S}_{j,z}$ denotes the size of an encoded message (normalized by $V$ bits) sent by Node $j\in \mathcal{S}$ in OSCT when sending the $(z,\mathcal{S})$-mapped IVs $\mathcal{V}_z^{\mathcal{S}}$. In Appendix \ref{derivationAlpha}, we obtain that if the deficit ratio of each node $j\in\mathcal{S}$ is less than the threshold value $\frac{|\mathcal S|-z}{z-1}$ and $|\mathcal{V}_{\mathcal {S}_z[i]}^{\mathcal {S} \backslash \mathcal {S}_z[i]}|\neq 0, $ for all $i\in\left[\binom{|\mathcal{S}|}{z}\right]$, the solution of the optimization problem $\mathcal{P}_{\textnormal{OSCT}}(\mathcal{S},z)$ is 
 \begin{align}
  \alpha_{j,z}^{\mathcal{S}}=\!\!\!\!\!\!\sum\limits_{\substack{\mathcal{S}_1\subseteq \mathcal{S}:\\|\mathcal{S}_1|=z,j\in \mathcal{S}_1}}\!\!\!\!\!\frac{|\mathcal {V}_{\mathcal{S}_1}^{\mathcal {S} \backslash \mathcal{S}_1}|}{{\tbinom{|\mathcal S|-1}{z-1}}}
 -\!\!\!\!\!\!\sum\limits_{\substack{\mathcal S_1\subseteq \mathcal{S}:\\|\mathcal S_1|=z,j\not\in \mathcal S_1}}\!\!\!\!\!\! \frac{(z-1)| \mathcal {V}_{\mathcal S_1}^{\mathcal {S} \backslash \mathcal S_1}|}{(|\mathcal S|-z) \binom{|\mathcal S|-1}{z-1}},~ \forall j\in\mathcal{S}. \label{explicitAlphaEq}
 \end{align}
\end{remark}

\subsubsection{Optimality of OSCT}

The following theorem states that the upper bound in the Theorem \ref{ThmOneshot} is tight in some cases.

%Although the first upper bound is based on a optimization problem, we have the following theorem to show its optimality.

\begin{theorem}
 If the optimal value of the objective function in $\mathcal{P}_{\textnormal{OSCT}}(\mathcal{S},z)$ satisfies 
 \begin{IEEEeqnarray}{rCl}%\label{conditionOptThm2}
 \sum_{i=1}^{\binom{|\mathcal S|}{z}}\Big(\sum_{j\in\mathcal S_z[i]}\alpha_{j,z}^{\mathcal{S}}-|\mathcal{V}_{\mathcal {S}_z[i]}^{\mathcal {S} \backslash \mathcal {S}_z[i]}&{}|{}&\Big)^2=0,~\forall \mathcal{S}\subseteq [K], r_{\min}\leq z\leq |\mathcal{S}|,\notag
\end{IEEEeqnarray}
 then the communication load in Theorem \ref{t1} is optimal. \label{t2}
\end{theorem}

\begin{IEEEproof}
 See Appendix \ref{ConverseOSCTOPt}.
\end{IEEEproof}
Intuitively, when the optimal value of $\mathcal{P}_{\text{OSCT}}(\mathcal{S},z)$ achieves zero for some $\mathcal{S}\in [K]$ and $r_{\min}\leq z\leq |\mathcal{S}|$, all IVs in $\mathcal{V}_{\mathcal {S}_z[i]}^{\mathcal {S} \backslash \mathcal {S}_z[i]},\forall i\in[\binom{|\mathcal{S}|}{z}]$ can be sent with a multicast gain $\frac{|\mathcal{S}|-z}{|\mathcal{S}|-1}$, coinciding with the lower bound in Lemma \ref{LemmaLower}.
To better illustrate the optimality of our OSCT, we present the following corollary based on Theorem \ref{t2}. It shows that our scheme achieves the minimum communication load  in several cases that include the existing optimality results in    \cite{CDC} and \cite{ThreeWorker}.

 \begin{Corollary}\label{OSCTCorollary}
 The communication load in Theorem \ref{t1} is optimal for the homogeneous system and semi-homogeneous system defined in Definition \ref{homo} and \ref{semi-sys} respectively. For the 3-node system considered in \cite{ThreeWorker}, which allows any data placement $\mathcal{M}$ and symmetric Reduce function assignment $\mathcal{W}$ with $|\mathcal{W}_i|=1,\mathcal{W}_i\cap \mathcal{W}_j=\emptyset$, for all $i\neq j$ and $i,j\in[3]$, the communication load in Theorem 1 is optimal if the pre-set $\mathcal{M}$ and $\mathcal{W}$ are the same as \cite{ThreeWorker}.
 %in 3-node system, given the design of data placement and Reduce function assignment in\cite{CDC} and \cite{ThreeWorker} respectively. The load given by Theorem \ref{t1} is optimal and matches the lower bound in Lemma 1. For three y
 \end{Corollary}

 % \begin{Corollary}[Homogeneous System]
 % If the system degenerates to a homogeneous system\cite{CDC} i.e., $m_k=m\in [\frac{1}{K},1]$ and $w_k=w\in [\frac{1}{K},K]$ and $m,w\in \mathbb{Q}$ for all $k\in [K]$, the load given by Theorem \ref{t1} will be reduced to the optimal communication load in \cite[Theorem 2]{CDC} given the design of data placement and Reduce function assignment in the same paper.
 % \end{Corollary}

 \begin{IEEEproof}
 For the homogeneous system with computation load $r$, $|\mathcal{V}_{\mathcal {S}_z[i]}^{\mathcal {S} \backslash \mathcal {S}_z[i]}|$ is the same for all $i\in[\binom{|\mathcal{S}|}{z}],z=r$. Since the deficit ratio of Node $j$ in round $z$ of cluster $\mathcal{S}$ is $\frac{\binom{|\mathcal{S}|-1}{z}}{\binom{|\mathcal{S}|-1}{z-1}}=\frac{|\mathcal{S}|-1}{z}<\frac{|\mathcal{S}|-z}{z-1}$, according to Remark \ref{explicitAlpha},  we have
 \begin{align}
  \alpha_{j,z}^{\mathcal{S}}&=\!\!\!\!\sum\limits_{\substack{\mathcal{S}_1\subseteq \mathcal{S}:\\|\mathcal{S}_1|=z,j\in \mathcal{S}_1}}\!\!\!\frac{|\mathcal {V}_{\mathcal{S}_1}^{\mathcal {S} \backslash \mathcal{S}_1}|}{{\tbinom{|\mathcal S|-1}{z-1}}}
 -\!\!\!\!\!\sum\limits_{\substack{\mathcal S_1\subseteq \mathcal{S}:\\|\mathcal S_1|=z,j\not\in \mathcal S_1}}\!\!\!\! \frac{(z-1)| \mathcal {V}_{\mathcal S_1}^{\mathcal {S} \backslash \mathcal S_1}|}{(|\mathcal S|-z) \binom{|\mathcal S|-1}{z-1}} =\frac{1}{z}|\mathcal{V}_{\mathcal {S}_z[i]}^{\mathcal {S} \backslash \mathcal {S}_z[i]}|, ~\forall j\in\mathcal{S}.
 \end{align}
Thus, the objective function $\sum_{j\in\mathcal S_z[i]}\alpha_{j,z}^{\mathcal{S}}-|\mathcal{V}_{\mathcal {S}_z[i]}^{\mathcal {S} \backslash \mathcal {S}_z[i]}|=0$, indicating the communication load given by OSCT is optimal according to Theorem \ref{t2}.
For the 3-node and semi-homogeneous systems, we prove the optimality in Appendix \ref{OSCT3wProve} and \ref{OSCTSemiProve}, respectively. 
  
  %we can get $\alpha^{\mathcal S}_{i,z}=\frac{NQ}{r}\frac{\binom{r}{|\mathcal S|-s}}{\binom{K}{r}\binom{K}{s}}, \forall i\in \mathcal{S}$ if $z=r$, $\max\{r+1,s\}<|S|<\min\{r+s,K\}$, and $\alpha^{\mathcal S}_{i,z}=0$ otherwise. Plug in (\ref{load1}) and we can obtain the load in \cite[Theorem 2]{CDC}.
 \end{IEEEproof}

 \subsection{Few-shot Coded Transmission (FSCT)} \label{subsec_FSCT}
 The OSCT achieving the upper bound in Theorem \ref{ThmOneshot} is a One-shot method in the sense that each \emph{single} message block can be independently decoded by its intended nodes. To further explore the gain of multicasting and reduce the communication load, we design a {Few-shot Coded Transmission} (FSCT) scheme where each node can jointly decode \emph{multiple} message blocks to obtain its desired IVs.

Here we first introduce a parameter $n_{i,z}^{\mathcal S}$ for $i \in\mathcal{S}$
\begin{align}
 n_{i,z}^{\mathcal S}=\!\!\!\!\!\!\sum\limits_{\substack{S_1\subseteq \mathcal S:\\|S_1|=z,i\in S_1}}(|\mathcal S|\!\!-\!\!z)|\mathcal {V}_{\mathcal {S}_1}^{\mathcal {S} \backslash \mathcal {S}_1}|\!-\!\!\!\!\sum\limits_{\substack{S_1\subseteq \mathcal S:\\|S_1|=z,i\not\in S_1}}(z-1)|\mathcal {V}_{\mathcal {S}_1}^{\mathcal {S} \backslash \mathcal {S}_1}|,\label{nAndBeta}
\end{align}
and the following two conditions:

1) \emph{Deficit Condition:} 
%For each $\mathcal{S}\in [K]$ and its $(z,\mathcal{S})$-mapped IVs,
 For Node $i\in\mathcal{S}$ and $(z,\mathcal{S})$-mapped IVs, define the deficit condition, denoted by $\mathscr{E}_1(i,\mathcal{S},z)$, as
\begin{IEEEeqnarray}{rCl}\label{defiCond}
\textit{D}(i,\mathcal S,z)\leq\frac{|\mathcal S|-z}{z-1},
\end{IEEEeqnarray}
where $\textit{D}(i,\mathcal S,z)$ is the deficit ratio defined in \eqref{deficit}.
 %which is a measurement of the ``ability'' of a node to balence its desired IVs with its mapped IVs.
From \eqref{deficit} and \eqref{nAndBeta}, we obtain that the deficit condition \eqref{defiCond} is equivalent to
\begin{align}
 n_{i,z}^{\mathcal S}\geq 0.\label{deficitN}
\end{align}

 %$\mathds{1}(\textit{D}(i,\mathcal{S},z))=1$ if the deficit of node $i$ is less than a threshold value $\frac{|\mathcal S|-z}{z-1}$, i.e., 
2) \emph{Feasible Condition:} For Node $i\in\mathcal{S}$ and $(z,\mathcal{S})$-mapped IVs, define the feasible condition, denoted by $\mathscr{E}_2(i,\mathcal{S},z)$, as follows. The following linear system about \{$\beta_{j,\mathcal
S_1}\in \mathbb{Q}^+:j\in \mathcal S_1,i\not\in \mathcal{S}_1, \mathcal S_1\subseteq \mathcal{S},|\mathcal S_1|=z $\} has a non-negative solution:
\begin{IEEEeqnarray}{rCl}
 \left\{
 \begin{aligned}
  &\sum_{\substack{j\in \mathcal S_1}}\beta_{j,\mathcal S_1} \!\geq\! (|\mathcal{S}|\!-\!1) |\mathcal {V}_{\mathcal {S}_1}^{\mathcal {S} \backslash \mathcal {S}_1}|,\ \forall \mathcal S_1\subseteq \mathcal{S}\!:\!|\mathcal S_1|=z,i\not\in\mathcal{S}_1\\
  &\sum_{\mathcal S_1\in\mathcal{S}:j\in \mathcal S_1,i\not\in \mathcal S_1}\beta_{j,\mathcal S_1} \leq (n_{j,z}^{\mathcal{S}})^+,~~\forall j\in \mathcal{S},j\neq i.
 \end{aligned}
 \right. \label{feasible}
\end{IEEEeqnarray}

 With the deficit and feasible conditions defined above, we now present the upper bound of optimal communication load achieved by FSCT:

 \begin{theorem}\label{ThmFewshot}
 Consider the distributed computing system with $K$ nodes, $N$ files and $Q$ Reduce functions under pre-set data placement $\mathcal{M}$ and Reduce function assignment $\mathcal{W}$ for $k\in [K]$. The optimal communication load is upper bounded by $L^*\leq L_{\textnormal{fsct}}$ where $L_{\textnormal{fsct}}$ is defined as:
 \begin{align}
 &L_{\textnormal{fsct}}=\frac{1}{QN}\sum_{t=1}^{K}\sum_{d=1}^{K-s}a_{t,d}\frac{d}{t+d-1}\notag\\
  &+\frac{1}{QN}\sum_{\mathcal S\subseteq[K]}\sum_{z=r_{\min}}^{|\mathcal S|-1}\frac{1}{|\mathcal{S}|-1}\bigg(\!\!\!\!\!\!\!\!\!\!\sum_{\substack{k\in\mathcal{S}:\\\mathds{1}(\mathscr{E}_1(k,z,\mathcal{S}))=0}}\!\!\!\!\!\!\!\!\!\!\! |n_{k,z}^{\mathcal{S}}|-\!\!\!\!\sum_{\substack{k\in\mathcal{S}:\\\mathds{1}(\mathscr{E}_1(k,z,\mathcal{S}))=1\\\mathds{1}(\mathscr{E}_2(k,z,\mathcal{S}))=0}}\!\!\!\!\!\!\!\!\!\!\! n_{k,z}^{\mathcal{S}}\bigg)\notag\\
  &+\frac{1}{QN}\sum_{\mathcal S\subseteq[K]}\sum_{z=r_{\min}}^{|\mathcal S|-1}\frac{1}{z}\sum_{\substack{i\in\mathcal{S}:\\\mathds{1}(\mathscr{E}_2(i,z,\mathcal{S}))=0}}\max_{j\in \mathcal{S}}\sum_{\substack{\mathcal{S}_1\subseteq \mathcal{S}:|S_1|=z\\j\in \mathcal S_1,i\not\in \mathcal S_1}}|\mathcal {V}_{\mathcal {S}_1}^{\mathcal {S} \backslash \mathcal {S}_1}|,
 \end{align}
 where
$\mathscr{E}_1(i,z,\mathcal{S})$ and $\mathscr{E}_2(i,z,\mathcal{S})$ are the deficit condition and feasible condition defined in \eqref{defiCond} and \eqref{feasible}, respectively.
 \label{th3}
 \end{theorem}

 We present the proposed scheme FSCT in Section \ref{SchFewShot} as proof of Theorem \ref{ThmFewshot}. The reader can also refer to the toy example in Section \ref{SecToyExample} and the example in Section \ref{ExFSCT} for a more intuitive illustration of the design of FSCT and the meaning of corresponding parameters.

 \begin{remark}
Here $n_{i,z}^{\mathcal{S}}$ represents the \emph{designed} number of linear combinations (LCs) transmitted by Node $i$  with regard to the $(z,\mathcal{S})$-mapped IVs. The first term in $L_{\textnormal{fsct}}$ corresponds to the lower bound provided in \eqref{lowerbound}. When some nodes that map fewer IVs require more IVs, the severe heterogeneity characterized by the violation of deficit condition leads to the second term in $L_{\textnormal{fsct}}$. Besides, the designed values of $n_{i,z}^{\mathcal{S}}$ for some $i\in\mathcal{S}$ may be too small to satisfy the feasible condition. To make all the encoded messages sent in FSCT decodable, we need extra communication cost that leads to the third term in $L_{\textnormal{fsct}}$.
 \end{remark}

\begin{theorem}
 If the deficit condition and the feasible condition are both satisfied for each node with $(z,\mathcal{S})$-mapped IVs, i.e., 
\begin{IEEEeqnarray}{rCl}\label{conditionOptThm4}
 \mathds{1}(\mathscr{E}_1(i,z,\mathcal{S}))=\mathds{1}(\mathscr{E}_2(i,z,\mathcal{S}))=1, ~\forall i,\mathcal{S}, z,
\end{IEEEeqnarray}
 then the communication load in Theorem \ref{th3} is optimal.\label{ThmOptFSCT}
\end{theorem}
\begin{IEEEproof}
 See Appendix \ref{ConverseFSCTOpt}.
\end{IEEEproof}

 The following corollary states the optimality of $L_{\textnormal{fsct}}$ in several cases.

\begin{Corollary}\label{FSCTcoro}
 The communication load in Theorem \ref{ThmFewshot} is optimal for the homogeneous system and semi-homogeneous system defined in Definition \ref{homo} and \ref{semi-sys} respectively. For the 3-node system considered in \cite{ThreeWorker}, which allows any data placement $\mathcal{M}$ and symmetric Reduce function assignment $\mathcal{W}$ with $|\mathcal{W}_i|=1,\mathcal{W}_i\cap \mathcal{W}_j=\emptyset$ for all $i\neq j$ and $i,j\in[3]$, the communication load in Theorem \ref{ThmFewshot} is optimal if the pre-set $\mathcal{M}$ and $\mathcal{W}$ are the same as in \cite{ThreeWorker}.
\end{Corollary}

 \begin{IEEEproof}
 See Appendix \ref{AppFSCTCoro}.
 \end{IEEEproof}

 The corollary above shows that our scheme is optimal in several cases including the optimal results in \cite{CDC} and \cite{ThreeWorker}.

\begin{remark}\label{CompareTwoschemes}
The direct comparison between OSCT and FSCT is difficult since in general there is no explicit solution to the optimization problem in OSCT, and the load of FSCT depends on the number of nodes satisfying the deficit and feasible conditions. 
By examples and simulation results in the following sections, we find that FSCT potentially improves OSCT, as the former allows the node to jointly decode multiple message blocks sent by different nodes. However, OSCT holds its merit of lower complexity and can also achieve optimal communication in some scenarios. Specifically, the first example of Section \ref{SecToyExample} shows that the OSCT can achieve the same optimal communication load with FSCT by splitting IVs into fewer pieces.
\end{remark}

\subsection{Comparison with Previous Works}\label{SecComparison}
We compare our OSCT and FSCT with the existing works on the HetDC systems \cite{ThreeWorker}, \cite{Ji}, and \cite{Tao},  from the perspective of applicability, optimality, and complexity, respectively.  The main differences are summarized in Table \ref{compare}.

\begin{table*}[!htbp]
  \caption{Comparison with previous works in \cite{Tao}, \cite{Ji} and \cite{ThreeWorker}}
  \centering\label{compare}
  \begin{tabular}{|l|c|c|c|c|c|} \hline
   & Scheme in \cite{Tao} & Scheme in \cite{Ji} & Scheme in \cite{ThreeWorker} & OSCT & FSCT\\
   \hline
   \tabincell{c}{Pre-set $\{\mathcal{M}_k\}\&\{\mathcal{W}_k\}$} & \tabincell{c}{ non-applicable} & non-applicable & non-applicable & applicable & applicable \\
   \hline
   3-node system & optimal & unknown & optimal & optimal & optimal\\
   \hline
   \tabincell{c}{Homo. system} & optimal & optimal & optimal & optimal & optimal\\
   \hline
   \tabincell{c}{Semi-homo. system} & non-applicable & non-applicable & non-applicable & optimal & optimal\\
   \hline
   Other optimal cases & - & - & - & Eq. \eqref{conditionOptThm2} & Eq. \eqref{conditionOptThm4} \\
   \hline
  \end{tabular}
  \end{table*}

\subsubsection{Applicability} For the general MapReduce system with pre-set data placement and Reduce function assignment, the scheme in \cite{Ji} is not applicable since the data placement may not follow their combinatorial design, in which each file and each Reduce function are required to be mapped or computed exactly $r$ times. The scheme proposed in \cite{Tao} is also not applicable since the Reduce function assignment in \cite{Tao} requires that $\mathcal{W}_i\cap\mathcal{W}_j=\emptyset$, for all $i\neq j$, and their scheme jointly designs data placement and transmission strategy to reduce the communication load. The scheme presented in \cite{ThreeWorker} is only valid for the 3-node system and requires the symmetric Reduce function assignment $\mathcal{W}$ with $|\mathcal{W}_i|=1,\mathcal{W}_i\cap \mathcal{W}_j=\emptyset$ for all $i\neq j$ and $i,j\in[3]$, and thus is infeasible for the general HetDC system. Both OSCT and FSCT are designed according to the pre-set data placement and reduce assignment (see Section \ref{SchOneshot} and \ref{SchFewShot}), thus are applicable for any prefixed HetDC system. Our work is the first attempt to propose feasible schemes for general HetDC systems with arbitrary and pre-set data placement and Reduce function assignment.

 \subsubsection{Optimality} It can be shown that all the schemes in comparison are optimal for homogeneous systems as in this case they are all reduced to the original CDC scheme in \cite{CDC}. For the 3-node system considered in \cite{ThreeWorker}, the optimality has been verified by the schemes except for the scheme in \cite{Ji}.  For the semi-homogeneous systems, Corollary \ref{OSCTCorollary} and \ref{FSCTcoro} show that OSCT and FSCT are both optimal in this case, while other schemes are non-applicable. More specifically, in \cite{Tao} and \cite{ThreeWorker}, they only considered the systems where each Reduce function is computed exactly once, and in \cite{Ji} the design of Reduce function assignment requires that each function is computed $r$ times, and thus can not be applied to the semi-homogeneous system when there exists some batch of Reduce functions $\mathcal{Q}_s \neq \emptyset$ with $s\neq r$.

 \section{One-shot Coded Transmission (Proof of Theorem \ref{ThmOneshot})} \label{SchOneshot} 
 By grouping nodes into clusters with different sizes and carefully categorizing IVs, we can characterize the heterogeneity from the perspective of IVs. The key idea of OSCT is to exploit the multicast opportunity for each set of $(z,\mathcal{S})$-mapped IVs by precisely designing the size of messages sent by every node with the help of parameters $\alpha_{j,z}^{\mathcal{S}}$ obtained by solving an optimization problem. 
More specifically, we first group the nodes into different clusters $\mathcal{S}\subseteq [K]$ of size $\max\{r_{\min}+1,q_{\min}\}\leq |\mathcal{S}|\leq \min\{K,r_{\min}+q_{\min}\}$ and divide the sending process in each cluster $\mathcal{S}$ into multiple rounds, indexed as $z\in\{r_{\min},r_{\min}+1,\cdots,|\mathcal{S}|-1\}$. By solving an optimization problem to maximize the IVs to be multicasted, the $(z,\mathcal{S})$-mapped IVs will be sent in round $z$ by messages of different sizes determined by $\alpha_{j,z}^{\mathcal{S}}$.

 \subsection{Definitions of Sending Clusters and Rounds}
 Note that after the Map phase, Node $k\in [K]$ has obtained its local IVs $\{v_{q,n}:q\in [Q], n\in \mathcal{M}_k\}$, and it requires IVs in $\{v_{q,n}:q\in \mathcal{W}_k, n\notin \mathcal{M}_k\}$. 
 
 Given the pre-set $\mathcal{M}$ and $\mathcal{W}$, first compute the minimum Mapping times $r_{\textnormal{min}} = \min\limits_{1\leq i\leq N}r_i$ and the minimum Reducing times $q_{\textnormal{min}} = \min\limits_{1\leq i\leq Q}q_i$. We call each subset $\mathcal{S}\subseteq [K]$ of size $\max\{r_{\min}+1,q_{\min}\}\leq |\mathcal{S}|\leq \min\{K,r_{\min}+q_{\min}\}$ a \emph{sending cluster}. For each sending cluster $\mathcal{S}$, any file $w_n$ exclusively stored by nodes in $\mathcal{S}$, i.e., $n\in\cup_{k\in\mathcal{S}}\mathcal{M}_k, n\notin \cup_{k\in[K]\backslash\mathcal{S}}\mathcal{M}_k,$ is possibly stored by $z\in\{r_\text{min},r_\text{min}+1, \ldots, |\mathcal{S}|\}$ nodes in $\mathcal{S}$. Thus, we divide the sending process into multiple \emph{rounds}, indexed as $z\in\{r_{\min},r_\text{min}+1,\cdots,|\mathcal{S}|-1\}$. The round $z$ is skipped if there exists no file in $(z,\mathcal{S})$-mapped IVs.

 In round $z$, there are $\binom{|\mathcal{S}|}{z}$ subsets of $\mathcal{S}$ of size $z$, namely, $ \mathcal S_z[1], \mathcal S_z[2], \cdots, \mathcal S_z[\binom{|\mathcal{S}|}{z}]$. According to the definition in \eqref{szmap}, we can write the set of IVs needed by all nodes in $\mathcal S\backslash \mathcal S_z[i]$, no nodes outside $\mathcal S$, and known \emph{exclusively} by nodes in $\mathcal \mathcal S_z[i]$ as 
 $\mathcal {V}_{\mathcal {S}_z[i]}^{\mathcal {S} \backslash \mathcal {S}_z[i]}$ for $i\in \big[\binom{|\mathcal{S}|}{z}\big]$, i.e., 
 \begin{align}
 \mathcal {V}_{\mathcal{S}_z[i]}^{\mathcal {S} \backslash \mathcal{S}_z[i]}&=\{v_{q,n}: q \in \underset {k \in \mathcal {S} \backslash {S}_z[i]}{\cap } {\mathcal{ W}}_{k}, q\notin \underset {k \notin \mathcal {S}}{\cup } {\mathcal{ W}}_{k},  n \in \underset {k \in {S}_z[i]}{\cap } \mathcal {M}_{k}, n\notin \underset {k \notin {S}_z[i]}{\cup } {\mathcal{ M}}_{k}\}. \label{zIVs}
 \end{align}

  Then, the $(z,\mathcal{S})$-mapped IVs in can be written as $\mathcal{V}_{z}^{\mathcal S} = \bigcup_{i=1}^{\binom{|\mathcal{S}|}{z}}\mathcal {V}_{\mathcal S_z[i]}^{\mathcal {S} \backslash \mathcal {S}_z[i]}$. As we will see, all IVs in $\mathcal{V}_{\mathcal {S}_z[i]}^{\mathcal {S} \backslash \mathcal {S}_z[i]}$ will be jointly sent out by the $z$ nodes in $\mathcal{S}_z[i]$ and all IVs in $\mathcal{V}_{z}^{\mathcal S}$ will be sent in round $z$ of cluster $\mathcal{S}$. Since the sizes of $\mathcal {V}_{\mathcal {S}_z[i]}^{\mathcal {S} \backslash \mathcal {S}_z[i]}$ for $i\in[\binom{|\mathcal S|}{z}]$ are not necessarily equal, we can not simply partition each IVs set into $z$ parts and let each node send one part. Instead, we design an optimization-based transmission strategy as follows to determine the size of each message.
 
 \subsection{Transmission Strategy of OSCT}
 In each sending cluster $\mathcal{S}$ and round $z$, introduce a parameter vector $\boldsymbol{\alpha}_z^S=(\alpha^S_{j,z}: j\in\mathcal{S})$ in the descending order of $j$, where $\{\alpha^S_{j,z}\}$ is the solution to the problem in $\mathcal{P}_{\textnormal{OSCT}}(\mathcal{S},z)$ present in Theorem \ref{ThmOneshot}. For convenience, we rewrite $\mathcal{P}_{\textnormal{OSCT}}(\mathcal{S},z)$ here:
 \begin{align}
 & \min_{\alpha_{j,z}^{\mathcal{S}}} \sum_{i=1}^{\binom{|\mathcal S|}{z}}\Big(\left|\mathcal{V}_{\mathcal {S}_z[i]}^{\mathcal {S} \backslash \mathcal {S}_z[i]}\right|-\sum_{j\in\mathcal S_z[i]}\alpha_{j,z}^{\mathcal{S}}\Big)^2 \tag{7a}\\
 &\quad {s.t.}~\alpha_{j,z}^{\mathcal{S}}\in \mathbb{Q}^+, \forall j\in\mathcal S\tag{7b}\\
 &\ \ \ \ \ \ \sum_{j\in \mathcal \mathcal S_z[i]}\alpha_{j,z}^{\mathcal{S}}=0,\ \text{if}\ |\mathcal{V}_{\mathcal {S}_z[i]}^{\mathcal {S} \backslash \mathcal {S}_z[i]}|=0, \forall i \in \left[\tbinom{|\mathcal{S}|}{z}\right]\tag{7c}.
 \end{align} 
 
 Let $$\tau_{i,z}^{\mathcal S} \triangleq \left|\mathcal{V}_{\mathcal {S}_z[i]}^{\mathcal {S} \backslash \mathcal {S}_z[i]}\right|-\sum_{j\in \mathcal S_z[i]}\alpha_{j,z}^{\mathcal S},~\forall i \in [\tbinom{|\mathcal{S}|}{z}]. $$ 
  
 We first concatenate all IVs in $\mathcal{V}_{\mathcal {S}_z[i]}^{\mathcal {S} \backslash \mathcal {S}_z[i]}$ to construct a symbol $U_{\mathcal {S}_z[i]}^{\mathcal {S} \backslash \mathcal {S}_z[i]}$, and then split $U_{\mathcal {S}_z[i]}^{\mathcal {S} \backslash \mathcal {S}_z[i]}$ as follows.
  
 If ${\tau}_{i,z}^{\mathcal S}\leq 0$, concatenate extra $|{\tau}_{i,z}^{\mathcal S}|\cdot V$ bits of 0 to the end of $U_{\mathcal {S}_z[i]}^{\mathcal {S} \backslash \mathcal {S}_z[i]}$, and then split $U_{\mathcal {S}_z[i]}^{\mathcal {S} \backslash \mathcal {S}_z[i]}$ into $z$ segments, i.e.,
 \begin{IEEEeqnarray}{rCl}\label{EqSegments}
  U_{\mathcal {S}_z[i]}^{\mathcal {S} \backslash \mathcal {S}_z[i]}\! =\! \left({U}_{\mathcal {S}_z[i],\sigma_1}^{\mathcal {S} \backslash \mathcal {S}_z[i]},{U}_{\mathcal {S}_z[i],\sigma_2}^{\mathcal {S} \backslash \mathcal {S}_z[i]},\cdots,{U}_{\mathcal {S}_z[i],\sigma_z}^{\mathcal {S} \backslash \mathcal {S}_z[i]}\right),\quad 
 \end{IEEEeqnarray}
 where ${U}_{\mathcal {S}_z[i],\sigma_j}^{\mathcal {S} \backslash \mathcal {S}_z[i]}\in \mathbb{F}_{2^{ (\alpha^{\mathcal{S}}_{\sigma_j,z}) V}}$ and $\sigma_j\in \mathcal{S}_z[i]$, for $j\in[z]$. Note that, unlike the original CDC scheme where each segment has equal size, here each segment ${U}_{\mathcal {S}_z[i],\sigma_j}^{\mathcal {S}}$ is of size $(\alpha^{\mathcal{S}}_{\sigma_j,z}) V$ bits, determined by solution to the optimization problem $\mathcal{P}_{\textnormal{OSCT}}(\mathcal{S},z)$. 
 If ${\tau}_{i,z}^{\mathcal S}> 0$, we directly split $U_{\mathcal {S}_z[i]}^{\mathcal {S} \backslash \mathcal {S}_z[i]}$ (without any 0 concatenated at the end) into $z+1$ segments. The first $z$ segments are the same as \eqref{EqSegments} (i.e., with size of $(\alpha^{\mathcal{S}}_{\sigma_j,z})V$ bits for each $\sigma_j\in \mathcal{S}_z[i]$, $j\in [z]$), and the $(z+1)$-th segment is ${U}_{{\tau}_{i,z}^{\mathcal S}> 0}\in \mathbb{F}_{2^{({\tau}_{i,z}^{\mathcal S})V}}$, i.e.,
 \begin{IEEEeqnarray}{rCl}\label{EqSegments2}
  U_{\mathcal {S}_z[i]}^{\mathcal {S} \backslash \mathcal {S}_z[i]}\! =\! \left({U}_{\mathcal {S}_z[i],\sigma_1}^{\mathcal {S} \backslash \mathcal {S}_z[i]},{U}_{\mathcal {S}_z[i],\sigma_2}^{\mathcal {S} \backslash \mathcal {S}_z[i]},\cdots,{U}_{\mathcal {S}_z[i],\sigma_z}^{\mathcal {S} \backslash \mathcal {S}_z[i]},{U}_{{\tau}_{i,z}^{\mathcal S}> 0}\right).\quad 
 \end{IEEEeqnarray}
 
 Here $\alpha_{\sigma_j,z}^{\mathcal{S}}V$ represents the size of segment ${U}_{\mathcal {S}_z[i],\sigma_j}^{\mathcal {S} \backslash \mathcal {S}_z[i]}$, and thus must be non-negative, leading to the constraint {(7b)} in $\mathcal{P}_{\textnormal{OSCT}}(\mathcal{S},z)$. Besides, if the $|\mathcal{V}_{\mathcal {S}_z[i]}^{\mathcal {S} \backslash \mathcal {S}_z[i]}|=0$, then the size of each splitted segments $|U_{\mathcal {S}_z[i], \sigma_j}^{\mathcal {S} \backslash \mathcal {S}_z[i]}|=0$, so the corresponding $\alpha_{\sigma_j,z}^{\mathcal{S}}$ for $\sigma_j\in\mathcal{S}_z[i]$ is 0, leading to the condition (7c) in $\mathcal{P}_{\textnormal{OSCT}}(\mathcal{S},z)$.
 In other words, this constraint {guarantees} that $z$ nodes in $\mathcal{S}_z[i]$ will jointly map at least one IV, otherwise the corresponding $\alpha_{\sigma_j,z}^{\mathcal{S}}$ for $\sigma_j\in\mathcal{S}_z[i]$ will be set to 0 in this round.

 In the following, we present how to deliver all segments in \eqref{EqSegments} or \eqref{EqSegments2}. 
 \subsubsection{Encoding} 
 If ${\tau}_{i,z}^{\mathcal S}> 0$, the $(z+1)$-th segment ${U}_{{\tau}_{i,z}^{\mathcal S}> 0}\in \mathbb{F}_{2^{(\tau_{i,z}^{\mathcal S})V}}$ in \eqref{EqSegments2} is directly unicasted by an arbitrary node in $\mathcal{S}_{z}[i]$ to the nodes in $\mathcal{S}\backslash\mathcal{S}_z[i]$. 
 
 For segments in \eqref{EqSegments} or the first $z$ segments in \eqref{EqSegments2}, i.e., $U_{\mathcal {S}_z[i], \sigma_j}^{\mathcal {S} \backslash \mathcal {S}_z[i]}$ for $i\in\big[\binom{|\mathcal{S}|}{z}\big]$, $j\in[z]$, Node $k\in \mathcal{S}$ encodes the segment vector $\mathbf{u}_{k,z}^{\mathcal{S}}=\big({U}_{\mathcal S_z[i],k}^{\mathcal {S} \backslash \mathcal S_z[i]}:i\in\big[\tbinom{|\mathcal{S}|}{z}\big], k\in \mathcal S_z[i]\big)$ including $\binom{|\mathcal S|-1}{z-1}$ segments  each of size $(\alpha_{k,z}^{\mathcal{S}}) V$ bits, into $\binom{|\mathcal{S}|-2}{z-1}$ numbers of LCs. Specifically, Node $k$ multicasts a message block $\big(X_{k,z}^{\mathcal{ S}}[{1}],X_{k,z}^{\mathcal{ S}}[{2}],\cdots,X_{k,z}^{\mathcal{ S}}[{\binom{|\mathcal S|-2}{z-1}}]\big)$ as follows:
 \begin{align} 
 \begin{bmatrix} 
 X_{k,z}^{\mathcal{ S}}[1] \\ 
 X_{k,z}^{\mathcal{ S}}[2] \\ 
 \vdots \\ 
 X_{k,z}^{\mathcal{ S}}[\binom{|\mathcal S|-2}{z-1}] 
 \end{bmatrix} 
 \!=\! 
 {\mathbf V}_z\cdot
 \mathbf{u}_{k,z}^{\mathcal{S}}, \label{block}
 \end{align} 
 where 
 \begin{align*}
 {\mathbf V}_z=\begin{bmatrix}\! 
  1 & 1 & \cdots & 1 \\ 
  v_1 & v_2 & \cdots &v_{\binom{|\mathcal S|-1}{z-1}} \\ 
  \vdots & \vdots & \ddots & \vdots \\ 
  v_1^{\binom{|\mathcal S|-2}{z-1}-1} & v_2^{\binom{|\mathcal S|-2}{z-1}-1} & \cdots & v_{\binom{|\mathcal S|-1}{z-1}}^{\binom{|\mathcal S|-2}{z-1}-1} \!
  \end{bmatrix}
 \end{align*}
 is a Vandermonde matrix\footnote{The Vandermonde matrix $\textbf{V}_z$ in sending clusters $\mathcal{S}$ and $\mathcal{S'}$ can be the same if $|\mathcal S|=|\mathcal S'|$.} for some coefficients $v_j\in \mathbb{F}_{2^{(\alpha^{\mathcal{S}}_{k,z})V}} $, $j\in[\binom{|\mathcal S|-1}{z-1}]$.
 
 Apply the above encoding process to send all IVs in $\bigcup_{z=r_{\min}}^{|\mathcal{S}|-1}\mathcal{V}_z^{\mathcal{S}}$ in sending cluster $\mathcal{S}\subseteq [K]$, where $\mathcal{V}_z^{\mathcal{S}}$ is the $(z,\mathcal{S})$-mapped IVs defined in \eqref{z-mapped}.
 
 \subsubsection{Decoding}
 
 In each sending cluster $\mathcal{S}$ and round $z$, after receiving the message block from Node $k$: $\big(X_{k,z}^{\mathcal{ S}}[{1}],X_{k,z}^{\mathcal{ S}}[{2}],\cdots,$ $X_{k,z}^{\mathcal{ S}}[{\binom{|\mathcal S|-2}{z-1}}]\big)$, Node $j\in\mathcal{S}$ with $j\neq k$ decodes its desired segments $\big\{{U}_{\mathcal{S}_{z}[i],k}^{\mathcal {S} \backslash \mathcal{S}_{z}[i]}: i\in\big[\binom{|\mathcal S|}{z}\big], k\in \mathcal{S}_{z}[i], j\notin \mathcal{S}_{z}[i]\big\}$ based on the local segments $\big\{{U}_{\mathcal{S}_{z}[i],k}^{\mathcal {S} \backslash \mathcal{S}_{z}[i]}: i\in\big[\binom{|\mathcal S|}{z}\big], k, j\in \mathcal{S}_{z}[i]\big\}$. Following the similar steps in \cite{CDC}, Node $j$ first removes the local segments from the received LCs and then decodes the desired segments by multiplying an invertible Vandermonde matrix. 
 
 \subsubsection{Correctness of the Scheme} Now we show that each Node $j\in\mathcal{S}\backslash \{k\}$ can decode its desired segments after receiving the block message sent from Node $k$: $\big(X_{k,z}^{\mathcal{ S}}[{1}],X_{k,z}^{\mathcal{ S}}[{2}],\cdots,X_{k,z}^{\mathcal{ S}}[{\binom{|\mathcal S|-2}{z-1}}]\big)$. Among all segments generating $\big(X_{k,z}^{\mathcal{ S}}[{1}],X_{k,z}^{\mathcal{ S}}[{2}],\cdots,$ $X_{k,z}^{\mathcal{ S}}[{\binom{|\mathcal S|-2}{z-1}}]\big)$, i.e., $\big\{{U}_{\mathcal{S}_{z}[i],k}^{\mathcal {S} \backslash \mathcal{S}_{z}[i]}: i\in\big[\binom{|\mathcal S|}{z}\big], k\in\mathcal{S}_z[i]\big\}$, there are $\binom{|\mathcal S|-2}{z-2}$ segments $\big\{{U}_{\mathcal{S}_{z}[i],k}^{\mathcal {S} \backslash \mathcal{S}_{z}[i]}: i\in\big[\binom{|\mathcal S|}{z}\big], k,j\in \mathcal{S}_{z}[i]\big\}$ already known by both Node $k$ and $j$, and $\binom{|\mathcal S|-1}{z-1}-\binom{|\mathcal S|-2}{z-2}=\binom{|\mathcal S|-2}{z-1}$ segments $\big\{{U}_{\mathcal{S}_{z}[i],k}^{\mathcal {S} \backslash \mathcal{S}_{z}[i]}: i\in\big[\binom{|\mathcal S|}{z}\big], k\in \mathcal{S}_{z}[i],j\notin \mathcal{S}_{z}[i]\big\}$ unknown by Node $j$. Constraint (7c) guarantees that all segments multicasted by Node $k$ and known by Node $j$ are not empty, i.e., $|{U}_{\mathcal{S}_{z}[i],k}^{\mathcal {S} \backslash \mathcal{S}_{z}[i]}| \neq 0$ for $i\in\big[\binom{|\mathcal S|}{z}\big], k,j\in \mathcal{S}_{z}[i]$. Thus, from $\binom{|\mathcal S|-2}{z-1}$ linearly independent combinations, Node $j$ can recover all $\binom{|\mathcal S|-2}{z-1}$ desired segments.
 
 \subsubsection{Overall Communication Load} In round $z$ of cluster $\mathcal{S}$, Node $k$ multicasts $\binom{|\mathcal S|-2}{z-1}$ numbers of LCs, each of size $\alpha_{k,z}^{\mathcal{S}}V$ bits. Additionally, the segment ${U}_{{\tau}_{i,z}^{\mathcal S}> 0}$, for $i\in\binom{|\mathcal{S}|}{z}$, of size ${\tau}_{i,z}^{\mathcal S}V$ bits is unicasted in this round if ${{\tau}_{i,z}^{\mathcal S}> 0}$. So the communication cost in round $z$ of cluster $\mathcal{S}$ is:
 \begin{align}
 L_{z,\mathcal{S}}=\sum_{k\in\mathcal{S}}\tbinom{|\mathcal S|-2}{z-1}\alpha_{k,z}^{\mathcal{S}}V+\sum_{i=1}^{\binom{|\mathcal{S}|}{z}}(\tau_{i,z}^{\mathcal{S}})^+V.\label{roundcost}
 \end{align}
 
 By considering all cluster $\mathcal{S}\subseteq [K]$ and all rounds $r_{\min}\leq z\leq |\mathcal{S}|-1$, and according to Definition \ref{Comload}, we obtain the overall communication load (normalized by $QNV$) of OSCT 
 \begin{align}
 L_{\textnormal{OSCT}}&=\sum_{\mathcal{S}\subset[K]}\sum_{z=r_{\min}}^{|\mathcal{S}|-1}\frac{ L_{z,\mathcal{S}}}{QNV}= \frac{1}{QN}{\sum\limits_{\mathcal{S}\subseteq[K]}\sum\limits_{z=r_{\min}}^{|\mathcal{S}|-1} \Big(\sum\limits_{k\in\mathcal{S}}\tbinom{|\mathcal S|-2}{z-1}\alpha_{k,z}^{\mathcal{S}}\! +\!\sum\limits_{i=1}^{\binom{|\mathcal{S}|}{z}}(\tau_{i,z}^{\mathcal{S}})^+\Big)}.
 \end{align}
 
 We call the approach above \emph{One-shot} Coded Transmission, since Node $k\in \mathcal{S}$ can decode its desired segments once Node $k$ receives a single message block $\big(X_{j,z}^{\mathcal{ S}}[{1}],X_{j,z}^{\mathcal{ S}}[{2}],\cdots,X_{j,z}^{\mathcal{ S}}[{\binom{|\mathcal S|-2}{z-1}}]\big)$. 
 
 \subsection{Complexity of OSCT}\label{ComplexityOSCT}
 For the OSCT strategy, consider the encoding and decoding processes of each node in round $z$ of cluster $\mathcal{S}$. As the multiplication of a $n\times n$ Vandermonde matrix and a vector can be done in $O\left(n\log^2 n\right)$ time\cite{VMV}, the time complexity of encoding process is $O\big(\binom{|\mathcal{S}|-1}{z-1}\log^2\binom{|\mathcal S|-1}{z-1}\big)$ for each node in round $z$. There are $|\mathcal{S}|$ nodes in the cluster, and we sum the time complexity in every round and every cluster, then we have the total time complexity of the encoding process as 
 \begin{align}
 O\bigg(\sum_{\mathcal{S}\subseteq [K]}\sum_{z=r_{\min}}^{|\mathcal S|-1}|\mathcal{S}|\tbinom{|\mathcal{S}|-1}{z-1}\log^2\tbinom{|\mathcal S|-1}{z-1}\bigg).
 \end{align}
 Similarly, as each node receives $|\mathcal{S}|-1$ message blocks from other nodes in round $z$ of cluster $\mathcal S$ and solving a linear system with $n$ equations requires $O(n^3)$ time\cite{IntroAlg}, the time complexity of decoding process is $\big((|\mathcal{S}|-1)\binom{|\mathcal{S}|-2}{z-1}^3\big)$ for each node in round $z$. Totally, the complexity of the decoding process is \begin{align}
 O\bigg(\sum_{\mathcal{S}\subseteq [K]}\sum_{z=r_{\min}}^{|\mathcal S|-1}|\mathcal{S}|(|\mathcal{S}|-1)\tbinom{|\mathcal{S}|-2}{z-1}^3\bigg).
 \end{align}

 \section{Few-shot Coded Transmission (Proof of Theorem \ref{ThmFewshot})}\label{SchFewShot}
The key idea of FSCT is to further exploit the multicast gain by delicately designing the number of linear combinations (LCs) to be transmitted for each set of $(z,\mathcal{S})$-mapped IVs, based on the introduced parameters $n_{j,z}^{\mathcal{S}}$, deficit condition, and feasible condition. Note that each node needs to jointly decode multiple message blocks, leading to a potentially lower communication load at the cost of higher coding complexity.

  Similar to OSCT, we group the nodes into different clusters $\mathcal{S}\subseteq [K]$ of size $\max\{r_{\min}+1,q_{\min}\}\leq |\mathcal{S}|\leq \min\{K,r_{\min}+q_{\min}\}$ and divide the sending process in each cluster $\mathcal{S}$ into multiple rounds, indexed as $z\in\{r_{\min},r_{\min}+1,\cdots,|\mathcal{S}|-1\}$. The intermediate values in $\mathcal{V}_{\mathcal {S}_z[i]}^{\mathcal {S}\backslash \mathcal {S}_z[i]},$ for all $ i\in\big[\binom{|\mathcal{S}|}{z}\big]$ defined in \eqref{zIVs} will be sent in round $z$.
  
  \subsection{Transmission Strategy of FSCT}
 
   \subsubsection{Parameter update based on feasible condition}\label{update}
   Recall the parameters $n_{k,z}^S$ in \eqref{nAndBeta} 
   \begin{IEEEeqnarray}{rCl}
    n_{k,z}^{\mathcal S}=\!\!\!\!\!\!\!\!\!\sum\limits_{\substack{S_1\subseteq \mathcal S:\\|S_1|=z,k\in S_1}}\!\!\!\!\!\!(|\mathcal S|\!\!-\!\!z)|\mathcal {V}_{\mathcal {S}_1}^{\mathcal {S} \backslash \mathcal {S}_1}|\!-\!\!\!\!\sum\limits_{\substack{S_1\subseteq \mathcal S:\\|S_1|=z,k\not\in S_1}}\!\!\!\!\!\!(z-1)|\mathcal {V}_{\mathcal {S}_1}^{\mathcal {S} \backslash \mathcal {S}_1}|. \label{nks}
   \end{IEEEeqnarray}
 
   Given all the $n_{k,z}^{\mathcal{S}}$ and $\mathcal{V}_{\mathcal {S}_z[i]}^{\mathcal {S}\backslash \mathcal {S}_z[i]}$ in round $z$, if the feasible condition in \eqref{feasible} is not satisfied for some nodes in $ \mathcal{S}$, then set 
   \begin{align}
    \bar{n}_{k,z}^{\mathcal{S}}=\begin{cases} (n_{k,z}^S)^+,\ \ \ \ \ \ \ \ \ \text{if}\ \mathds{1}(\mathscr{E}_2(i,z,\mathcal{S}))=1,\forall i\in \mathcal{S}, \\
    \!\!\frac{|S|-1}{z}\max\limits_{j\in \mathcal{S}, j\neq k}\left\{\!\!\sum\limits_{\substack{\mathcal{S}_1\subseteq \mathcal{S}:|S_1|=z\\j\in S_1,k\not\in S_1}}\!\!|\mathcal {V}_{\mathcal {S}_1}^{\mathcal {S} \backslash \mathcal {S}_1}|\!\right\},\ \text{otherwise,}
    \end{cases}
   \end{align}
  for each $k\in\mathcal{S}$ in round $z$.
   
  Note that the updated values of $\{\bar{n}_{k,z}^{\mathcal{S}}\}$ always satisfy the feasible condition because $\{\beta_{i,\mathcal S_1}=\frac{|\mathcal{S}|-1}{z}|\mathcal {V}_{\mathcal {S}_1}^{\mathcal {S} \backslash \mathcal {S}_1}|:\mathcal{S}_1\subseteq\mathcal{S}, |\mathcal{S}_1|=z, i\in\mathcal{S}_1\}$ is a non-negative solution set of \eqref{feasible}. As we will see later, $\bar{n}_{k,z}^S$ is the number of LCs sent by Node $k$ with respect to the $(z,\mathcal{S})$-mapped IVs in FSCT scheme. \label{AlwaysSatisfy}
  
  \subsubsection{Encoding}
  
  For each $i\in[\binom{|\mathcal{S}|}{z}]$, we concatenate all IVs in $\mathcal{V}_{\mathcal {S}_z[i]}^{\mathcal {S}\backslash \mathcal {S}_z[i]}$ to construct a symbol $U_{\mathcal {S}_z[i]}^{\mathcal {S}\backslash \mathcal {S}_z[i]}$, and evenly split $U_{\mathcal {S}_z[i]}^{\mathcal {S}\backslash \mathcal {S}_z[i]}$ into $(|\mathcal S|-1)\big|\mathcal{V}_{\mathcal {S}_z[i]}^{\mathcal {S}\backslash \mathcal {S}_z[i]}\big|$ segments each of size $\frac{V}{|\mathcal{S}|-1}$, i.e., 
   \begin{align}
   U_{\mathcal {S}_z[i]}^{\mathcal {S}\backslash \mathcal {S}_z[i]}=\left({U}_{\mathcal {S}_z[i]}^{\mathcal {S}\backslash \mathcal {S}_z[i]}[1],\cdots,{U}_{\mathcal {S}_z[i]}^{\mathcal {S}\backslash \mathcal {S}_z[i]}\left[(\!|\mathcal S|\!-\!1)\!\left|\!\mathcal{V}_{\mathcal {S}_z[i]}^{\mathcal {S}\backslash \mathcal {S}_z[i]}\!\right|\right]\right),
   \end{align}
   where  ${U}_{\mathcal {S}_z[i]}^{\mathcal {S} \backslash \mathcal {S}_z[i]}[j]\in \mathbb{F}_{2^{\frac{V}{|\mathcal{S}|-1}}}$, for all $j\in\big[(\!|\mathcal S|\!-\!1)\!\big|\!\mathcal{V}_{\mathcal {S}_z[i]}^{\mathcal {S}\backslash \mathcal {S}_z[i]}\!\big|\big]$.
  
   Node $k\in \mathcal{S}$ has mapped the IVs in $\mathcal{V}_{\mathcal{S}_{z}[i]}^{\mathcal {S}\backslash \mathcal{S}_{z}[i]}$ if $k\in\mathcal{S}_{z}[i]$, $i\in \big[\binom{|\mathcal{S}|}{z}\big]$, indicating that it knows all IVs segments in 
  $\mathcal K_{k,z}^S\triangleq \Big\{{U}_{\mathcal{S}_{z}[i]}^{\mathcal {S}\backslash \mathcal{S}_{z}[i]}[\gamma]\!:\!i\!\in \! \big[\binom{|\mathcal{S}|}{z}\big], \!k\in\mathcal{S}_{z}[i],\gamma\!\in\! \big[(|\mathcal S|\!-\!1)\big|\mathcal{V}_{\mathcal{S}_{z}[i]}^{\mathcal {S}\backslash \mathcal{S}_{z}[i]}\big|\big] \Big\}.$
  Index the segments in $\mathcal K_{k,z}^S$ as $\big\{U_{k,z}^{\mathcal S}[1],\cdots,U_{k,z}^{\mathcal S}[|\mathcal K_{k,z}^S|]\big\}$, where 
  \begin{align}
  |\mathcal K_{k,z}^S|=(|\mathcal S|-1)\sum_{\!i\in \left[\binom{|\mathcal{S}|}{z}\right], k\in\mathcal{S}_{z}[i]}\bigg|\mathcal{V}_{\mathcal{S}_{z}[i]}^{\mathcal {S}\backslash \mathcal{S}_{z}[i]}\bigg|.\label{nk}
  \end{align}
  
  Then Node $k$ encodes all the segments in $\mathcal K_{k,z}^{\mathcal{S}}$ into a message block $\left(X_{k,z}^{\mathcal{ S}}[{1}],X_{k,z}^{\mathcal{ S}}[{2}],\cdots,X_{k,z}^{\mathcal{ S}}[{\bar{n}_{k,z}^{\mathcal{S}}}]\right)$ as follows. 
  \begin{align} 
  \begin{bmatrix} 
  X_{k,z}^{\mathcal{ S}}[1] \\ 
  X_{k,z}^{\mathcal{ S}}[2] \\ 
  \vdots \\ 
  X_{k,z}^{\mathcal{ S}}[\bar{n}_{k,z}^{\mathcal{S}}] 
  \end{bmatrix} 
  \!=\! 
  {{\mathbf R}_{z}} \cdot\!
   \begin{bmatrix}\! 
   U_{k,z}^{\mathcal S}[1] \\ 
   U_{k,z}^{\mathcal S}[2] \\ 
   \vdots\\ 
   U_{k,z}^{\mathcal S}[|\mathcal K_{k,z}^S|] \!
   \end{bmatrix}\!,\!\!\! \label{FSCTblock}
  \end{align} 
  where ${{\mathbf R}_z}$ is a $\bar{n}_{k,z}^{\mathcal{S}}\times |\mathcal K_{k,z}^{\mathcal{S}}|$ matrix with randomly generated elements. In round $z$, each node $k\in\mathcal{S}$ sends the message block in \eqref{FSCTblock}.
  
  \subsubsection{Decoding}
  In round $z$, after receiving the message blocks from other nodes in $\mathcal{S}$, each receiver $j\in\mathcal{S}$ decodes its desired segments $\Big\{{U}_{\mathcal{S}_{z}[i]}^{\mathcal {S} \backslash \mathcal{S}_{z}[i]}[\gamma]: i\in\big[\binom{|\mathcal S|}{z}\big], j\notin \mathcal{S}_{z}[i], \gamma\!\in\! \big[(|\mathcal S|\!-\!1)\big|\mathcal{V}_{\mathcal{S}_{z}[i]}^{\mathcal {S}\backslash \mathcal{S}_{z}[i]}\!\big|\big] \Big\}$ from all the received message blocks in this round $\{X_{k,z}^{\mathcal{ S}}[{i}]:k\in\mathcal{S}\backslash \{j\}, i\in[\bar{n}_{k,z}^{\mathcal{S}}]\}$. Node $j$ first removes the local segments $\Big\{{U}_{\mathcal{S}_{z}[i]}^{\mathcal {S} \backslash \mathcal{S}_{z}[i]}[\gamma]: i\in\big[\binom{|\mathcal S|}{z}\big], j\in \mathcal{S}_{z}[i],\gamma\!\in\! \big[(|\mathcal S|\!-\!1)\big|\mathcal{V}_{\mathcal{S}_{z}[i]}^{\mathcal {S}\backslash \mathcal{S}_{z}[i]}\!\big|\big] \Big\}$
  from the message to form a new linear system with $(|\mathcal S|-1) \sum\limits_{i\in\big[\binom{|\mathcal S|}{z}\big], j\notin \mathcal{S}_{z}[i]}\big|{\mathcal{V}}_{\mathcal{S}_{z}[i]}^{\mathcal {S} \backslash \mathcal{S}_{z}[i]}\big|$ unknown variables and $\sum_{k\in\mathcal{S}\backslash j}\bar{n}_{k,z}^{\mathcal{S}}$ numbers of LCs, and then obtains its desired IVs by solving this linear system. In Appendix \ref{CorrectFSCT}, we prove that Node $j$ can successfully solve this linear system and obtain the desired IVs.
  
  \subsubsection{Overall Communication Load} 
  First consider the nodes which satisfy both feasible condition and deficit condition on the $(z,\mathcal{S})$-mapped IVs, i.e., $\{k\in\mathcal{S}:\mathds{1}(\mathscr{E}_1(k,\mathcal{S},z))=\mathds{1}(\mathscr{E}_2(k,\mathcal{S},z))=1\}$. For  each of such node, say  Node $k$, it multicasts $\bar{n}_{k,z}^{\mathcal{S}}=(n_{k,z}^{\mathcal{S}})^+=n_{k,z}^{\mathcal{S}}$ numbers of LCs of each size $\frac{V}{|\mathcal{S}|-1}$-bits, the communication load incurred by those nodes is 
  \begin{align}
   \ell_1 &=\frac{V}{|\mathcal{S}|-1}\sum_{\substack{k\in\mathcal{S}: \mathds{1}(\mathscr{E}_1(k,z,\mathcal{S}))=1\\\quad \mathds{1}(\mathscr{E}_2(k,z,\mathcal{S}))=1}} n_{k,z}^{\mathcal{S}}.
  \end{align}
  
  Now consider the nodes that satisfy feasible condition but not the deficit condition, i.e., $\{k\in\mathcal{S}:\mathds{1}(\mathscr{E}_1(k,\mathcal{S},z))=0, \mathds{1}(\mathscr{E}_2(k,\mathcal{S},z))=1\}$. The communication load incurred by those nodes is zero. This is because the   deficit ratios of these nodes are greater than $\frac{|\mathcal S|-z}{z-1}$, which implies $n_{k,z}^S<0$ and $\bar{n}_{k,z}^{\mathcal{S}}=(n_{k,z}^S)^+ =0$ for all $k$ that satisfy feasible condition but not the deficit condition.
  
  Finally, consider the nodes that do not satisfy the feasible condition, i.e., $\{i\in\mathcal{S}:\mathds{1}(\mathscr{E}_2(i,\mathcal{S},z))=0\}$. Since each node multicasts the updated $\bar{n}_{k,z}^{\mathcal{S}}$ LCs of size $\frac{V}{|\mathcal{S}|-1}$-bits, the communication load incurred by these nodes can be computed as:
  \begin{align}
  \ell_2 &=\frac{V}{|\mathcal{S}|-1}\sum_{\substack{k\in\mathcal{S}:\\\mathds{1}(\mathscr{E}_2(k,z,\mathcal{S}))=0}} \bar{n}_{k,z}^{\mathcal{S}}=\frac{1}{z}\sum_{\substack{k\in\mathcal{S}:\\\mathds{1}(\mathscr{E}_2(k,z,\mathcal{S}))=0}} \max_{j\in \mathcal{S}}\sum_{\substack{\mathcal{S}_1\subseteq \mathcal{S}:|S_1|=z\\j\in S_1,k\not\in S_1}}|\mathcal {V}_{\mathcal {S}_1}^{\mathcal {S} \backslash \mathcal {S}_1}|\cdot V.\label{lbits}
  \end{align}
  
  Thus, the overall communication bits in round $z$ of cluster $\mathcal{S}$ can be computed as:
  \begin{align}
  L(&z,\mathcal{S})=\ell_1+\ell_2\notag\\
  &=\frac{V}{|\mathcal{S}|-1}\sum_{\substack{k\in\mathcal{S}:\\\mathds{1}(\mathscr{E}_1(k,z,\mathcal{S}))=1\\ \mathds{1}(\mathscr{E}_2(k,z,\mathcal{S}))=1}} n_{k,z}^{\mathcal{S}} + \ell_2\notag\\
  &=\frac{1}{|\mathcal{S}|-1}\bigg(\sum_{k\in\mathcal S} n_{k,z}^{\mathcal{S}}\!-\!\!\!\!\!\!\!\!\!\!\sum_{\substack{k\in\mathcal{S}:\\\mathds{1}(\mathscr{E}_1(k,z,\mathcal{S}))=0}}\!\!\!\!\!\!\!\!\!\!\! n_{k,z}^{\mathcal{S}}-\!\!\!\!\!\sum_{\substack{k\in\mathcal{S}:\\\mathds{1}(\mathscr{E}_1(k,z,\mathcal{S}))=1\\\mathds{1}(\mathscr{E}_2(k,z,\mathcal{S}))=0}}\!\!\!\!\!\!\!\!\!\!\! n_{k,z}^{\mathcal{S}}\bigg)\cdot V + \ell_2\notag\\
  &\overset{(a)}{=}\frac{|\mathcal{S}|-z}{|\mathcal{S}|-1}\sum_{i=1}^{\binom{|\mathcal{S}|}{z}}|\mathcal{V}_{\mathcal{S}_z[i]}^{\mathcal{S}\backslash\mathcal{S}_z[i]}|\cdot V + \frac{1}{|\mathcal{S}|-1}\bigg(\!\!\!\!\!\!\!\!\!\!\sum_{\substack{k\in\mathcal{S}:\\\mathds{1}(\mathscr{E}_1(k,z,\mathcal{S}))=0}}\!\!\!\!\!\!\!\!\!\!\! |n_{k,z}^{\mathcal{S}}|-\!\!\!\!\sum_{\substack{k\in\mathcal{S}:\\\mathds{1}(\mathscr{E}_1(k,z,\mathcal{S}))=1\\\mathds{1}(\mathscr{E}_2(k,z,\mathcal{S}))=0}}\!\!\!\!\!\!\!\!\!\!\! n_{k,z}^{\mathcal{S}}\bigg)\cdot V + \ell_2, \label{fLzS}
  \end{align}
  where (a) holds by the following equalities:
  \begin{align}
   \sum_{k\in\mathcal{S}}n_{k,z}^{\mathcal{S}} &\overset{(a)}{=} \sum_{k\in\mathcal{S}}\Biggl((|\mathcal S|\!\!-\!\!z)\!\!\!\!\sum_{\!i\in \left[\binom{|\mathcal{S}|}{z}\right], k\in\mathcal{S}_{z}[i]}\bigg|\mathcal{V}_{\mathcal{S}_{z}[i]}^{\mathcal {S}\backslash \mathcal{S}_{z}[i]}\bigg|-(z-1)\!\!\!\!\sum_{\!i\in \left[\binom{|\mathcal{S}|}{z}\right], k\notin\mathcal{S}_{z}[i]}\bigg|\mathcal{V}_{\mathcal{S}_{z}[i]}^{\mathcal {S}\backslash \mathcal{S}_{z}[i]}\bigg|\Biggl)\notag
   \label{b}\\
   &\!\!\!\!\!\!\!\!\!\!\!\overset{(b)}{=} (|\mathcal{S}|-z)\bigg(z\sum_{i=1}^{\binom{|\mathcal{S}|}{z}}|\mathcal{V}_{\mathcal{S}_z[i]}^{\mathcal{S}\backslash\mathcal{S}_z[i]}| - (z-1) \sum_{i=1}^{\binom{|\mathcal{S}|}{z}}|\mathcal{V}_{\mathcal{S}_z[i]}^{\mathcal{S}\backslash\mathcal{S}_z[i]}|\bigg)\notag\\
   &\!\!\!\!\!\!\!\!\!\!\!= (|\mathcal{S}|-z)\sum_{i=1}^{\binom{|\mathcal{S}|}{z}}|\mathcal{V}_{\mathcal{S}_z[i]}^{\mathcal{S}\backslash\mathcal{S}_z[i]}|,
  \end{align}
   where (a) follows from the definition of $n_{k,z}^{\mathcal{S}}$ in \eqref{nks},  and (b) holds because
   \begin{IEEEeqnarray}{rCl}
   \sum_{k\in\mathcal{S}}\sum_{\!i\in \left[\binom{|\mathcal{S}|}{z}\right], k\in\mathcal{S}_{z}[i]}\bigg|\mathcal{V}_{\mathcal{S}_{z}[i]}^{\mathcal {S}\backslash \mathcal{S}_{z}[i]}\bigg| &=& z\sum_{i=1}^{\binom{|\mathcal{S}|}{z}}|\mathcal{V}_{\mathcal{S}_z[i]}^{\mathcal{S}\backslash\mathcal{S}_z[i]}|, \label{sumsum}\\
   \sum_{k\in\mathcal{S}}\sum_{\!i\in \left[\binom{|\mathcal{S}|}{z}\right], k\notin\mathcal{S}_{z}[i]}\bigg|\mathcal{V}_{\mathcal{S}_{z}[i]}^{\mathcal {S}\backslash \mathcal{S}_{z}[i]}\bigg|&=& (|\mathcal{S}|-z)\sum_{i=1}^{\binom{|\mathcal{S}|}{z}}|\mathcal{V}_{\mathcal{S}_z[i]}^{\mathcal{S}\backslash\mathcal{S}_z[i]}|.\label{sumsumb}
   \end{IEEEeqnarray}
   Specifically, for each node $k\in\mathcal{S}$, since each subset $\mathcal S_z[i]\subseteq\mathcal{S}$ with $k\in\mathcal{S}_z[i]$ contains $z$ nodes, each IV in $\mathcal {V}_{\mathcal{S}_z[i]}^{\mathcal {S} \backslash \mathcal{S}_z[i]}$ is stored by $z$ nodes in $S_z[i]$. By summing all nodes in $\mathcal{S}$, we obtain the right side of \eqref{sumsum}. Following a similar way, we obtain the right side of \eqref{sumsumb}.
  
   Considering all clusters $\mathcal{S}\subset [K]$ and all rounds $z$ ($r_{\min}\leq z\leq |\mathcal{S}|-1$), we obtain the overall communication load (normalized by $QNV$) of FSCT as 
   \begin{align}
   &L_{\textnormal{fsct}} = \frac{1}{QNV}\sum_{\mathcal{S}\subseteq[K]}\sum_{z=r_{\min}}^{|\mathcal{S}|-1}L({z,\mathcal{S}})\notag\\ 
   &\overset{(a)}{=} \frac{1}{QN}\sum_{\mathcal{S}\subseteq[K]}\sum_{z=r_{\min}}^{|\mathcal{S}|-1}\frac{|\mathcal{S}|-z}{|\mathcal{S}|-1}\sum_{i=1}^{\binom{|\mathcal{S}|}{z}}|\mathcal{V}_{\mathcal{S}_z[i]}^{\mathcal{S}\backslash\mathcal{S}_z[i]}|\notag\\
   &\!+\!\frac{1}{QN}\sum_{\mathcal{S}\subseteq[K]}\sum_{z=r_{\min}}^{|\mathcal{S}|-1}\frac{1}{|\mathcal{S}|-1}\bigg(\!\!\!\!\!\!\!\!\!\!\sum_{\substack{k\in\mathcal{S}:\\\mathds{1}(\mathscr{E}_1(k,z,\mathcal{S}))=0}}\!\!\!\!\!\!\!\!\!\!\! |n_{k,z}^{\mathcal{S}}|-\!\!\!\!\sum_{\substack{k\in\mathcal{S}:\\\mathds{1}(\mathscr{E}_1(k,z,\mathcal{S}))=1\\\mathds{1}(\mathscr{E}_2(k,z,\mathcal{S}))=0}}\!\!\!\!\!\!\!\!\!\!\! n_{k,z}^{\mathcal{S}}\bigg)\notag\\
   &\!+\!\frac{1}{QN}\sum_{\mathcal{S}\subseteq[K]}\sum_{z=r_{\min}}^{|\mathcal{S}|-1}\frac{1}{z}\sum_{\substack{k\in\mathcal{S}:\\\mathds{1}(\mathscr{E}_2(k,z,\mathcal{S}))=0}} \max_{j\in \mathcal{S}}\sum_{\substack{\mathcal{S}_1\subseteq \mathcal{S}:|S_1|=z\\j\in S_1,k\not\in S_1}}|\mathcal {V}_{\mathcal {S}_1}^{\mathcal {S} \backslash \mathcal {S}_1}|\label{followprove}\\
   &\overset{(b)}{=} \frac{1}{QN}\sum_{t=1}^{K}\sum_{d=1}^{K-s}a_{t,d}\frac{d}{t+d-1}\notag\\
   &+\frac{1}{QN}\sum_{S\subseteq[K]}\sum_{z=r_{\min}}^{|\mathcal S|-1}\frac{1}{|\mathcal{S}|-1}\bigg(\!\!\!\!\!\!\!\!\!\!\sum_{\substack{k\in\mathcal{S}:\\\mathds{1}(\mathscr{E}_1(k,z,\mathcal{S}))=0}}\!\!\!\!\!\!\!\!\!\!\! |n_{k,z}^{\mathcal{S}}|-\!\!\!\!\sum_{\substack{k\in\mathcal{S}:\\\mathds{1}(\mathscr{E}_1(k,z,\mathcal{S}))=1\\\mathds{1}(\mathscr{E}_2(k,z,\mathcal{S}))=0}}\!\!\!\!\!\!\!\!\!\!\! n_{k,z}^{\mathcal{S}}\bigg)\notag\\
   &+\frac{1}{QN}\sum_{S\subseteq[K]}\sum_{z=r_{\min}}^{|\mathcal S|-1}\frac{1}{z}\sum_{\substack{i\in\mathcal{S}:\\\mathds{1}(\mathscr{E}_2(i,z,\mathcal{S}))=0}}\max_{j\in \mathcal{S}}\sum_{\substack{\mathcal{S}_1\subseteq \mathcal{S}:|S_1|=z\\j\in S_1,i\not\in S_1}}|\mathcal {V}_{\mathcal {S}_1}^{\mathcal {S} \backslash \mathcal {S}_1}|,\label{loadFSCTcompu}
   \end{align} 
  where (a) holds by combining \eqref{lbits} and \eqref{fLzS}, and (b) follows from the definition of $a_{t,d}$ in Lemma 1 and by letting $|\mathcal{S}|=\ell=t+d$ and $z=t$. Specifically, by the definition of $a_{t,d}$,
   \begin{IEEEeqnarray}{rCl}
  &&\sum\limits_{\substack{\mathcal{S}\subseteq [K], |\mathcal{S}|=\ell}}\sum_{i=1}^{\binom{\ell}{t}}|\mathcal{V}_{\mathcal{S}_t[i]}^{\mathcal{S}\backslash\mathcal{S}_t[i]}| = a_{t,\ell-t},~~\forall t\in[ r_{\min},\ell-1],\nonumber\\
  && \sum\limits_{\substack{\mathcal{S}\subseteq [K], |\mathcal{S}|=\ell}}\sum_{i=1}^{\binom{\ell}{t}}|\mathcal{V}_{\mathcal{S}_t[i]}^{\mathcal{S}\backslash\mathcal{S}_t[i]}|=0,~~ \forall t\in [r_{\min}-1]$ or $t=\ell, \nonumber
  \end{IEEEeqnarray}
   so the first term in \eqref{followprove} can be computed as
  \begin{align}
   &\frac{1}{QN}\sum_{\mathcal{S}\subseteq[K]}\sum_{z=r_{\min}}^{|\mathcal{S}|-1}\frac{|\mathcal{S}|-z}{|\mathcal{S}|-1}\sum_{i=1}^{\binom{|\mathcal{S}|}{z}}|\mathcal{V}_{\mathcal{S}_z[i]}^{\mathcal{S}\backslash\mathcal{S}_z[i]}|\notag \\
   &=\frac{1}{QN}\sum_{\ell=1}^{K}\sum_{t=1}^{\ell}\sum_{\substack{\mathcal{S}\subseteq [K], |\mathcal{S}|=\ell}}\frac{\ell-t}{\ell-1}\sum_{i=1}^{\binom{\ell}{t}}|\mathcal{V}_{\mathcal{S}_t[i]}^{\mathcal{S}\backslash\mathcal{S}_t[i]}|\notag \\
   &=\frac{1}{QN}\sum_{\ell=1}^{K}\sum_{t=1}^{\ell}\frac{\ell-z}{\ell-1}a_{t,\ell-t}\nonumber\\
   &=\frac{1}{QN}\sum_{t=1}^{K}\sum_{d=1}^{K-s}\frac{d}{t+d-1}a_{t,d},
   \label{subload}
  \end{align}
  where the last equality holds by exchanging the order of summation and $l=t+d$.

  \subsection{Complexity of FSCT}\label{ComplexityFSCT}
  For the encoding process of FSCT, Node $k$ will encode $(|\mathcal S|-1)\sum\limits_{\!i\in \left[\binom{|\mathcal{S}|}{z}\right], k\in\mathcal{S}_{z}[i]}\big|\mathcal{V}_{\mathcal{S}_{z}[i]}^{\mathcal {S}\backslash \mathcal{S}_{z}[i]}\big|$ segments into $\bar{n}_{k,z}^{\mathcal{S}}$ numbers of LCs in round $z$ of cluster $\mathcal S$. The time complexity is 
 $O\Big(\bar{n}_{k,z}^{\mathcal{S}}(|\mathcal S|-1)\sum\limits_{\!i\in \left[\binom{|\mathcal{S}|}{z}\right], k\in\mathcal{S}_{z}[i]}\big|\mathcal{V}_{\mathcal{S}_{z}[i]}^{\mathcal {S}\backslash \mathcal{S}_{z}[i]}\big|\Big)$.
  By summing the time complexity in every round and every cluster, we have the total time complexity of the encoding process as
  \begin{align}
  O\bigg(\sum_{\mathcal{S}\subseteq [K]}\sum_{z=r_{\min}}^{|\mathcal S|-1}\sum_{k\in\mathcal{S}}\bar{n}_{k,z}^{\mathcal{S}}(|\mathcal S|-1)\!\!\!\!\!\!\sum\limits_{\!i\in \left[\binom{|\mathcal{S}|}{z}\right], k\in\mathcal{S}_{z}[i]}\!\!\!\!\!\!\big|\mathcal{V}_{\mathcal{S}_{z}[i]}^{\mathcal {S}\backslash \mathcal{S}_{z}[i]}\big|\bigg).
  \end{align}
  
  Also, we can obtain that the decoding complexity for each node $k\in \mathcal{S}$ is $O\Big((|\mathcal{S}|-1)^3\big(\sum\limits_{\!i\in \left[\binom{|\mathcal{S}|}{z}\right], k\notin\mathcal{S}_{z}[i]}\big|\mathcal{V}_{\mathcal{S}_{z}[i]}^{\mathcal {S}\backslash \mathcal{S}_{z}[i]}\big|\big)^3\Big)$ since Node $k$ solves a linear system with $(|\mathcal{S}|-1)\!\!\!\sum\limits_{\!i\in \left[\binom{|\mathcal{S}|}{z}\right], k\notin\mathcal{S}_{z}[i]}\!\!\!\!\!\!\big|\mathcal{V}_{\mathcal{S}_{z}[i]}^{\mathcal {S}\backslash \mathcal{S}_{z}[i]}\big|$ elements in round $z$ of cluster $\mathcal{S}$. Totally, the decoding complexity is
 
  \begin{align}
  O\bigg(\sum_{\mathcal{S}\subseteq [K]}\sum_{z=r_{\min}}^{|\mathcal S|-1}\sum_{k\in\mathcal{S}}\Big((|\mathcal{S}|-1)^3\big(\!\!\!\!\sum_{\!i\in \left[\binom{|\mathcal{S}|}{z}\right], k\in\mathcal{S}_{z}[i]}\!\!\!\!\!\!\!\!\big|\mathcal{V}_{\mathcal{S}_{z}[i]}^{\mathcal {S}\backslash \mathcal{S}_{z}[i]}\big|\big)^3\Big)\bigg),
  \end{align} 
  for the entire decoding process.

 \subsection{Example 2: Illustration of FSCT and Its Optimality}\label{ExFSCT}
 Here, we provide another example to illustrate the detailed transmission strategy of FSCT and show that it could potentially outperform OSCT on communication load.
 Consider the file and Reduce function allocations shown in Fig. \ref{ex2sys} with four nodes, seven files, and 6 Reduce functions. 
 \begin{figure}
   \centering
   \includegraphics[width=8.6cm]{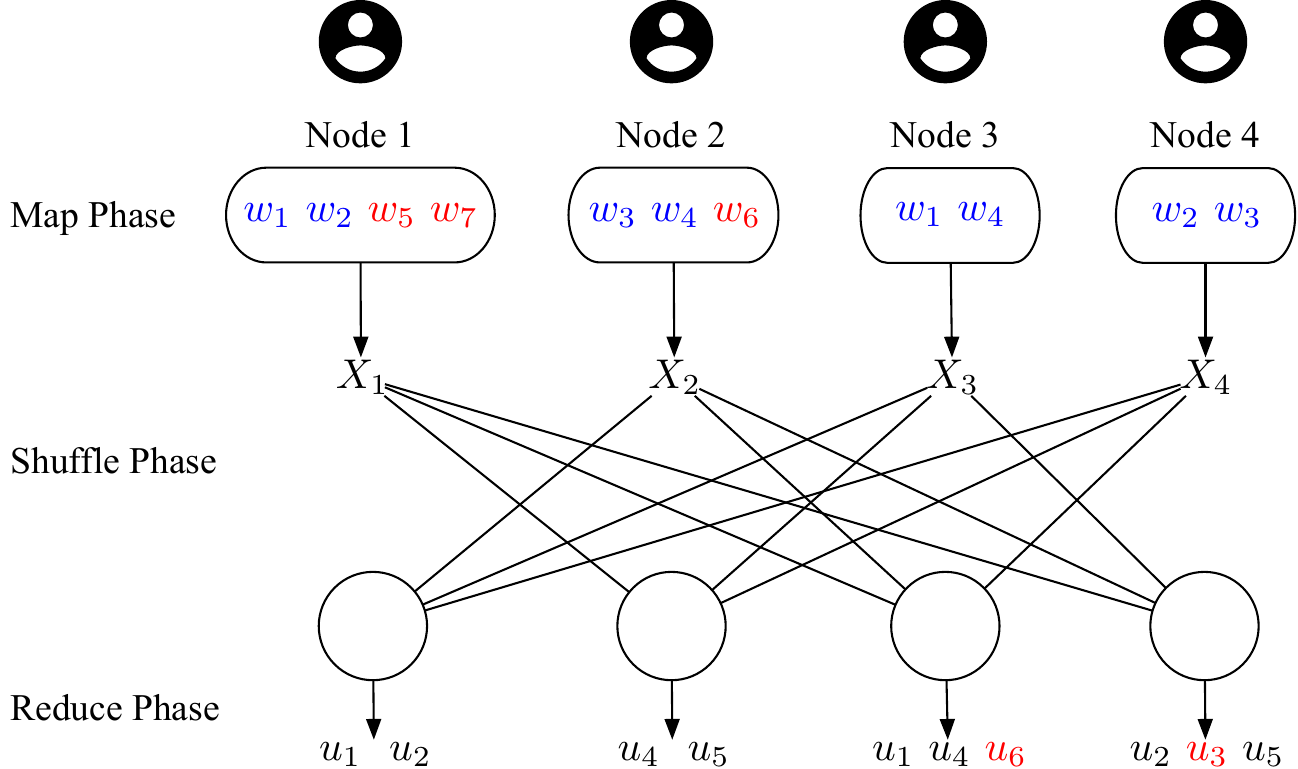}
   \caption{Data allocation and Reduce function assignment of Example 2.}\label{ex2sys}
   \end{figure}
 
 \emph{Method illustration:} 
 First, we regroup the nodes into clusters $\mathcal{S}$ and divide the sending process in each cluster into multiple rounds $r_{\min}\leq z\leq|\mathcal{S}|-1$. The IVs in $(z,\mathcal{S})$-mapped set $\mathcal{V}_z^{\mathcal{S}}$ will be sent in round $z$ of cluster $\mathcal{S}$.
 
 Then, in each round, we check whether the deficit condition and feasible condition defined in \eqref{defiCond} and \eqref{feasible} are satisfied for each node. Taking round $z=2$ of cluster $\mathcal S=\{1,2,3,4\}$ as an example, the IVs categories are shown in Table \ref{ex2t}. According to \eqref{nAndBeta}, we obtain that $n_{1,2}^{\{1,2,3,4\}}=n_{2,2}^{\{1,2,3,4\}}=n_{3,2}^{\{1,2,3,4\}}=n_{4,2}^{\{1,2,3,4\}}=2>0$, which satisfies the deficit conditions for all nodes by checking \eqref{deficitN}. One can verify that all nodes also satisfy the feasible conditions. For example, $\beta_{2,\{2,3\}}=\beta_{3,\{2,3\}}=\beta_{2,\{2,4\}}=\beta_{4,\{2,4\}}=1, \beta_{3,\{3,4\}}=\beta_{4,\{3,4\}}=0$ is a non-negative solution of \eqref{feasible} for Node 1. 
 
 Each IV sent in round $z$ of cluster $S$ will be split into $|\mathcal{S}|-1$ pieces, and each node that satisfies both the deficit condition and feasible condition encodes all locally mapped IVs in this round into a message with $n_{i,z}^{\mathcal{S}}$ LCs, where $n_{i,z}^{\mathcal{S}}$ is the parameter introduced in \eqref{nAndBeta}. For example, each IV in round 2 of $\{1,2,3,4\}$ will be splitted into $|\mathcal{S}|-1=3$ pieces, (e.g., $V_{5,1}$ is splitted into $v_{5,1}^{(1)},v_{5,1}^{(2)}$ and $v_{5,1}^{(3)}$). Node 1 encodes all IV segments in $\mathcal{V}_{ \{1,2\} }^{ \{3,4\} }$, $\mathcal{V}_{ \{1,3\} }^{ \{2,4\} }$ and $\mathcal{V}_{ \{1,4\} }^{ \{2,3\} }$, i.e., $\{v_{5,1}^{(1)}$, $v_{4,2}^{(1)}$, $v_{5,1}^{(2)}, v_{4,2}^{(2)}, v_{5,1}^{(3)}, v_{4,2}^{(3)}\}$. into $n_{1,2}^{\{1,2,3,4\}}=2$ LCs as follows:
 \begin{align}
  &a_1v_{5,1}^{(1)}+a_2v_{4,2}^{(1)}+a_3v_{5,1}^{(2)}+a_4v_{4,2}^{(2)}+a_5v_{5,1}^{(3)}+a_6v_{4,2}^{(3)},\notag\\
  &a_1^2v_{5,1}^{(1)}+a_2^2v_{4,2}^{(1)}+a_3^2v_{5,1}^{(2)}+a_4^2v_{4,2}^{(2)}+a_5^2v_{5,1}^{(3)}+a_6^2v_{4,2}^{(3)},\notag
 \end{align}
 where $a_i\in \mathbb{F}_{2^{\frac{V}{|\mathcal{S}|-1}}}$ for $i\in[6]$ is the randomly generated coefficient.
 
 \begin{table}
  \caption{IVs categories in cluster $\{1,2,3,4\}$ for example 2}
  \centering \label{ex2t}
  \begin{tabular}{|c|c|c|c|c|c|} \hline
  $\mathcal{V}_{ \{1,2\} }^{ \{3,4\} }$ & $\mathcal{V}_{ \{1,3\} }^{ \{2,4\} }$ & $\mathcal{V}_{ \{1,4\} }^{ \{2,3\} }$ & $\mathcal{V}_{ \{2,3\} }^{ \{1,4\} }$ & $\mathcal{V}_{ \{2,4\} }^{ \{1,3\} }$ & $\mathcal{V}_{ \{3,4\} }^{ \{1,2\} }$\\\hline
  $\emptyset$ & $v_{5,1}$ & $v_{4,2}$ & $v_{2,4}$ & $v_{1,3}$ & $\emptyset$ \\\hline
  \end{tabular}
 \end{table}
 Coded shuffle strategies for  Node  2, 3, and 4 are shown below:
 \begin{itemize}
  \item Node 2 multicasts 2 LCs of $v_{2,4}^{(1)}$, $v_{1,3}^{(1)}$, $v_{2,4}^{(2)}$, $v_{1,3}^{(2)}$, $v_{2,4}^{(3)}$, $v_{1,3}^{(3)}$.
  \item Node 3 multicasts 2 LCs of $v_{5,1}^{(1)}$, $v_{2,4}^{(1)}$, $v_{5,1}^{(2)}$, $v_{2,4}^{(2)}$, $v_{5,1}^{(3)}$, $v_{2,4}^{(3)}$.
  \item Node 4 multicasts 2 LCs of $v_{4,2}^{(1)}$, $v_{1,3}^{(1)}$, $v_{4,2}^{(2)}$, $v_{1,3}^{(2)}$, $v_{4,2}^{(3)}$, $v_{1,3}^{(3)}$.
 \end{itemize}
 
 Although each node can not decode the desired IVs directly from a \emph{single} message block, each node can jointly solve its needed IVs segments after receiving all the LCs in round $z$ of cluster $\mathcal{S}$. For example, Node 1 receives six LCs of 12 segments from the other three nodes in round 2 of cluster $\{1,2,3,4\}$. After removing its known segments $\{v_{5,1}^{(1)}$, $v_{4,2}^{(1)}$, $v_{5,1}^{(2)}, v_{4,2}^{(2)}, v_{5,1}^{(3)}, v_{4,2}^{(3)}\}$, Node 1 will solve a linear system with six independent linear combinations of its six desired segments $v_{2,4}^{(1)}$, $v_{1,3}^{(1)}$, $v_{2,4}^{(2)}$, $v_{1,3}^{(2)}$, $v_{2,4}^{(3)}$, $v_{1,3}^{(3)}$.
 
 \emph{Comparison with previous schemes and OSCT:} 
 For our FSCT, as the deficit condition and feasible condition are satisfied by all the nodes in every cluster, the communication load is $L_{\textnormal{fsct}}=0.635$ according to Theorem \ref{ThmFewshot}. According to Lemma 1, the theoretical lower bound of the communication load is $L^{*}=0.635=L_{\textnormal{fsct}}$, which means that FSCT is optimal. 
 
 In this example, FSCT dominates OSCT by exploiting more multicasting gain. According to Theorem \ref{ThmOneshot}, the achievable communication load given by OSCT is $L_{\textnormal{OSCT}}=\frac{28}{42}=0.667$, % smaller than that of \cite{Tao}, but 
 larger than $0.635$ (i.e., not optimal). 
 Specifically, in round 2 of cluster $\{1,2,3,4\}$, due to $\mathcal {V}_{\{ 1,2\}}^{\{3,4\}}=\mathcal {V}_{\{ 3,4\}}^{\{1,2\}}=\emptyset$, we have $\alpha_{1,2}^{\{1,2,3,4\}}=\alpha_{2,2}^{\{1,2,3,4\}}=\alpha_{3,2}^{\{1,2,3,4\}}=\alpha_{4,2}^{\{1,2,3,4\}}=0$. Thus, the IVs will be unicasted and OSCT fails to exploit the multicast gain in this specific round, causing higher communication load. This example shows that FSCT outperforms OSCT by jointly decoding multiple message blocks.
 
 For the scheme proposed in \cite{Ji}, the data placement and Reduce function assignment in this example are fixed and do not follow their combinatorial design, in which each file and each Reduce function are required to be mapped or computed exactly $r$ times. Thus, their Shuffle phase design is not applicable in this example. For the scheme in \cite{Tao}, the data and Reduce function assignment do not follow their design either. If we only borrow the transmission strategy and solve the optimization problem in \cite[Theorem 1]{Tao} with the Mapping load $(m_1,\ldots,m_4)=\left(\frac{4}{7},\frac{3}{7},\frac{2}{7},\frac{2}{7}\right)$ and Reducing load $(\omega_1,\ldots,\omega_4)=\left(\frac{1}{3},\frac{1}{3},\frac{1}{2},\frac{1}{2}\right)$, we can get the total communication load as $\frac{32}{42}=0.762$, larger than both OSCT (with load 0.667) and FSCT (with load 0.635).

 \section{Numerical Results and Discussions}\label{Subsec_Numerical}
  \subsection{Numerical Results}
 In this section, we compare our communication loads in Theorem \ref{ThmOneshot} and Theorem \ref{ThmFewshot} with the uncoded scheme and lower bound in Lemma \ref{LemmaLower}, under pre-set data placement and Reduce function assignment. We consider a 4-node system, with heterogeneous mapping loads $\mathbf{m}=(\frac{1}{2}-d, \frac{1}{2}-d, \frac{1}{2}+d, \frac{1}{2}+d)$ and reducing load $(\frac{1}{2}+d, \frac{1}{2}+d, \frac{1}{2}-d, \frac{1}{2}-d)$, with $d=0,\frac{1}{64},\cdots,\frac{31}{64}$. When load bias $d$ grows larger, nodes with higher mapping load will compute more Reduce functions, which means the system becomes more  biased. Since the coding strategy of OSCT and FSCT are based on pre-set data placement and Reduce function assignment, we randomly generate \emph{50 samples} of different data placements and Reduce function assignments for each $d$, and compute the corresponding communication load under each assignment, then we take the \emph{average} of 50 communication loads. 
 
 Note that our schemes allow pre-set data placement and Reduce function assignment, while previous works in \cite{ThreeWorker}, \cite{Ji} and \cite{Tao} relied on free data placement design, we can not directly compare the communication loads achieved by the aforementioned works here. As we discussed in the example of \ref{SchFewShot}, our OSCT and FSCT can achieve a lower communication load than \cite{Tao} when only applying the coding strategy of their scheme given concrete file and Reduce function assignment as Fig. \ref{ex2sys}, where schremes in \cite{ThreeWorker,Ji} are not applicable.
 
 \begin{figure}[!h]
   \centering
   \includegraphics[width=8.5cm]{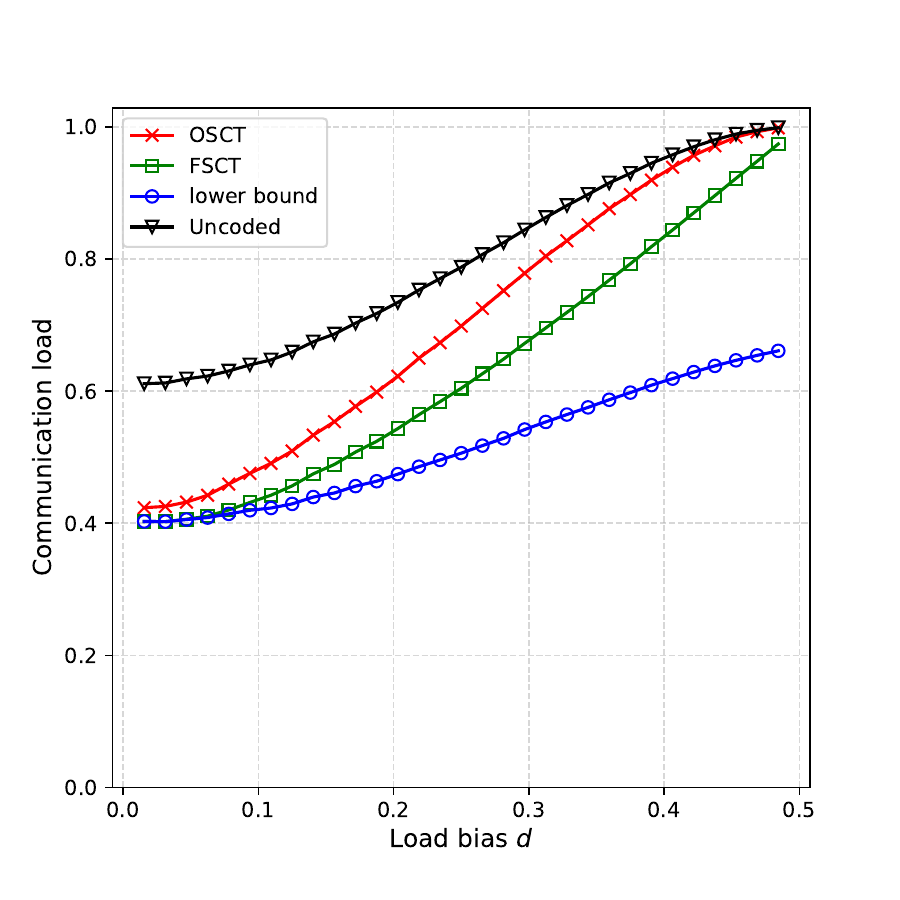}
   \vspace{-3mm}
   \caption{Comparison of communication load achieved by OSCT and FSCT with lower bound in systems with mapping load $\mathbf{m}=(\frac{1}{2}-d, \frac{1}{2}-d, \frac{1}{2}+d, \frac{1}{2}+d)$ and reducing load $(\frac{1}{2}+d, \frac{1}{2}+d, \frac{1}{2}-d, \frac{1}{2}-d)$}\label{load_comparison}
   \vspace{-3mm}
 \end{figure}
 
 In Fig. \ref{load_comparison}, we compare the average communication load achieved by OSCT, FSCT, and uncoded scheme, as well as the lower bound with different load bias $d$. It shows that both OSCT and FSCT achieve lower communication loads than the uncoded scheme given any specific file and Reduce function assignments for all $d$. Specifically, with relatively small $d$, both OSCT and FSCT outperform the uncoded scheme with a 33\% reduction in communication load, and their upper bounds are close to the lower bound. FSCT achieves better performance compared to OSCT due to the Few-shot transmission, which is consistent with Remark \ref{CompareTwoschemes}. Besides, both OSCT and FSCT converge to the uncoded scheme when $d$ approaches $\frac{1}{2}$. This is because the nodes with mapping load $\frac{1}{2}-d$ are the bottleneck of coded multicast gains. More specifically,  when $d$ approaches $\frac{1}{2}$, to ensure that the nodes with mapping load $\frac{1}{2}-d$ can decode their desired IVs,  most IVs should be unicasted to the intended nodes, leading to communication load close to that of the uncoded scheme.
 \begin{figure}[!htbp]
  \vspace{-8mm}
   \centering
   \includegraphics[width=8.5cm]{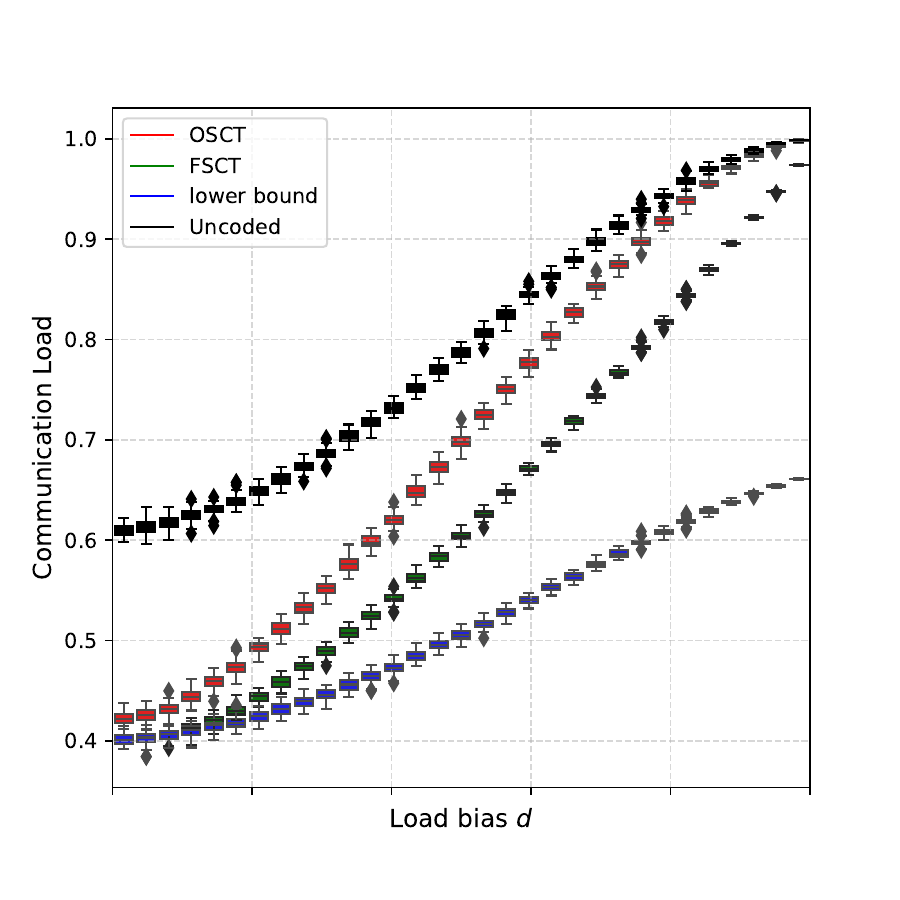}\
   \vspace{-3mm}
   \caption{Boxplot of communication loads corresponding to OSCT (red color), FSCT (yellow color), and lower bound (blue color) among 50 samples for each $d$.}\label{FSCT_box}
 \end{figure}
 
 The impact brought by specific data placement and Reduce function assignment given the same mapping load and Reducing load is illustrated in the boxplot of Fig. \ref{FSCT_box}, which characterizes the degree of variation among different file and Reduce function assignments for each $d$. We observe some consistent behaviors of variance among OSCT, FSCT, uncoded scheme, and lower bound, i.e., the dispersion is large when $d$ is small and gradually vanishes as $d$ increases. The reason for this phenomenon is that when the degree of heterogeneity rises, it becomes harder for schemes to exploit the multicast gain due to the users with mapping load $\frac{1}{2}-d$, and the influence of different file and function assignments decreases, resulting in all communication loads approaching 1 when $d$ is close to $\frac{1}{2}$. Besides, compared to OSCT, FSCT achieves a relatively smaller variance for different file and function assignments, especially when $d$ is greater than $0.4$, and has a similar degree of dispersion with the lower bound, indicating that FSCT is more robust to the variation of $d$ than OSCT.

  \subsection{Discussions and Future Works}
 \subsubsection{Complexity} 
 In Section \ref{ComplexityOSCT} and \ref{ComplexityFSCT}, we have analyzed the computation complexities of  OSCT and FSCT, and provided the equations characterizing the encoding and decoding complexities respectively. 
Additionally, OSCT needs an additional complexity incurred by solving the optimization problem $\mathcal{P}_{\textnormal{OSCT}}(\mathcal{S},z)$ for each cluster $\mathcal{S}$ and round $z$, which requires an extra computation complexity   no more than $O\big(\sum_{|\mathcal{S}|\subseteq [K]}\sum_{z=r_{\min}}^{|\mathcal{S}|}|\mathcal{S}|^3=\sum_{\ell=1}^K\ell^4\binom{K}{\ell}\big)$ according to \cite{CVX}.
  
  Note that in each round $z$ of cluster $\mathcal{S}$, the encoding and decoding complexity of OSCT is similar to that of the original CDC scheme in \cite{CDC}, whose encoding complexity is $O\left(\sum_{\mathcal{S}\subseteq [K]}|\mathcal{S}|\binom{|\mathcal{S}|-1}{r-1}\log^2\binom{|\mathcal S|-1}{r-1}\right)$ with the computation load $r$. As we divide the sending process into multiple rounds, there is, at most, an extra linear multiplier $z\leq K$ in our OSCT complexity compared with the CDC scheme in \cite{CDC}. Briefly, the more heterogeneous a system is, the more rounds it needs to consider and the higher the complexity. Meanwhile, it is possible to reduce the complexity by sacrificing the optimality of the communication load. For example, by transmitting $(z,\mathcal{S})$-mapped IVs in round $z^{'}$ ($z^{'}< z$) of cluster $\mathcal{S}$ together with the $(z^{'},\mathcal{S})$-mapped IVs, we can eliminate round $z$ in the transmission and reduce the complexity. However, such a scheme is sub-optimal since the multicast gain is not fully exploited for the $(z,\mathcal{S})$-mapped IVs. 
  
  There exist several works aiming to reduce the number of files to be split to decrease the complexity\cite{Resolve, FLCD}. However, their designs rely on homogeneous or specific combinatorial file allocation design, which are not applicable to systems with arbitrary data and Reduced function assignments. Balancing complexity and communication load remains an open challenge for general HetDC systems. 

  \subsubsection{Overhead and privacy issue during the initialization} The transmission strategy designs of OSCT and FSCT need the information of $\mathcal {V}_{\mathcal {S}_{1}}^{\mathcal {S} \backslash \mathcal {S}_{1}}$ defined in \eqref{szmap}, so each node should be aware of the file allocations and Reduce function assignments of the system.  To inform each node of the distribution of files, $KN$ bits overhead need to be considered. This is a common assumption through the CDC literature\cite{CDC,Ji,Tao}, which also holds in this paper.
  
  In practical scenarios such as distributed learning, sharing file allocations may lead to privacy issues\cite{CDCSurvey}, especially when sensitive data is involved. 
 There exist several works focusing on preserving security and privacy\cite{secure1,secure2,secure3} while tackling privacy issues in systems with arbitrary data and Reduced function assignments remains a topic for future research.

\section{Conclusion}\label{Conclusion}
In this paper, we consider the general MapReduce-type system where the data placement and Reduce function assignment are arbitrary and pre-set among all computing nodes. We propose two coded distributed computing approaches for the general HetDC system: One-shot Coded Transmission (OSCT) and Few-shot Coded Transmission (FSCT). Both approaches encode the intermediate values into message blocks with different sizes to exploit the multicasting gain, while the former allows each message block to be successfully decoded by the intended nodes, and the latter has each node jointly decode multiple message blocks, and thus can further reduce the communication load. With a theoretical lower bound on the optimal communication load, we characterize the sufficient conditions for the optimality of OSCT and FSCT respectively, and prove that OSCT and FSCT are optimal for several classes of MapReduce systems.

% or


\appendices
\section{Proof of Theorem \ref{t2} and Theorem \ref{ThmOptFSCT}}\label{TwoTheorems}

\subsection{Proof of Theorem \ref{t2}}\label{ConverseOSCTOPt}

As shown in \eqref{roundcost}, the communication cost in round $z$ of clusters $\mathcal S$ can be computed as
\begin{align}
L_{z,\mathcal{S}}=\sum_{j\in\mathcal{S}}\tbinom{|\mathcal S|-2}{z-1}\alpha_{j,z}^{\mathcal{S}}V+\sum_{i=1}^{\binom{|\mathcal{S}|}{z}}(\tau_{i,z}^{\mathcal{S}})^+V.\label{Lo}
\end{align}
When $\sum_{i=1}^{\binom{|\mathcal S|}{z}}\big(\sum_{j\in \mathcal \mathcal S_z[i]}\alpha_{j,z}^S-\big|\mathcal{V}_{\mathcal{S}_z[i]}^{\mathcal{S}\backslash\mathcal{S}_z[i]}\big|\big)^2=0$, we have $$\tau_{i,z}^{\mathcal S} = \left|\mathcal{V}_{\mathcal {S}_z[i]}^{\mathcal {S} \backslash \mathcal {S}_z[i]}\right|-\sum_{j\in \mathcal S_z[i]}\alpha_{j,z}^{\mathcal S}=0,\forall i \in [\tbinom{|\mathcal{S}|}{z}].$$ 
By taking summation over all the subsets $\mathcal \mathcal S_z[i]$ for $i\in\big[\binom{|\mathcal S|}{z}\big]$ on both sides of $\sum_{j\in \mathcal \mathcal S_z[i]}\alpha_{j,z}^{\mathcal{S}}=\big|\mathcal{V}_{\mathcal{S}_z[i]}^{\mathcal{S}\backslash\mathcal{S}_z[i]}\big|$, we have
\begin{align}
\sum_{i=1}^{\binom{|\mathcal S|}{z}}\sum_{j\in \mathcal \mathcal S_z[i]}\alpha_{j,z}^{\mathcal{S}}\overset{(a)}{=}\tbinom{|\mathcal S|-1}{z-1}\cdot\sum_{j\in\mathcal{S}}\alpha_{j,z}^{\mathcal{S}} = \sum_{i=1}^{\binom{|\mathcal S|}{z}}\big|\mathcal{V}_{\mathcal{S}_z[i]}^{\mathcal{S}\backslash\mathcal{S}_z[i]}\big|\label{sd}
\end{align}
where the (a) follows from the fact that there are $\binom{|\mathcal{S}|-1}{z-1}$ subsets with size $z$ who contain Node $j$ for each $j\in \mathcal{S}$, which means that each $\alpha_{j,z}^S$ will be added $\binom{|\mathcal{S}|-1}{z-1}$ times. 
Plugging \eqref{sd} into (\ref{Lo}), we have
\begin{align}
L_{z,\mathcal{S}}=\sum_{i=1}^{\binom{|\mathcal S|}{z}}\big|\mathcal{V}_{\mathcal{S}_z[i]}^{\mathcal{S}\backslash\mathcal{S}_z[i]}\big|\frac{\binom{|\mathcal S|-2}{z-1}}{\binom{|\mathcal S|-1}{z-1}}V.
\end{align}
By summing all $\mathcal{S}\subseteq [K]$ with $|\mathcal{S}|=\ell$, we have
\begin{align}
L_{z,\ell} &=\sum_{\substack{\mathcal{S}\subseteq[K],\\|\mathcal{S}|=\ell}}\sum_{i=1}^{\binom{|\mathcal S|}{z}}\big|\mathcal{V}_{\mathcal{S}_z[i]}^{\mathcal{S}\backslash\mathcal{S}_z[i]}\big|\frac{\binom{|\mathcal S|-2}{z-1}}{\binom{|\mathcal S|-1}{z-1}}V\overset{(a)}{=}a_{z,\ell-z} \frac{\binom{|\mathcal S|-2}{z-1}}{\binom{|\mathcal S|-1}{z-1}} V\overset{(b)}{=}a_{t,d}\frac{d}{t+d-1}V,
\end{align}
where the (a) holds by the definition of $a_{t,d}$ in Lemma 1, and (b) holds by rewriting $z$ as $t$ and $\ell$ and $t+d$.
Since $|\mathcal S|=t+d$,
\begin{align*}
 L_{t,d}=a_{t,d}\frac{t+d-s}{t+d-1}V=a_{t,d}\frac{d}{t+d-1}V.
\end{align*}
%where $s$ comes from $1$ to $K$ and $d$ comes from $1$ to $K-s$,
By summing the communication loads $L_{t,d}$ in all rounds and all clusters $\mathcal{S}\subseteq [K]$, we obtain that the communication load (normalized by $QNV$) is
\begin{equation*} 
L =\frac{1}{QN}\sum_{s=1}^{K}\sum_{d=1}^{K-s}a_{t,d}\frac{d}{t+d-1} \notag,
\end{equation*}
which coincides with the lower bound in Lemma 1.

\subsection{Proof of Theorem \ref{ThmOptFSCT}} \label{ConverseFSCTOpt}
According to Theorem \ref{ThmFewshot}, the upper bound of communication load achieved by FSCT is 
\begin{align}
&L_{\textnormal{fsct}}\!=\!\frac{1}{QN}\sum_{s=1}^{K}\sum_{d=1}^{K-s}a_{t,d}\frac{d}{t+d-1}\notag\\
&+\frac{1}{QN}\sum_{S\subseteq[K]}\sum_{z=r_{\min}}^{|\mathcal S|-1}\frac{1}{|\mathcal{S}|-1}\bigg(\!\!\!\!\!\!\!\!\!\!\sum_{\substack{k\in\mathcal{S}:\\\mathds{1}(\mathscr{E}_1(k,z,\mathcal{S}))=0}}\!\!\!\!\!\!\!\!\!\!\! |n_{k,z}^{\mathcal{S}}|-\!\!\!\!\sum_{\substack{k\in\mathcal{S}:\\\mathds{1}(\mathscr{E}_1(k,z,\mathcal{S}))=1\\\mathds{1}(\mathscr{E}_2(k,z,\mathcal{S}))=0}}\!\!\!\!\!\!\!\!\!\!\! n_{k,z}^{\mathcal{S}}\bigg)\notag\\
&+\frac{1}{QN}\sum_{S\subseteq[K]}\sum_{z=r_{\min}}^{|\mathcal S|-1}\sum_{\substack{i\in\mathcal{S}:\\\mathds{1}(\mathscr{E}_2(i,z,\mathcal{S}))=0}}\!\!\max_{j\in \mathcal{S}}\!\!\!\sum_{\substack{\mathcal{S}_1\subseteq \mathcal{S}:|S_1|=z\\j\in S_1,i\not\in S_1}}\frac{|\mathcal {V}_{\mathcal {S}_1}^{\mathcal {S} \backslash \mathcal {S}_1}|}{z}.\label{load22}
\end{align} 
If the deficit conditions and feasible conditions are satisfied for all nodes in each cluster, i.e., $\mathds{1}(\mathscr{E}_1(i,z,\mathcal{S}))=1,\mathds{1}(\mathscr{E}_1(i,z,\mathcal{S}))=1$ for all $\mathcal{S}\subseteq [K], r_{\min}\leq z \leq |\mathcal{S}|-1, i\in\mathcal{S}$, then the second and third terms in \eqref{load22} will be 0. Hence, the communication load can be reduced to
\begin{align}
L_{\textnormal{fsct}}\!=\!\frac{1}{QN}\sum_{s=1}^{K}\sum_{d=1}^{K-s}a_{t,d}\frac{d}{t+d-1},\notag
\end{align}
matching the lower bound in Lemma 1, which means the FSCT is optimal.

\section{Proof of Corollary \ref{OSCTCorollary}}\label{AppOSCT}
In this section, we prove that our OSCT scheme is optimal in 3-node systems and semi-homogeneous systems.

\subsection{Proof of Optimality in 3-node System}\label{OSCT3wProve}

In \cite{ThreeWorker} the authors considered a 3-node system, and provided an optimal data placement and Shuffle phase design that minimize the communication load in the 3-node system. Now we prove that given any data placement $\mathcal{M}$, the communication load in Theorem \ref{ThmOneshot} is the same as \cite[Lemma 1]{ThreeWorker}, which is given below.

\begin{lemma}(\cite[Lemma 1]{ThreeWorker})\label{3wk}
Given data placement $\mathcal{M}\triangleq(\mathcal{M}_1 , \mathcal{M}_2 , \mathcal{M}_3)$, the communication load $\mathcal{L}_{\mathcal{M}}$ is achievable, where:
\begin{align*} 
 \mathcal{L}_\mathcal{M}=2(S_{1}+S_{2}+S_{3})+g(S_{12},\ S_{13},\ S_{23}),
\end{align*}
and 
\begin{align}
 g(x_{1},x_{2},x_{3})\!=\!\frac{1}{2}\Big(\Big| \max\limits_{k}x_{k}\!+\!\sum\limits_{k}\frac{x_{k}}{2}\Big| \!+\!\Big| \max\limits_{k}x_{k}\!-\!\sum\limits_{k}\frac{x_{k}}{2}\Big|\Big),\label{gFunc}
\end{align}
where
\begin{align*} 
 &\mathcal{S}_{123} \ {\buildrel\triangle\over =}\ \mathcal{M}_{1}\cap \mathcal{M}_{2}\cap \mathcal{M}_{3},\\ 
 &\ \mathcal{S}_{12} \ {\buildrel\triangle\over =}\ (\mathcal{M}_{1}\cap \mathcal{M}_{2})\ \backslash\ \mathcal{S}_{123},\quad \mathcal{S}_{1}\ {\buildrel\triangle\over =}\ \mathcal{M}_{1}\ \backslash\ (\mathcal{M}_{2}\cup \mathcal{M}_{3}),\\ 
 &\ \mathcal{S}_{13} \ {\buildrel\triangle\over =}\ (\mathcal{M}_{1}\cap \mathcal{M}_{3})\ \backslash\ \mathcal{S}_{123},\quad \mathcal{S}_{2}\ {\buildrel\triangle\over =}\ \mathcal{M}_{2}\ \backslash\ (\mathcal{M}_{1}\cup \mathcal{M}_{3}),\\ 
 &\ \mathcal{S}_{23} \ {\buildrel\triangle\over =}\ (\mathcal{M}_{2}\cap\mathcal{M}_{2})\ \backslash\ \mathcal{S}_{123},\quad \mathcal{S}_{3}\ {\buildrel\triangle\over =}\ \mathcal{M}_{3}\ \backslash\ (\mathcal{M}_{1}\cup \mathcal{M}_{2}). 
\end{align*}
\end{lemma}
For simplicity, we denote the set cardinality by $S \triangleq |\mathcal S|$.
Without loss of generality, assume $S_{12}\leq S_{13}\leq S_{23}$. 

Since the subset of files $\mathcal{S}_{123}$ is available at every node, we do not need to consider the communication cost incurred by files in $\mathcal{S}_{123}$. As the subsets of files $\mathcal{S}_1$, $\mathcal{S}_2$ and $\mathcal{S}_3$ are available at only one node, each node $k\in[3]$ need to unicast $v_{j,n}$ to Node $j$ for $j\neq k,n\in\mathcal{S}_k$, which causes a total of $2(S_1 + S_2 + S_3)$ transmissions. Now we need to calculate the communication cost in the round 2 of cluster $S=\{1,2,3\}$ in which the file sets $\mathcal{S}_{12},\mathcal{S}_{13}$ and $\mathcal{S}_{23}$ are considered. Here we rewrite $\alpha_{k,z}^{\mathcal{S}}$ as $\alpha_k$ for simplicity.

\subsubsection{Case1} If $S_{12}+S_{13}-S_{23}\geq 0$, we can obtain $\alpha_1,\alpha_2$ and $\alpha_3$ below by soloving the $\mathcal{P}_{\textnormal{OSCT}}$ in Theorem \ref{ThmOneshot}, 
\begin{align}
\alpha_1 = \frac{1}{2}(S_{12} + S_{13} - S_{23})\notag\\
\alpha_2 = \frac{1}{2}(S_{12} + S_{23} - S_{13})\notag\\
\alpha_3 = \frac{1}{2}(S_{13} + S_{23} - S_{12}).\label{sol}
\end{align}

As $\alpha_1,\alpha_2$ and $\alpha_3$ obtained in \eqref{sol} are non-negative, the communication load can be computed as
\begin{align}
L&=2\left(S_1+S_2+S_3\right)+(\alpha_1+\alpha_2+\alpha_3)\notag\\
&=2\left(S_1+S_2+S_3\right)+\frac{1}{2}(S_{12} + S_{13} + S_{23}). \label{case1}
\end{align}

\subsubsection{Case2} If $S_{12}+S_{13}-S_{23}<0$, from KKT condition, we can find that $\alpha_1=0,\alpha_2=\frac{1}{3}(2S_{12}+S_{23}-S_{13}),\alpha_3=\frac{1}{3}(2S_{13}+S_{23}-S_{12})$ is an optimal solution of $\mathcal{P}_\textnormal{osct}$. 
According to Theorem \ref{ThmOneshot}, we have $$\tau_1=\alpha_1+\alpha_2-S_{12}=\frac{1}{3}(S_{23}-S_{12}-S_{13}),$$ $$\tau_2=\alpha_1+\alpha_3-S_{13}=\frac{1}{3}(S_{23}-S_{12}-S_{13}),$$ $$\tau_3=\alpha_2+\alpha_3-S_{23}=\frac{1}{3}(S_{12}-S_{13}-S_{23}).$$

Hence, according to Theorem\ref{ThmOneshot}, the communication load is
\begin{align}
L&=2\left(S_1+S_2+S_3\right)+\alpha_2+\alpha_3+\sum_{i=1}^{3}(\tau_i)^+ = 2\left(S_1+S_2+S_3\right)+S_{23}.\label{case2}
\end{align}

In \cite{ThreeWorker}, the authors proved that the communication load \eqref{case1} and \eqref{case2} are coherent with \eqref{gFunc} in Lemma 
\ref{3wk}, i.e., our communication load in Theorem \ref{ThmOneshot} is tight.

\subsection{Proof of Optimality in Semi-Homogeneous System}\label{OSCTSemiProve}

First, we present the corresponding communication load in Theorem \ref{ThmOneshot} for the semi-homogenous system, and then prove that it is tight.

According to $\mathcal{P}_{\textnormal{OSCT}}$ in Theorem \ref{ThmOneshot}, we can get the solution:
\[
\alpha^{\mathcal S}_{i,z}=\frac{N}{r}\sum\limits_{s=|\mathcal S|-r}^{|\mathcal S|}\frac{\binom{r}{|\mathcal S|-s}}{\binom{K}{r}\binom{K}{s}}Q_s, \forall i\in\mathcal{S}, z=r,|\mathcal{S}|\geq r+1,
\] 
and $\alpha^{\mathcal S}_{i,z}=0$ for other values of $\mathcal{S}$, $z$ and $i$. Then from Theorem \ref{ThmOneshot}, the communication load achieved by OSCT is
\begin{align}
L_{\textnormal{OSCT}}&=\sum_{S\subseteq[K]}\sum_{z=1}^{|\mathcal S|-1}\sum_{i\in S}\binom{|\mathcal S|-2}{z-1}\sum\limits_{s=|\mathcal S|-r}^{|\mathcal S|}\frac{\binom{r}{|\mathcal S|-s}}{\binom{K}{r}\binom{K}{s}r}\frac{Q_s}{Q}=\sum_{S\subseteq[K]}\sum\limits_{s=|\mathcal S|-r}^{|\mathcal S|}\frac{|\mathcal{S}|\binom{|\mathcal S|-2}{r-1}\binom{r}{|\mathcal S|-s}}{r\binom{K}{r}\binom{K}{s}}\frac{Q_s}{Q}.
\end{align}
As there are $\binom{K}{\ell}$ subsets $\mathcal{S}\subseteq[K]$ with size $r+1\leq\ell \leq K$, by arranging the order of summation and rewrite $|\mathcal{S}|$ as $\ell$, the communication load is
\begin{align}
 L_{\textnormal{OSCT}}&=\sum_{s=1}^{K}\sum \limits _{\ell = \max \{r+1,s\}}^{\min \{r+s,K\}} \frac {\ell {\binom{K }{ \ell }} {\binom{\ell -2 }{ r -1}} {\binom{r }{ \ell -s}}}{r {\binom{K }{ r}}{\binom{K }{ s}}}\frac{Q_s}{Q}=\sum_{s=1}^K L_\textnormal{CDC}(s)\frac{Q_s}{Q},\label{upSemi}
\end{align} 
where $L_\textnormal{CDC}(s)=\sum \limits _{\ell = \max \{r+1,s\}}^{\min \{r+s,K\}} \frac {\ell {\binom{K }{ \ell }} {\binom{\ell -2 }{ r -1}} {\binom{r }{ \ell -s}}}{r {\binom{K }{ r}}{\binom{K }{ s}}}$ is the communication load in \cite[Theorem 2]{CDC} with each Reduce function computed by $s$ nodes.

Next, we prove that the communication load in \eqref{upSemi} is tight. For the semi-homogeneous system, according to the definition of $a_{t,d}$ in Lemma \ref{LemmaLower}, we obtain that for $t,d$ with $t+d\leq K,$
\begin{align}
a_{t,d}=\sum_{s=d}^{t+d}NQ_s\binom{t}{t+d-s}\frac{\binom{K}{t+d}\binom{t+d}{t}}{\binom{K}{t}\binom{K}{s}}. \label{asd}
\end{align} 
Combining \eqref{lowerbound} in Lemma \ref{LemmaLower} and \eqref{asd}, we obtain that the lower bound of the optimal communication load is
\begin{align}
L^*&\geq \sum_{d=1}^{K-r}\sum_{s=d}^{r+d}\frac{\binom{r}{r+d-s}\binom{K}{r+d}\binom{r+d}{r}}{\binom{K}{r}\binom{K}{s}}\frac{Q_s}{Q}\frac{d}{r+d-1}\overset{(a)}{=}\sum_{\ell=r+1}^{K}\sum_{s=d}^{\ell}\frac{\binom{r}{\ell-s}\binom{K}{\ell}\binom{\ell}{r}(\ell-r)}{\binom{K}{r}\binom{K}{s}(\ell-1)}\frac{Q_s}{Q}, \label{semi}
\end{align}
where (a) is obtained by rewriting $r+d$ as $\ell$. Arranging the order of summation in \eqref{semi}, we have:
\begin{align}
L^*\geq \sum_{s=1}^{K}\sum \limits _{\ell = \max \{r+1,s\}}^{\min \{r+s,K\}} \frac {\ell {\binom{K }{ \ell }} {\binom{\ell -2 }{ r -1}} {\binom{r }{ \ell -s}}}{r {\binom{K }{ r}}{\binom{K }{ s}}}\frac{Q_s}{Q},\label{semiLowBound}
\end{align}
which is the same as $ L_{\textnormal{OSCT}}$ in \eqref{upSemi}. Thus, OSCT is optimal for the semi-homogeneous system.

\section{A Solution $a_{i,z}^{\mathcal{S}}$ to $\mathcal{P}_{\textnormal{OSCT}}(\mathcal{S},z)$}\label{derivationAlpha}
 In this section, we derive an optimal solution of $\mathcal{P}_{\textnormal{OSCT}}(\mathcal{S},z)$ if the deficit ratio of each node $k\in\mathcal{S}$ is less than the threshold value $\frac{|\mathcal S|-z}{z-1}$ and $|\mathcal{V}_{\mathcal {S}_z[i]}^{\mathcal {S} \backslash \mathcal {S}_z[i]}|\neq 0 $ for all $i\in\big[\binom{|\mathcal{S}|}{z}\big]$.
 
  Recall that the objective function of $\mathcal{P}_{\textnormal{OSCT}}$ in Theorem \ref{t1} is 
  \begin{align}
  f &= \Big(\!\sum_{k\in \mathcal S_z[1]}\!\!\alpha_{k,z}^{\mathcal S}-|\mathcal{V}_{\mathcal {S}_z[1]}^{\mathcal {S} \backslash \mathcal {S}_z[1]}|\Big)^2\!\!+\!\Big(\!\sum_{k\in \mathcal S_z[2]}\!\!\alpha_{k,z}^{\mathcal S}-|\mathcal{V}_{\mathcal {S}_z[2]}^{\mathcal {S} \backslash \mathcal {S}_z[2]}|\Big)^2 +\cdots+\Big(\sum_{k\in S[\binom{|\mathcal S|}{z}]}\alpha_{k,z}^{\mathcal S}-|\mathcal{V}_{\mathcal {S}_z[\binom{|\mathcal S|}{z}]}^{\mathcal {S} \backslash \mathcal {S}_z[\binom{|\mathcal S|}{z}]}|\Big)^2. \label{fff}
  \end{align}
  
   By unfolding the square and merging the similar term in \eqref{fff}, $f$ can be expressed as:
  \begin{align}
  f&= \binom{|\mathcal S|-1}{z-1}\sum_{k\in\mathcal{S}}\left(\alpha_{k,z}^{\mathcal S}\right)^2 + 2 \binom{|\mathcal S|-2}{z-2}\sum_{i,j\in\mathcal{S},i<j} \alpha_{i,z}^{\mathcal S}\alpha_{j,z}^S -2\sum_{k\in\mathcal{S}}\Big(\sum_{\!i\in \left[\binom{|\mathcal{S}|}{z}\right], k\in\mathcal{S}_{z}[i]}\big|\mathcal{V}_{\mathcal{S}_{z}[i]}^{\mathcal {S}\backslash \mathcal{S}_{z}[i]}\big|\Big)\alpha_{k,z}^{\mathcal{S}}+\Delta,
  \end{align}
  where $\Delta$ is a term irrelevant to $\alpha_{k,z}^{\mathcal{S}}$ for all $k\in\mathcal S$.
  
  Take partial differential over each $\alpha_{k,z}^{\mathcal S}$:
  \begin{align*}
  \frac{\partial f}{\partial \alpha_{k,z}^{\mathcal S}} 
  = &2\tbinom{|\mathcal S|-1}{z-1}\alpha_{k,z}^{\mathcal S}\! +\! 2\tbinom{|\mathcal S|-2}{z-2}\sum_{j\neq k}\alpha_{j,z}^S - 2\!\!\!\!\sum_{\!i\in \left[\binom{|\mathcal{S}|}{z}\right], k\in\mathcal{S}_{z}[i]}\!\big|\mathcal{V}_{\mathcal{S}_{z}[i]}^{\mathcal {S}\backslash \mathcal{S}_{z}[i]}\big|, ~~ \forall k\in\mathcal S.
  \end{align*}
  
  For convenience, assume $\mathcal{S}=\{1,2,\cdots,|\mathcal{S}|\}$. By setting all partial differentials to be zero, we can get the LCs written in the following matrix form:
  \begin{align} 
  \underbrace{
    \begin{bmatrix}\! 
      \binom{|\mathcal S|-1}{z-1} & \binom{|\mathcal S|-2}{z-2} & \cdots & \binom{|\mathcal S|-2}{z-2} \\ 
      \binom{|\mathcal S|-2}{z-2} & \binom{|\mathcal S|-1}{z-1} & \cdots & \binom{|\mathcal S|-2}{z-2} \\ 
      \vdots & \vdots & \ddots & \vdots \\ 
      \binom{|\mathcal S|-2}{z-2} & \binom{|\mathcal S|-2}{z-2} & \cdots & \binom{|\mathcal S|-1}{z-1} \!\!\!
      \end{bmatrix}.
  }_{\mathbf A} \!
   \begin{bmatrix}\! 
    \alpha_{1,z}^S\\ 
    \alpha_{2,z}^S\\ 
    \vdots\\ 
    \alpha_{|\mathcal S|,z}^S \!
   \end{bmatrix}\! \!\! 
   \!=\!\!\! 
   \underbrace{\begin{bmatrix} 
   \sum\limits_{\!i\in \big[\binom{|\mathcal{S}|}{z}\big], 1\in\mathcal{S}_{z}[i]}\big|\mathcal{V}_{\mathcal{S}_{z}[i]}^{\mathcal {S}\backslash \mathcal{S}_{z}[i]}\big| \\ 
   \sum\limits_{\!i\in \big[\binom{|\mathcal{S}|}{z}\big], 2\in\mathcal{S}_{z}[i]}\big|\mathcal{V}_{\mathcal{S}_{z}[i]}^{\mathcal {S}\backslash \mathcal{S}_{z}[i]}\big|\\ 
   \sum\limits_{\!i\in \big[\binom{|\mathcal{S}|}{z}\big], |\mathcal{S}|\in\mathcal{S}_{z}[i]}\big|\mathcal{V}_{\mathcal{S}_{z}[i]}^{\mathcal {S}\backslash \mathcal{S}_{z}[i]}\big| 
  \end{bmatrix}}_{\mathbf{y}}. \label{LC}
  \end{align}
  
  To solve each $\alpha_{k,z}^{\mathcal S}$, we just need to find the inverse matrix of ${\mathbf A}$. From \cite{LA}, we know that $\left({\mathbf A}\right)^{-1}$ always exists for $|\mathcal S|> 1$, and can be written as
  \begin{align} 
  \left({\mathbf A}\right)^{-1} = 
   \begin{bmatrix}\! 
   c_z & p_z & \cdots & p_z \\ 
   p_z & c_z & \cdots & p_z \\ 
   \vdots & \vdots & \ddots & \vdots \\ 
   p_z & p_z & \cdots & c_z \!
   \end{bmatrix}, 
  \end{align} 
  where
  \begin{align}
  c_z &= \frac{1}{z\cdot\binom{|\mathcal S|-2}{z-1}}+(|\mathcal S|-2)\cdot\frac{z-1}{z\cdot(|\mathcal S|-z)\cdot\binom{|\mathcal S|-1}{z-1}}\label{cz},\\
  p_z &= -\frac{z-1}{z\cdot(|\mathcal S|-z)\cdot\binom{|\mathcal S|-1}{z-1}}\label{pz}.
  \end{align}

 Together with (\ref{LC}), we can solve each $\alpha_{k,z}^{\mathcal S}$:
  \begin{align}
  \alpha_{k,z}^{\mathcal S} \overset {(a)}{=} &~
  c_z\!\!\!\sum\limits_{\substack{\mathcal S_1\in\mathcal{S}:\\|\mathcal S_1|=z,k\in \mathcal S}_1}\!\!\!\big|\mathcal{V}_{\mathcal{S}_1}^{\mathcal{S}\backslash\mathcal{S}_1}\big| +p_z\sum_{j\in\mathcal{S}\backslash\{ k\}}\!\!\!\sum\limits_{\substack{\mathcal S_1\in\mathcal{S}:\\|\mathcal S_1|=z,j\in \mathcal S}_1}\!\!\!\big|\mathcal{V}_{\mathcal{S}_1}^{\mathcal{S}\backslash\mathcal{S}_1}\big|\notag\\
   \overset {(b)}{=}&~c_z\!\!\sum\limits_{\substack{\mathcal S_1\in\mathcal{S}:\\|\mathcal S_1|=z,k\in \mathcal S}_1}\!\!\!\big|\mathcal{V}_{\mathcal{S}_1}^{\mathcal{S}\backslash\mathcal{S}_1}\big|
   \!+\!p_z\bigg(z\!\!\!\!\!\sum\limits_{\substack{\mathcal S_1\in\mathcal{S}:\\|\mathcal S_1|=z,k\notin \mathcal S}_1}\!\!\!\big|\mathcal{V}_{\mathcal{S}_1}^{\mathcal{S}\backslash\mathcal{S}_1}\big| +(z-1)\sum\limits_{\substack{\mathcal S_1\in\mathcal{S}:\\|\mathcal S_1|=z,k\in \mathcal S_1}}\!\!\!\big|\mathcal{V}_{\mathcal{S}_1}^{\mathcal{S}\backslash\mathcal{S}_1}\big|\bigg)\notag\\
   =&\!\!\!\sum\limits_{\mathcal{\substack{\mathcal S_1\in\mathcal{S}:\\|\mathcal S_1|=z,k\in \mathcal S_1}}}\!\!\!\![c_z\!+\!(z\!-\!1) p_z]\big|\mathcal{V}_{\mathcal{S}_1}^{\mathcal{S}\backslash\mathcal{S}_1}\big| \!+ \!\!\! \sum\limits_{\substack{\mathcal S_1\in\mathcal{S}:\\|\mathcal S_1|=z,k\notin \mathcal S_1}}\!\!\!\!\!z p_z\big|\mathcal{V}_{\mathcal{S}_1}^{\mathcal{S}\backslash\mathcal{S}_1}\big|\notag\\
  \overset {(c)}{=}&\!\!\!\!\!\!\sum\limits_{\substack{\mathcal{S}_1\subseteq \mathcal{S}:\\|\mathcal{S}_1|=z,k\in \mathcal{S}_1}}\!\!\!\!\!\frac{|\mathcal {V}_{\mathcal{S}_1}^{\mathcal {S} \backslash \mathcal{S}_1}|}{{\tbinom{|\mathcal S|-1}{z-1}}}
  -\!\!\!\!\!\!\sum\limits_{\substack{\mathcal S_1\subseteq \mathcal{S}:\\|\mathcal S_1|=z,k\not\in \mathcal S_1}}\!\!\!\!\!\! \frac{(z-1)| \mathcal {V}_{\mathcal S_1}^{\mathcal {S} \backslash \mathcal S_1}|}{(|\mathcal S|-z) \binom{|\mathcal S|-1}{z-1}},\label{aizs}
  \end{align}
  where (a) follows by   taking the $k$-th row of $\mathbf{(A)}^{-1}\mathbf{y}$ and rewriting the notation $\mathcal{S}[i]$ as $\mathcal{S}_1\in\mathcal{S},|\mathcal{S}_1|=z$;  (c) holds by \eqref{cz} and \eqref{pz}; (b) follows from that: Considering a subset $\mathcal S_1\subseteq \mathcal{S}$ with $|\mathcal S_1|=z$, says $\mathcal S_1=\{t_1,\ldots,t_z\}$, if $k\notin \mathcal S_1$, then we have $t_1,\ldots,t_z\neq k$.
  This indicates that in the summation over $j\in\mathcal{S}\backslash\{k\}$, the subset $\mathcal S_1=\{t_1,\ldots,t_z\}$ will appear $z$ times (i.e., $j=t_1,\ldots,t_z$), leading to 
  \begin{IEEEeqnarray}{rCl}
  \sum_{j\in\mathcal{S}\backslash\{ k\}}\sum\limits_{\!\substack{\mathcal S_1\in\mathcal{S}:|\mathcal S_1|=z,\\j\in\mathcal S_1,k\notin \mathcal S_1}}\big|\mathcal{V}_{\mathcal{S}_{z}[i]}^{\mathcal {S}\backslash \mathcal{S}_{z}[i]}\big|=\sum\limits_{\substack{\mathcal S_1\in\mathcal{S}:\\|\mathcal S_1|=z,k\notin \mathcal S_1}}\!\!\!\!\!z|\mathcal{V}_{\mathcal{S}_1}^{\mathcal{S}\backslash\mathcal{S}_1}|\label{notinK}
 \end{IEEEeqnarray}
  If $k\in \mathcal S_1$, then among all $j\in \mathcal{S}\backslash\{k\}$, the subset $\mathcal S_1=\{t_1,\ldots,t_z\}$ contains $z-1$ of them, leading to 
   \begin{IEEEeqnarray}{rCl}
   \!\!\sum_{j\in\mathcal{S}\backslash\{ k\}}\ \sum\limits_{\!\substack{\mathcal S_1\in\mathcal{S}:|\mathcal S_1|=z,\\j\in\mathcal S_1, k\in \mathcal S_1}}\!\!\!\!\!\big|\mathcal{V}_{\mathcal{S}_{z}[i]}^{\mathcal {S}\backslash \mathcal{S}_{z}[i]}\big|=\!\!\!\!\!\!\!\sum\limits_{\substack{\mathcal S_1\in\mathcal{S}:\\|\mathcal S_1|=z,k\in \mathcal S_1}}\!\!\!\!\!(z-1)\big|\mathcal{V}_{\mathcal{S}_1}^{\mathcal{S}\backslash\mathcal{S}_1}\big|\label{inK}.
 \end{IEEEeqnarray}

 Note that if $|\mathcal{V}_{\mathcal {S}_z[i]}^{\mathcal {S} \backslash \mathcal {S}_z[i]}|\neq 0$, for all $i\in\big[\binom{|\mathcal{S}|}{z}\big]$, the constraint (7c) is satisfied. As the deficit ratio of node $k$ is less than the threshold value $\frac{|\mathcal S|-z}{z-1}$, the constant (7b) is also satisfied.

 \section{Proof of Corollary\ref{FSCTcoro}}\label{AppFSCTCoro}
 In this section, we prove that our FSCT scheme is optimal for the homogeneous, semi-homogeneous and 3-node systems.
 %the upper bound given by \eqref{conditionOptThm4} can be reduced to the optimal communication load in \cite[Theorem 2]{CDC}, \cite[Lemma 1]{ThreeWorker}, and \eqref{semiLowBound} respectively.

 \subsection{Homogeneous System}\label{Ca}
For the homogeneous system, we prove that both the feasible and deficit conditions are satisfied in the homogeneous system.

First, consider the deficit condition in \eqref{conditionOptThm4}. For a homogeneous system where each Reduce function is computed $s$ times, in cluster $\mathcal{S}$, we have $\big|\mathcal{V}_{\mathcal {S}_z[i]}^{\mathcal {S} \backslash \mathcal {S}_z[i]}\big|
 =\frac{NQ}{\binom{K}{r}\binom{K}{s}}\binom{r}{|\mathcal{S}|-s}$ for all $i\in[\binom{|\mathcal{S}|}{z}],z=r$, and $\big|\mathcal{V}_{\mathcal {S}_z[i]}^{\mathcal {S} \backslash \mathcal {S}_z[i]}\big|
 =0$ for $z\neq r$. Then according to \eqref{nAndBeta},
\begin{align}
n_{k,z}^{\mathcal{S}}
&=(|\mathcal{S}|-z)\!\!\!\!\!\!\!\!\!\sum\limits_{\!i\in \left[\binom{|\mathcal{S}|}{z}\right], k\in\mathcal{S}_{z}[i]}\!\!\!\!\!\!\!\!\!\big|\mathcal{V}_{\mathcal{S}_{z}[i]}^{\mathcal {S}\backslash \mathcal{S}_{z}[i]}\big|
\!-\! (z\!-\!1)\!\!\!\!\!\!\!\!\!\sum\limits_{\!i\in \left[\binom{|\mathcal{S}|}{z}\right], k\notin\mathcal{S}_{z}[i]}\!\!\!\!\!\!\!\!\!\big|\mathcal{V}_{\mathcal{S}_{z}[i]}^{\mathcal {S}\backslash \mathcal{S}_{z}[i]}\big|=\tbinom{|\mathcal{S}|-1}{r}\frac{NQ}{\binom{K}{r}\binom{K}{s}}\tbinom{r}{|\mathcal{S}|-s}>0, ~\forall k\in\mathcal{S}, z=r,
\end{align}
and $n_{k,z}^{\mathcal{S}}=0$ for $z\neq r$, which means the deficit condition is satisfied according to \eqref{deficitN}.

Then, consider the feasible condition. In the round $z=r$ of cluster $\mathcal{S}$, we can find a solution of \eqref{feasible} for Node $i$ as
\begin{align}
 &\beta_{j,\mathcal S_1}=\frac{(|\mathcal{S}|\!-\!1)}{z}|\mathcal{V}_{\mathcal {S}_1}^{\mathcal {S} \backslash \mathcal {S}_1}|=\frac{(|\mathcal{S}|\!-\!1)}{r}\frac{NQ}{\binom{K}{r}\tbinom{K}{s}}\binom{r}{|\mathcal{S}|-s},\notag\\
 &\forall j\in \mathcal S_1, i\not\in\mathcal S_1, \mathcal S_1\subseteq \mathcal{S},|\mathcal S_1|=z.
\end{align}
The $\beta_{j,\mathcal S_1}$ above is a solution for Node $i\in\mathcal{S}$ since 
\begin{align*}
&\sum_{\substack{j:j\in S_1}}\beta_{j,\mathcal S_1} = (|\mathcal{S}|\!-\!1) |\mathcal {V}_{\mathcal {S}_1}^{\mathcal {S} \backslash \mathcal {S}_1}|,\ \forall \mathcal S_1\subseteq \mathcal{S}\!:\!|\mathcal S_1|=z,i\not\in\mathcal{S}_1,\\
 &\sum_{S_1:j\in \mathcal S_1,i\not\in \mathcal S_1}\beta_{j,\mathcal S_1} = \tbinom{|\mathcal{S}|-2}{z-1}\frac{(|\mathcal{S}|-1)}{z}|\mathcal{V}_{\mathcal {S}_1}^{\mathcal {S} \backslash \mathcal {S}_1}| = \tbinom{|\mathcal{S}|-1}{z}|\mathcal{V}_{\mathcal {S}_1}^{\mathcal {S} \backslash \mathcal {S}_1}|= n_{j,z}^{\mathcal{S}},\forall j\in \mathcal{S},j\neq i,\notag
\end{align*}
which satisfies the feasible condition in \eqref{feasible}.

As both feasible and deficit conditions are satisfied, our FSCT is optimal according to Theorem \ref{ThmFewshot}.

\subsection{3-node System}
Similar to Appendix \ref{OSCT3wProve}, we only calculate the communication load in the round 2 of cluster $S=\{1,2,3\}$ in which we consider file sets $\mathcal{S}_{12},\mathcal{S}_{13}$ and $\mathcal{S}_{23}$. Rewrite $n_{k,z}^{\mathcal{S}}$ as $n_k$ for simplicity.

 \subsubsection{Case1} If $S_{12}+S_{13}-S_{23}\geq 0$, according to \eqref{nAndBeta}, we can get all $n_1,n_2,n_3$ below
 \begin{align}
  n_1 = S_{12} + S_{13} - S_{23}>0\notag\\
  n_2 = S_{12} + S_{23} - S_{13}>0\notag\\
  n_3 = S_{13} + S_{23} - S_{12}>0.\label{beta3w}
 \end{align}
 
 By \eqref{loadFSCTcompu}, Node $i$ will multicast $n_i$ LCs with size of $\frac{1}{2}$ IV, so the communication load can be computed as:
 \begin{align}
 L&=2\left(S_1+S_2+S_3\right)+\frac{1}{2}(S_{12} + S_{13} + S_{23}).\label{case1F}
 \end{align}
 
 \subsubsection{Case2} If $S_{12}+S_{13}-S_{23}<0$, from \eqref{beta3w}, we can get 
 \begin{align}
 &(n_1)^+ = 0\notag\\
 &(n_2)^+ = S_{12} + S_{23} - S_{13} \notag\\
 &(n_3)^+ = S_{13} + S_{23} - S_{12}.
 \end{align}

 Now we show that the feasible condition is satisfied. Take Node 3 as an example, the feasible condition of Node 3 in \eqref{feasible} will be
 \begin{align}
 &\beta_{1,\{1,2\}} \leq 0\notag\\
 &\beta_{1,\{1,2\}} + \beta_{2,\{1,2\}} \geq 2S_{12}\notag\\
 &\beta_{2,\{1,2\}} \leq S_{12} + S_{23} - S_{13}.\label{newsol}
 \end{align}
 We can find that $\beta_{1,\{1,2\}}=0$ and $2S_{12}\leq\beta_{2,\{1,2\}}\leq S_{12} + S_{23} - S_{13}$ is a solution to \eqref{newsol}. Since $S_{12} + S_{23} - S_{13}>2S_{12}$, the feasible condition is satisfied.
 Thus, by \eqref{loadFSCTcompu}, the communication load can be computed as
 \begin{align}
 L&=2\left(S_1+S_2+S_3\right)+(n_{1})^++(n_{2})^++(n_3)^+\notag\\
 &= 2\left(S_1+S_2+S_3\right)+S_{23}.\label{case2F}
 \end{align}
 
 In \cite{ThreeWorker}, the authors proved that the communication load \eqref{case1F} and \eqref{case2F} are coherent with \eqref{gFunc} in Lemma 2, i.e., our communication load in Theorem 3 is tight.

\subsection{Semi-Homogeneous System}
For a semi-homogeneous system, according to the definition of $\mathcal{V}_{\mathcal {S}_z[i]}^{\mathcal {S} \backslash \mathcal {S}_z[i]}$ \eqref{zIVs},
\begin{align}
 |\mathcal{V}_{\mathcal {S}_z[i]}^{\mathcal {S} \backslash \mathcal {S}_z[i]}|
 =\sum\limits_{s=|\mathcal S|-r}^{|\mathcal S|}\frac{\binom{r}{|\mathcal S|-s}}{\binom{K}{r}\binom{K}{s}}Q_sN,~\forall i\in[\tbinom{|\mathcal{S}|}{z}],z=r,
\end{align}
 and $|\mathcal{V}_{\mathcal {S}_z[i]}^{\mathcal {S} \backslash \mathcal {S}_z[i]}|=0$ for $z\neq r$. Then according to the definition of $\bar{n}_{k,z}^{\mathcal{S}}$ \eqref{nAndBeta}, we obtain that for all $ k\in\mathcal{S}$ and $z=r$, 
\begin{align}
 n_{k,z}^{\mathcal{S}}
 &=(|\mathcal{S}|-z)\!\!\!\!\!\!\!\!\!\sum\limits_{\!i\in \left[\binom{|\mathcal{S}|}{z}\right], k\in\mathcal{S}_{z}[i]}\!\!\!\!\!\!\!\!\!\Big|\mathcal{V}_{\mathcal{S}_{z}[i]}^{\mathcal {S}\backslash \mathcal{S}_{z}[i]}\Big|
 \!-\! (z\!-\!1)\!\!\!\!\!\!\!\!\!\sum\limits_{\!i\in \left[\binom{|\mathcal{S}|}{z}\right], k\notin\mathcal{S}_{z}[i]}\!\!\!\!\!\!\!\!\!\Big|\mathcal{V}_{\mathcal{S}_{z}[i]}^{\mathcal {S}\backslash \mathcal{S}_{z}[i]}\Big|=\tbinom{|\mathcal{S}|-1}{r}\!\!\sum\limits_{s=|\mathcal S|-r}^{|\mathcal S|}\frac{\binom{r}{|\mathcal S|-s}}{\binom{K}{r}\binom{K}{s}}Q_sN>0, \notag
 \end{align} 
which means the deficit condition is satisfied.

Then, consider the feasible condition $\mathscr{E}_2(i,\mathcal{S},z)$. For each Node $i\in \mathcal{S}$ in round $z=r$, similar to Appendix \ref{Ca}, we can find a solution of \eqref{feasible} for Node $i$ below
\begin{align}
 \beta_{j,\mathcal S_1}&=\frac{(|\mathcal{S}|\!-\!1)}{z}|\mathcal{V}_{\mathcal {S}_1}^{\mathcal {S} \backslash \mathcal {S}_1}|=\frac{N(|\mathcal{S}|-1)}{r}\sum\limits_{s=|\mathcal S|-r}^{|\mathcal S|}\frac{\binom{r}{|\mathcal S|-s}}{\binom{K}{r}\binom{K}{s}}Q_s, \notag\\
 &\forall j\in \mathcal S_1, i\not\in \mathcal{S}_1, \mathcal S_1\subseteq \mathcal{S},|\mathcal S_1|=z.
\end{align}
 The aforementioned $\beta_{j,\mathcal{S}_1}$ is a solution since
 \begin{align*}
 &\sum_{\substack{j:j\in S_1}}\beta_{j,\mathcal S_1} = (|\mathcal{S}|\!-\!1) |\mathcal {V}_{\mathcal {S}_1}^{\mathcal {S} \backslash \mathcal {S}_1}|,\ \forall S_1\subseteq \mathcal{S}\!:\!|S_1|=z,i\not\in\mathcal{S}_1\notag\\
 &\sum_{\mathcal S_1:j\in \mathcal S_1,i\not\in \mathcal S_1}\beta_{j,\mathcal S_1} = \tbinom{|\mathcal{S}|-2}{z-1}\frac{(|\mathcal{S}|-1)}{z}|\mathcal{V}_{\mathcal {S}_1}^{\mathcal {S} \backslash \mathcal {S}_1}|= \tbinom{|\mathcal{S}|-1}{z}|\mathcal{V}_{\mathcal {S}_1}^{\mathcal {S} \backslash \mathcal {S}_1}|= n_{j,z}^{\mathcal{S}},\ \ \ \forall j\in \mathcal{S},j\neq i.
 \end{align*}
As both feasible and deficit conditions are satisfied, our FSCT is optimal according to Theorem \ref{ThmFewshot}.

\section{Correctness of FSCT}\label{CorrectFSCT}

\begin{figure*}
  \centering
  \includegraphics[scale = 0.58]{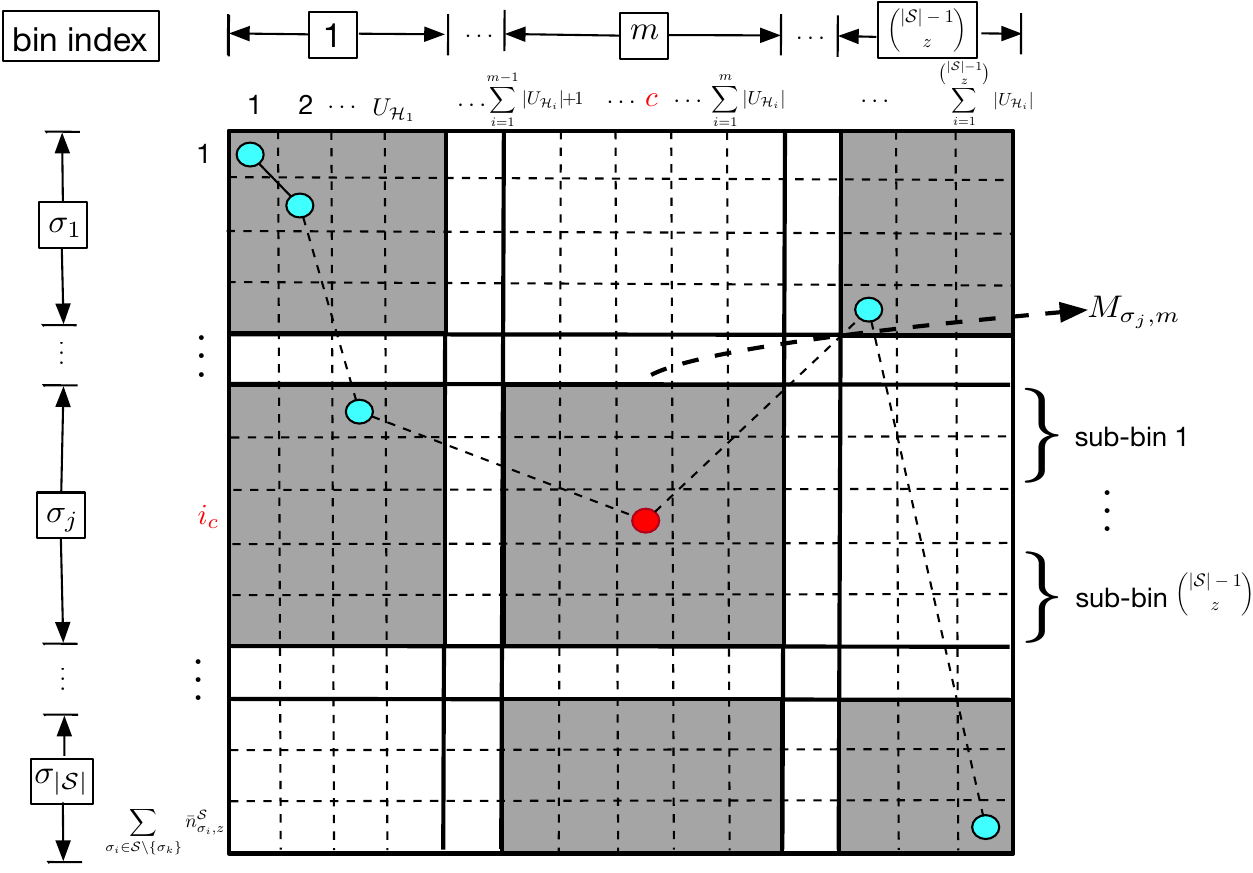}
  \center \caption{
   Structure of decoding matrix $\mathbf{M}$, which contains $R$ rows and $C$ columns. Each dashed box represents an element of the matrix $\mathbf{M}$. $\mathbf{M}$ is divided into submatrices $M_{\sigma_i,m}$ with $\sigma_i\in\mathcal{S}\backslash\{\sigma_k\}$ and $m\in[\tbinom{|\mathcal{S}|-1}{z}]$. The large grey rectangles are submatrices $M_{\sigma_i,c}$ with $\sigma_i\in\mathcal{H}_m$ whose elements are all randomly generated and the large white rectangles are submatrices with $\sigma_i\not\in\mathcal{H}_m$ whose elements are all zero. The numbers in the boxes are row and column bin indices. The numbers along with the row and column indicate the row index and column index of the matrix respectively. For example, the submatrix $M_{\sigma_j,m}$ has $\bar{n}_{\sigma_j,z}^{\mathcal{S}}$ rows and $|\mathcal{V}_{\mathcal{H}_m}^{\mathcal{S}\backslash \mathcal{H}_m}|$ columns. The right-side numbers are the indices of sub-bins partitioned based on the solution of the feasible conditions: \{$\beta_{\sigma_i,\mathcal T}\in \mathbb{Q}^+:\sigma_i\in \mathcal T, \sigma_k\not\in \mathcal{T}, \mathcal T\subseteq \mathcal{S},|\mathcal T|=z$\} and $\beta_{\sigma_i,\mathcal{T}}=0$ if $\sigma_i\not\in \mathcal T$. For example, the $\sigma_j$-th row bin of the matrix is partitioned into $\binom{|\mathcal{S}|-1}{z}$ sub-bins where the $l$-th part for $l\in[\binom{|\mathcal{S}|-1}{z}]$ contains $\beta_{\sigma_j,\mathcal{H}_l}$ rows. The cyan points (and the red point) in the figure are the elements we choose for the non-zero path $(m_{i_1,1},m_{i_2,2},\cdots,m_{i_C,C}):m_{i_c,c}\neq 0$ for $c\in[C]$. For example, we can choose the element in the $p$-th row of the $m$-th sub-bin in the $\sigma_j$-th row bin for the $c$-th column, where $m$ is the index that $\sum_{i=1}^{m-1}{U}_{\mathcal{H}_i}+1\leq c\leq \sum_{i=1}^{m}{U}_{\mathcal{H}_i}$, $\sigma_j=k_t\in\mathcal{H}=\{k_1,\cdots,k_z\}$, which satisfy
   $\!\!\!\!\!\!\!\sum\limits_{\sigma_i\in{\{k_1,\cdots,k_{t-1}\}}}\!\!\!\!\!\!\!\!\!\!\!\beta_{\sigma_i,\mathcal{H}_m}\!+\!1\leq c\!-\!\sum_{i=1}^{m-1}\!\!{U}_{\mathcal{H}_i}\leq \!\!\!\!\!\!\!\!\!\sum\limits_{\sigma_i\in{\{k_1,\cdots,k_{t}\}}}\!\!\!\!\!\!\!\!\beta_{\sigma_i,\mathcal{H}_m}$, and $p=c-\sum_{i=1}^{m-1}|\mathcal{U}_{\mathcal{H}_j}^{\mathcal{S}\backslash\mathcal{H}_j}|-\!\!\!\!\!\!\sum\limits_{\sigma_i\in{\{k_1,\cdots,k_{t-1}\}}}\!\!\!\!\!\!\!\!\!\!\!\beta_{\sigma_i,\mathcal{H}_m}$, which is the red point in the figure.
  }\label{matrix}
  \end{figure*}

 For cluster $\mathcal{S}=\{\sigma_1,\sigma_2,\cdots, \sigma_{|\mathcal{S}|}\}$, we prove that each receiver $\sigma_k$ in the round $z$ can successfully decode its desired IV segments by jointly solving the LCs sent from other nodes in $\mathcal{S}\backslash \{\sigma_k\}$.
 
 We index the $\binom{|\mathcal{S}|-1}{z}$ subsets of $\mathcal{S}$ of size $z$ which do not contain Node $\sigma_k$ (i.e., $\{\mathcal{T}:\mathcal{T}\subseteq\mathcal{S},|\mathcal{T}|=z,\sigma_k\not\in \mathcal{T}\}$) as $\{\mathcal{H}_1,\cdots,\mathcal{H}_{\binom{|\mathcal{S}|-1}{z}}\}$. Note that FSCT cuts each IV into $|\mathcal{S}|-1$ segments in cluster $\mathcal{S}$. For convenience, we use ${U}_{\mathcal{H}_i}$ to denote the number of \emph{segments} generated from IVs in $\mathcal{V}_{\mathcal{H}_i}^{\mathcal{S}\backslash\mathcal{H}_i}$. Recall that after receiving the message blocks from other nodes in $\mathcal{S}$, i.e., $\{X_{k,z}^{\mathcal{ S}}[{i}]:k\in\mathcal{S}\backslash \{\sigma_k\}, i\in[\bar{n}_{k,z}^{\mathcal{S}}]\}$, receiver $\sigma_k\in\mathcal{S}$ first removes the local segments $\big\{{U}_{\mathcal{T}}^{\mathcal {S} \backslash \mathcal{T}}:\mathcal{T}\subseteq\mathcal{S},|\mathcal{T}|=z,\sigma_k\in\mathcal{T} \big\}$ from the messages to generate $R\triangleq\sum_{k\in\mathcal{S}\backslash \{\sigma_k\}}\bar{n}_{k,z}^{\mathcal{S}}$ numbers of LCs consisting of $C\triangleq \sum_{i=1}^{\binom{|\mathcal{S}|-1}{z}}{U}_{\mathcal{H}_i}$ unknown variables. 
  Let matrix $\mathbf{M}\in \mathbb{R}_{R\times C}$ denote the coefficient matrix of received LCs after removing local segments, which is shown in Fig. \ref{matrix}. 
 
 The $R$ rows of the $\mathbf{M}$ will be divided into $|\mathcal{S}|-1$ \emph{row bin}, with the {row bin} indices $\sigma_1,\ldots,\sigma_{k-1},\sigma_{k+1}\ldots, \sigma_{|\mathcal{S}|}$. Each row bin $\sigma_i$ contains $n_{\sigma_i,z}^{\mathcal{S}}$ rows, which represent the number of LCs sent by Node $\sigma_i$ in round $z$ of cluster $\mathcal{S}$. Each row bin $\sigma_i$, for $i\in\mathcal{S}\backslash\{k\}$, is further partitioned into $\binom{|\mathcal{S}|-1}{z}$ sub-bins, indexed as $\{1,2,\cdots, \binom{|\mathcal{S}|-1}{z}\}$, with the $l$-th sub-bin ($l\in[\tbinom{|\mathcal{S}|-1}{z}]$) containing $\beta_{\sigma_i,\mathcal H_l}$ rows, where $\beta_{\sigma_i,\mathcal H_l}$ is the non-negative solution of \eqref{feasible}: \{$\beta_{\sigma_i,\mathcal T}\in \mathbb{Q}^+:\sigma_i\in \mathcal T, \sigma_k\not\in \mathcal{T}, \mathcal T\subseteq \mathcal{S},|\mathcal T|=z$\}, and let $\beta_{\sigma_i,\mathcal T}=0$ if $\sigma_i\not\in \mathcal{T}$. 
  Note that from the second constraint $\sum_{\mathcal T:\sigma_i\in \mathcal T,\sigma_k\not\in \mathcal T}\beta_{\sigma_i,\mathcal T} \leq \bar{n}_{\sigma_i,z}^{\mathcal{S}}$ in \eqref{feasible}, we can guarantee that there are enough rows to do partition, and the residual rows after the partition can be brought into the last sub-bin.
 
  The $C$ columns of the $\mathbf{M}$ can be divided into $\binom{|\mathcal{S}|-1}{z}$ \emph{column bin}, and the $m$-th {column bin} contains $|U_{\mathcal{H}_m}|$ columns which represent the number of segments in $\mathcal{V}_{\mathcal{H}_m}^{\mathcal{S}\backslash\mathcal{H}_m}$.
 
  Now we split the matrix $\mathbf{M}$ into $\binom{|\mathcal{S}|-1}{z}\times (|\mathcal{S}|-1)$ submatrices $\{M_{\sigma_i,m}:\sigma_i\in\mathcal{S}\backslash\{\sigma_k\},m\in[\binom{|\mathcal{S}|-1}{z}]\}$ based on the row and column bins as shown in Fig. \ref{matrix}. Submatrix $M_{\sigma_i,m}$ has $\bar{n}_{\sigma_i,z}^{\mathcal{S}}$ rows and $|U_{\mathcal{H}_m}|$ columns. If $\sigma_i\in\mathcal{H}_m$, then the elements of submatrix $M_{\sigma_i,m}$ are the coefficients of the LCs for sending the segments in $\mathcal{V}_{\mathcal{H}_m}^{\mathcal{S}\backslash\mathcal{H}_m}$ by Node $\sigma_i$ (highlight as the grey rectangles in Fig. \ref{matrix}). If $\sigma_i\not\in\mathcal{H}_m$, since Node $\sigma_i$ did not encode these segments, the elements of $M_{\sigma_i,m}$ are all zero (highlight as the white rectangles in Fig. \ref{matrix}). 
 Note that the transmission strategy of FSCT guarantees that the feasible condition is satisfied in each round $z$ of cluster $\mathcal{S}$ as stated in Section \ref{AlwaysSatisfy}. 
  
 Now we need to prove that the linear system is solvable. The main idea is that we first find a non-zero path $(m_{i_1,1},m_{i_2,2},\cdots,$ $m_{i_C,C})$ where $m_{i_c,c}\neq 0$ for all $c\in[C]$ and $m_{i_c,c}$ is the element in row $i_c$ and column $c$ of $\mathbf{M}$ with all $i_c$ different, then select $C$ rows from $\mathbf{M}$ to form a square matrix $(\mathbf{M})^{'}$ with non-zero diagonal elements, and finally we view the determinant of $(\mathbf{M})^{'}$ as a non-zero polynomial of $(m_{i_1,1},m_{i_2,2},\cdots, m_{i_C,C})$ \cite{poly} so that $(\mathbf{M})^{'}$ is invertible with high probability by Schwartz-Zippel Lemma \cite{SZL}.
  
 Suppose the column $c\in[C]$ is in the $m$-th column bin. Then, it must satisfy
 %For each column $c$ in the $m$-th column bin, i.e.,
 \begin{IEEEeqnarray}{rCl}
  \sum_{i=1}^{m-1}{U}_{\mathcal{H}_i}+1\leq c\leq \sum_{i=1}^{m}{U}_{\mathcal{H}_i}.\nonumber
  \end{IEEEeqnarray}
  For this $m$, we can find $\mathcal{H}_m$ and represent its indices as $\mathcal{H}_m=\{k_1,\cdots,k_z\}$. Now we choose the row bin 
  \begin{subequations}
  \begin{IEEEeqnarray}{rCl}
  \sigma_j=k_t\in\mathcal{H}_m
 \end{IEEEeqnarray}
 such that 
 \begin{IEEEeqnarray}{rCl}
 \sum_{\sigma_i\in{\{k_1,\cdots,k_{t-1}\}}}\!\!\!\!\!\!\!\!\!\!\!\beta_{\sigma_i,\mathcal{H}_m}\!+\!1\leq c\!-\!\sum_{i=1}^{m-1}\!\!{U}_{\mathcal{H}_i}\leq \!\!\!\!\!\!\!\!\!\sum_{\sigma_i\in{\{k_1,\cdots,k_{t}\}}}\!\!\!\!\!\!\!\!\beta_{\sigma_i,\mathcal{H}_m}.~~\label{Findt}
 \end{IEEEeqnarray}
  \end{subequations}
 
 Note that by the first constraint in \eqref{feasible}, there always exists such $k_t\in\{k_1,\cdots,k_z\}$. Then, we choose the $m$-th sub-bin of the row bin $\sigma_j$, which contains $\beta_{\sigma_j,\mathcal{H}_m}$ rows. Finally, we choose the $p$-th row of the $m$-th sub-bin, where 
 \begin{align}
  p=c-\sum_{i=1}^{m-1}|\mathcal{U}_{\mathcal{H}_j}^{\mathcal{S}\backslash\mathcal{H}_j}|-\!\!\!\!\!\!\sum_{\sigma_i\in{\{k_1,\cdots,k_{t-1}\}}}\!\!\!\!\!\!\!\!\!\!\!\beta_{\sigma_i,\mathcal{H}_m}\label{Findp}
 \end{align}
 Thus, for each element $m_{i_c,c}, c\in[C]$ in the path, we choose $i_c$ as the $p$-th row of the $m$-th sub-bin of the row bin $\sigma_j$, i.e.,
 \begin{align}
  i_c=\sum\limits_{\sigma_{i}\in\{\sigma_1,\cdots,\sigma_j\}}\bar{n}_{\sigma_i,z}^{\mathcal{S}}+\sum\limits_{i=1}^{m-1}\beta_{\sigma_j,\mathcal{H}_i} + p,\label{59}
 \end{align}
 where $m,\sigma_j,p$ are defined above. 
 
 The choice above will lead to a unique row $i_c$ for each column $c$, i.e., $i_a\neq i_b$ if $a\neq b$ for all $a,b\in[C]$. If $i_a= i_b$, $m,t,p$ in \eqref{59} for $i_a,i_b$ will be the same. As $a=b=p+\sum_{i=1}^{m-1}{U}_{\mathcal{H}_i}+\!\!\!\!\!\!\!\!\sum\limits_{\sigma_i\in{\{k_1,\cdots,k_{t-1}\}}}\!\!\!\!\!\!\!\!\!\!\!\beta_{\sigma_i,\mathcal{H}_m}$, we can imply that $i_a\neq i_b$ for $a\neq b$.
 
 Each element $m_{i_c,c}$ is valid and randomly generated. From \eqref{Findt} and \eqref{Findp}, we have $1\leq p\leq \beta_{\sigma_j,\mathcal{H}_m}$. Hence, 
 \begin{align}
  i_c\geq\sum\limits_{\sigma_{i}\in\{\sigma_1,\cdots,\sigma_j\}}\bar{n}_{\sigma_i,z}^{\mathcal{S}}+\sum\limits_{i=1}^m\beta_{\sigma_j,\mathcal{H}_i}+1,
 \end{align}
 and
 \begin{align}
  i_c\leq\sum\limits_{\sigma_{i}\in\{\sigma_1,\cdots,\sigma_j\}}\bar{n}_{\sigma_i,z}^{\mathcal{S}}+\sum\limits_{i=1}^m\beta_{\sigma_j,\mathcal{H}_i}.
 \end{align}
 Thus, $m_{i_c,c}$ is in the submatrix $M_{\sigma_j,m}$ with $\sigma_j\in \mathcal \mathcal{H}_m$, which implies $m_{i_c,c}$ is randomly generated for each $c$.
 
 After we find the non-zero path $(m_{i_1,1},m_{i_2,2},\cdots, m_{i_C,C}):m_{i_c,c}\neq 0$, we select $C$ rows from $\mathbf{M}$ to form a square matrix $(\mathbf{M})^{'}$ with no non-zero elements along the diagonal:
 \begin{itemize}
  \item The $c$-th row of $(\mathbf{M})^{'}$ is the $(i_{c})$-th row of $\mathbf{M}$.
  \item The diagonal elements of $(\mathbf{M})^{'}$: $m_{c,c}^{'}=m_{i_{c},c}$ is randomly generated.
 \end{itemize}
 
 View the determinant of $(\mathbf{M})^{'}$ as a non-zero functions of $(m^{'}_{1,1},m^{'}_{2,2},\cdots, m^{'}_{C,C})$. By the Schwartz-Zippel Lemma, the probability of $(\mathbf{M})^{'}$ being invertible is high as the field size is sufficiently large.

 \begin{table*}
  \caption{List of notations in Appendix \ref{CorrectFSCT}}
  \centering\label{notationCorrFSCT}
  \begin{tabular}{|c||c|} \hline
   \textbf{Notation} & \textbf{Description} \\\hline
   $\{\mathcal{H}_1,\cdots,\mathcal{H}_{\binom{|\mathcal{S}|-1}{z}}\}$ & Subsets of $\mathcal{S}$ of size $z$ not containing $\sigma_k$, i.e., $\{\mathcal{T}:\mathcal{T}\subseteq\mathcal{S},|\mathcal{T}|=z,\sigma_k\not\in \mathcal{T}\}$ \\\hline
   ${U}_{\mathcal{H}_i}$ & Number of \emph{segments} generated from IVs in $\mathcal{V}_{\mathcal{H}_i}^{\mathcal{S}\backslash\mathcal{H}_i}$, i.e., ${U}_{\mathcal{H}_i}=(|\mathcal{S}|-1)\big|\mathcal{V}_{\mathcal{H}_i}^{\mathcal{S}\backslash\mathcal{H}_i}\big|$\\\hline
   $R$ & Total number of LCs of $(z,\mathcal{S})$-mapped IVs received by Node $\sigma_k$ , i.e., $R=\sum\limits_{k\in\mathcal{S}\backslash \{\sigma_k\}}\bar{n}_{k,z}^{\mathcal{S}}$ \\\hline
   $C$ & Number of unknown segments related to $(z,\mathcal{S})$-mapped IVs, i.e., $C=\sum\limits_{i=1}^{\binom{|\mathcal{S}|-1}{z}}{U}_{\mathcal{H}_i}$ \\\hline
   \tabincell{c}{$\beta_{\sigma_i,\mathcal H_l}$\\ for $\sigma_i\in\mathcal{S}\backslash \sigma_k$, $l\!\in\![\tbinom{|\mathcal{S}|-1}{z}]$} & Non-negative solution of feasible condition \eqref{feasible}\\ \hline
   \tabincell{c}{Row bin $\sigma_i$\\ for $\sigma_i\in\mathcal{S}\backslash \sigma_k$} & \tabincell{c}{Row $\sum\limits_{\sigma_j\in\mathcal{S}\backslash \sigma_k ,j< i}\bar{n}_{j,z}^{\mathcal{S}}+1$ to Row $\sum\limits_{\sigma_j\in\mathcal{S}\backslash \sigma_k,j\leq i}\bar{n}_{j,z}^{\mathcal{S}}$} \\\hline
   \tabincell{c}{Sub-bin $l$ of row bin $\sigma_i$\\ for $i\in[|\mathcal{S}|]\backslash k$, $l\in[\binom{|\mathcal{S}|-1}{z}]$} & \tabincell{c}{Row $\sum\limits_{\sigma_j\in\mathcal{S}\backslash \sigma_k ,j< i}\!\!\!\!\bar{n}_{j,z}^{\mathcal{S}}\!\!+\!\sum\limits_{n=1}^{l}\beta_{\sigma_i,\mathcal{H}_n}\!\!+\!1$ to Row $\sum\limits_{\sigma_j\in\mathcal{S}\backslash \sigma_k ,j< i}\!\!\!\!\bar{n}_{j,z}^{\mathcal{S}}\!\!+\!\sum\limits_{n=1}^{l+1}\beta_{\sigma_i,\mathcal{H}_n}\!\!$} \\\hline
   \tabincell{c}{Column bin $m$\\ for $m\in[\binom{|\mathcal{S}|-1}{z}]$} & \tabincell{c}{Column $\sum\limits_{i=1}^{m-1}{U}_{\mathcal{H}_i}+1$ to column $\sum\limits_{i=1}^{m}{U}_{\mathcal{H}_i}$} \\\hline
  \end{tabular}
 \end{table*}

\section*{Acknowledgment}
The authors would like to thank Prof. Manolis Tsakiris and Dr. Zhenchao Ge for insightful discussions.

\ifCLASSOPTIONcaptionsoff
  \newpage
\fi



%\bibliographystyle{IEEEtran}
\bibliography{refer}
\bibliographystyle{IEEEtran}
\end{document}